\documentclass[
  a4paper,
  oneside,
  headsepline,
  final,
  report,
  unicode,
  12pt,
  chapterprefix = true,
  appendixprefix = true
]{scrbook}

\usepackage{sty/thesis}

\usepackage[T1]{CJKutf8} 

\usepackage{ascmac}
\usepackage{graphicx}
\usepackage{moreverb}
\usepackage{longtable}
\usepackage{algorithmic}
\usepackage{listings}

\usepackage{wrapfig}

\usepackage{enumitem}  

\usepackage{amsmath,amssymb,mathtools} 
\allowdisplaybreaks

\usepackage{stmaryrd}
\usepackage[dvipsnames]{xcolor}
\usepackage{tikz}
\usetikzlibrary{intersections, calc}
\usepackage{cprotect}

\usepackage{mdframed}

\usepackage{framed} 

\usepackage{geometry}
\newcommand{\doubleplus}{\mathbin{+\mkern-10mu+}}

\definecolor{myblue}{RGB}{50, 160, 255}
\definecolor{codegreen}{rgb}{0,0.6,0}
\definecolor{codegray}{rgb}{0.5,0.5,0.5}
\definecolor{codepurple}{rgb}{0.58,0,0.82}
\definecolor{backcolour}{rgb}{0.95,0.95,0.92}

\lstdefinelanguage{lmntal}{
  morekeywords = [1]{int, float, unary, ground, string}, 
  morekeywords = [2]{:-, |, .}, 
  morecomment  = [l]{\%},
  morecomment  = [s]{/*}{*/},
  morestring   = *[d]{"},
  alsoletter   = {:-.|'},
  alsodigit    = {:-.|},
  sensitive    = true
}

\definecolor{backcolour}{rgb}{0.95,0.95,0.92}
\definecolor{codegreen}{rgb}{0,0.6,0}
\definecolor{codelightgreen}{rgb}{0.3, 0.6, 0.3}
\definecolor{codegray}{rgb}{0.5,0.5,0.5}
\definecolor{codepurple}{rgb}{0.58,0,0.82}

\newcommand{\fn}[1]{$\mathit{fn}(#1)$}

\newenvironment{condition}[1][]{%
 \begin{oframed}%
    \noindent \textbf{#1:}%
    \rmfamily%
}{%
  \end{oframed}%
}

\renewcommand{\arraystretch}{1.3}

\definecolor{link}{HTML}{800006}
\definecolor{cite}{HTML}{2E7E2A}
\definecolor{file}{HTML}{131877}
\definecolor{url} {HTML}{8A0087}
\definecolor{menu}{HTML}{727500}
\definecolor{run} {HTML}{137776}
\usepackage{url}
\usepackage{hyperref}
\hypersetup{
    pdfauthor={Jin Sano, },
    pdftitle={Implementing G-Machine in HyperLMNtal},
    pdfsubject={Graduation Thesis, Computer Science, Waseda University},
    linkcolor = link, citecolor = cite, filecolor = file, urlcolor = url, menucolor = menu, runcolor  = run,
    linkbordercolor = link!40!white, citebordercolor = cite!40!white, filebordercolor = file!40!white, urlbordercolor = url!40!white, menubordercolor = menu!40!white, runbordercolor  = run!40!white,
    allcolors = MidnightBlue,
    colorlinks, linktocpage,
}

\RedeclareSectionCommand[tocnumwidth=6.5em]{chapter}
\let\originaladdchaptertocentry\addchaptertocentry
\renewcommand*{\addchaptertocentry}[2]{%
  \IfArgIsEmpty{#1}{
    \originaladdchaptertocentry{#1}{#2}%
  }{
    \originaladdchaptertocentry{\chapapp~#1}{#2}%
  }%
}

\usepackage[nameinlink]{cleveref}
\usepackage[most]{tcolorbox}

\newtcbtheorem[auto counter,number within=chapter,Crefname={Example}{Examples}]{example}%
  {Example}{breakable,
  fonttitle = \bfseries\upshape,
  colback = green!5!white,
  colframe = green!75!black
  }{ex}

\newtcbtheorem[auto counter,number within=chapter,Crefname={Definition}{Definitions}]{definition}%
  {Definition}{breakable,
  fonttitle = \bfseries\upshape,
  colback = blue!5!white,
  colframe = blue!75!black
  }{def}

\newtcbtheorem[auto counter,number within=chapter,Crefname={Theorem}{Theorems}]{theorem}%
  {Theorem}{breakable,
  fonttitle = \bfseries\upshape,
  colback = gray!25!white,
  colframe = gray!35!black
  }{th}

\newtcbtheorem[auto counter,number within=chapter,Crefname={Proof}{Proofs}]{proof}%
  {Proof}{breakable,
  fonttitle =\bfseries\upshape,
  colback   = gray!10!white,
  colframe  = gray!75!black
  }{proof}

\newtcbtheorem[auto counter,number within=chapter,Crefname={Lemma}{Lemmas}]{lemma}%
  {Lemma}{breakable,
  fonttitle =\bfseries\upshape,
  colback = gray!25!white,
  colframe = gray!35!black
  }{lem}

\makeatletter
\newcommand\tcb@cnt@evidenceautorefname{Evidence}
\newcommand\tcb@cnt@exampleautorefname{Example}
\newcommand\tcb@cnt@definitionautorefname{Definition}
\newcommand\tcb@cnt@theoremautorefname{Theorem}
\newcommand\tcb@cnt@proofautorefname{Proof}
\newcommand\tcb@cnt@lemmaautorefname{Lemma}
\makeatother

\setlist[itemize]{leftmargin = 6mm}

\newcommand{\equivby}[1]{\equiv_{\mbox{\scriptsize #1}}}
\newcommand{\mydagger}{{\color{Magenta} \(\dagger\)}}
\newcommand{\vb}[1]{{\textnormal{\texttt{#1}}}}%

\begin{document}

\begin{titlepage}
  \newgeometry{left=1truecm, right=1truecm, bottom=1truecm}
  \begin{center}
    {\Large Graduation Thesis\\[2.5truecm]}
    {\Huge
      Implementing G-Machine
      in HyperLMNtal}\\[0.5truecm]
    \begin{CJK}{UTF8}{ipxm}
      {\Large HyperLMNtal を用いた G-Machine の実装}
      \end{CJK}
    \\[6truecm]
      {\Large
        \begin{tabular}{ll}
          Submition Date:& January 22, 2021 \\
          Supervisor:&  Prof. Dr. Kazunori Ueda \\
        \end{tabular}\\[3truecm]
        Waseda University\\
        Computer Science and Engennering\\[0.8truecm]
        \begin{tabular}{ll}        
          Student ID:& 1W172146-1\\
        \end{tabular}\\[0.2truecm]
        Jin Sano 
      }
  \end{center}
  \restoregeometry
\end{titlepage}

\frontmatter

\vspace*{1.5truecm}
\section*{Abstract}
\thispagestyle{empty}

Since language processing systems generally allocate/discard memory with complex reference relationships, including circular and indirect references, their implementation is often not trivial.
Here, the allocated memory and the references can be abstracted to the labeled vertices and edges of a graph.
And there exists a \emph{graph rewriting language}, a programming language or a calculation model that can handle graph intuitively, safely and efficiently.
Therefore, the implementation of a language processing system can be highly expected as an application field of graph rewriting language.
To show this, in this research, we implemented \emph{G-machine}, the virtual machine for lazy evaluation, in hypergraph rewriting language, \emph{HyperLMNtal}.

HyperLMNtal is extended from the graph rewriting language/calculus model \emph{LMNtal}.
However, it lacked the rigid definition: it was more an extension of the implementation than on the calculus model.
The semantics of LMNtal features fine-grained concurrency based on local rewriting.
However, since we could not determine the locality of the hyperlink in HyperLMNtal, we couldn't incorporate it into the LMNtal semantics.
Thus, we first introduced a scope (\emph{link creation}) and defined the locality of a hyperlink to formalize the syntax/semantics.
Now, HyperLMNtal is not just a programming language extended from the basic calculation model, but also a concurrent calculation model based on strict and formal definitions.

G-machine is a virtual machine that performs lazy evaluation, which is the basis of implementation of lazy functional programming languages such as Haskell.
The implementation of G-machine requires a heap, which is a more general graph than just a tree so as can share subgraphs.
Therefore, HyperLMNtal is ideal to implement this.
In our research, we implemented a compiler which translates the source language, the core language, into the execution code for G-machine and G-machine, in HyperLMNtal.
We have succeeded to implement the compiler in 404 lines and G-machine in 570 lines and showed that we can implement the language processing system that handles complex data structures in graph rewriting language tersely.
In addition, we achieved to visualize G-machine using the HyperLMNtal visualizer.

\vspace*{1.5truecm}
\begin{CJK}{UTF8}{ipxm}
\section*{概要}
\thispagestyle{empty}

言語処理系は一般的に，動的なメモリの確保及び破棄を行い，循環参照や間接参照を含む複雑な参照関係にあるデータを扱うため，その実装は自明ではないことが多い.
ここで，メモリ領域及びその参照はグラフ理論におけるグラフのラベル（データ）付き頂点と辺に抽象化できる.
また，グラフを直感的に安全に効率良く扱うことができるプログラミング言語または計算モデルとして，グラフ書き換え言語がある.
従って言語処理系の実装は，グラフ書き換え言語の応用分野として高い期待が持てるのではないかと考えた.
そこで，本研究では，ハイパーグラフ書き換え言語；HyperLMNtal を用いて，遅延評価を行う仮想機械；G-machine の実装を行った.

HyperLMNtal はグラフ書き換え言語であり計算モデルであるLMNtalの拡張である．
ただし，これは実装上の拡張であり，厳密な意味論が存在していなかった.
LMNtal の意味論上の特徴は局所的な書き換えに基づく細粒度の並行性だが,
HyperLMNtal のハイパーリンクは
局所性が判別できないため，そのままではLMNtal の意味論に自然に組み込むことはできなかった.
そこで，まず，スコープ(link creation) を導入し，
ハイパーリンクに局所性を定義することで，意味論を厳密に定めた.
これによって，HyperLMNtal は，基本計算モデルから拡張された単なるプログラミング言語であるだけではなく，それ自体が，厳密で形式的な定義に基づいた並行計算モデルとなった．

G-machine は Haskell などのプログラミング言語の処理系の実装の基盤となっている遅延評価を実行する仮想機械である.
G-machine は，サブグラフの共有が行えるような，木に限らない，より一般のグラフであるヒープを必要とする.
そこで， HyperLMNtal を用いて，G-machine，及び関数型言語から G-machine の実行用コードを得るコンパイラを作成した.
結果として，G-machine は 404 行，コンパイラは 570 行で実装でき，グラフ書き換え言語で複雑なデータ構造を扱う言語処理系の実装が非常に簡潔にできることがわかった.
また，HyperLMNtal のビジュアル化機能を用いて，G-machine の挙動を可視化できた.

\end{CJK}

\tableofcontents 
\listoffigures   
\newpage

\mainmatter

\chapter{Introduction}
\label{cha:intro}
\par

\section{Background of the research}
\subsection{Implementation of language processing system}
Language processing systems generally allocate/discard memory with complex reference relationships, including circular and indirect references.
Thus, their implementation is often not trivial.
Here, the allocated memory and the references can be abstracted to the labeled vertices and edges of a graph.
That is, language processing systems can be regarded to be dealing a graph throughout their execution steps.
To implement a compiler, functional languages are often used.
However, functional languages cannot handle general graphs other than trees%
\footnote{
We can \emph{emulate} graph using a tree with keys.
However, we must pay at least \(\mathcal{O}(\log N)\) to traverse an edge.
In LMNtal (a graph rewriting language) on the other hand, we can achieve this in \(\mathcal{O}(1)\).
}.
And there exists a cost in time and space, to achieve a closure (e.g. we need a garbage collector).
Thus, it is not very suitable  especially to implement a runtime environment.
On the other hand, \emph{graph rewriting language}\cite{handbookgraph} is a programming language or a calculation model that can handle graph intuitively, safely and efficiently.
Therefore, the implementation of a language processing system can be highly expected as an application field of graph rewriting language.
To show this, in this research, we implemented the \emph{G-machine}\cite{g-machine}, the virtual machine for lazy evaluation, in hypergraph rewriting language, \emph{HyperLMNtal}\cite{hyperlmntal}.

\subsection{Hypergraph rewriting language: HyperLMNtal}

HyperLMNtal is a extension of the graph rewriting language/calculus model LMNtal\cite{logiclmntal}.
A graph in LMNtal is consisted of atoms, labeled vertices/data, possibly connected by links, the edges that connects precisely two endpoints.
Even though achieving to let graph to be its 1st class citizen, LMNtal is pointer safe;
we won't face any null or dangling pointers when we are using LMNtal. 
In short, we can handle graph safely and easily using LMNtal.

However, since links in LMNtal are only allowed to have exactly two endpoints, it is
rather hard to implement a shared data whose number of references changes dynamically.
Although it is possible, it was not very suitable from both the viewpoint of programming and the efficiency of the implementation.
On the other hand, HyperLMNtal allows existence of the hyperlink,
which is a link but allowed to have arbitrary number of endpoints.
Thus, we can easily implement a shared data with HyperLMNtal.

\section{Previous relevant researches}
This section gives brief summary of the previous relevant researches.

\subsection{Translating HIRAM into GP2}
Detlef Plump has shown that the simple imperative language, Hi-Level Random Access Machine (from now on, we abbreviate this as HIRAM), can be translated into the rule-based graph rewriting language, GP2\cite{gp2} in \cite{imp2gp2}.
The syntax of the HIRAM is shown in \Cref{table:hiram-syntax}%
\footnote{
This is mostly based on the \cite{imp2gp2} but we changed some notations
}.

Where \(a\) denotes an address (the index of the register) (\(a \in \mathbb{N} \cup \{0\}\)), \(\mathit{num}\) denotes a integer numeral (\(\mathit{num} \in \mathbb{Z}\)) and the \(\sigma\) (\(s\) in the original paper) denotes the states of the registers as follows:

\[
\begin{Bmatrix}
  \begin{array}{ccl}
  0 & \mapsto & \mathit{list}_0\\
  &\vdots \\
  n - 1 & \mapsto & \mathit{list}_{n - 1}\\
  n     & \mapsto & \mathit{list}_n\\
  n + 1 & \mapsto & \mathtt{empty}\\
  n + 2 & \mapsto & \mathtt{empty}\\
  n + 3 & \mapsto & \mathtt{empty}\\
  &\vdots \\
  \end{array}
\end{Bmatrix}
\]

where \(n\) is a finite number%
\footnote{
This is a simplified definition.
Please take a look at \cite{imp2gp2} for more detailed/formal definition
}.
A bit more precisely, it is a function from the address to the value (list) stored at the address.
Notice that a number is also a list of length 1 in HIRAM.
Thus \(a \subseteq \mathit{num} \subseteq \mathit{list}\) and the registers in HIRAM are also allowed to store addresses, \(a\), to perform pointer manipulations.
\(\sigma[\mathit{list}/a]\) denotes the updated state defined as follows:

\[
\sigma[\mathit{list}/a_1](a_2) \overset{def}{=}
\left\{
\begin{array}{ll}
  \mathit{list}   & \mbox{if } a_1 = a_2 \\
  \sigma(a_2) & \mbox{if } a_1 \neq a_2
\end{array}
\right.
\]

HIRAM has a loop instruction and is Turing complete. 
Moreover, it features a list, a structured data.
And in that sense, it maybe at a higher level.
However, it even lacks a function/procedure calls
and is certainly not a modern programming language but sort of an assembly language:
it is just a list of instructions for a random access machine.

\begin{figure}
  \hrulefill
  \begin{center}
    \scalebox{1.0}{
    \begin{tabular}{lrlp{0.4\textwidth}} 
      (Program)
      &\(\mathit{prog}\)
      &::= \(\mathit{prog} \vb{;} \mathit{prog}\) & Sequential composition\\
      &&\ $|$\ \(\vb{if } b \vb{ then } \mathit{prog} \vb{ else } \mathit{prog}\)
      & Branching\\
      &&\ $|$\ \(\vb{while } b \vb{ do } \mathit{prog}\)
      & Loop\\
      &&\ $|$\ \(a \vb{ := \$} \mathit{list}\)
      & \(\sigma[\mathit{list}/a]\)\\
      &&\ $|$\ \(a_1 \vb{ := } a_2\)
      & \(\sigma[\sigma(a_2)/a_1]\)\\
      &&\ $|$\ \(a_1 \vb{ := head } a_2\)
      & \(\sigma[\mathit{head}(\sigma(a_2))/a_1]\) \par
      (fails if \(\sigma(a_2) = \vb{empty}\))\\
      &&\ $|$\ \(a_1 \vb{ := tail } a_2\)
      & \(\sigma[\mathit{tail}(\sigma(a_2))/a_1]\) \par
      (fails if \(\sigma(a_2) = \vb{empty}\))\\
      &&\ $|$\ \(a_1 \vb{ := } a_2 \vb{:} a_3\)
      & \(\sigma[\sigma(a_2) \vb{:} \sigma(a_3)/a_1]\)\\
      &&\ $|$\ \(a_1 \vb{ := *} a_2\)
      & \(\sigma[\sigma(\sigma(a_2))/a_1]\)\\
      &&\ $|$\ \(\vb{*} a_1 \vb{ := } a_2\)
      & \(\sigma[\sigma(a_2)/\sigma(a_1)]\)\\
      &&\ $|$\ \(a_1 \vb{ := inc } a_2\)
      & \(\sigma[\sigma(a_2) + 1/a_1]\)\\
      &&\ $|$\ \(a_1 \vb{ := dec } a_2\)
      & \(\sigma[\sigma(a_2) - 1/a_1]\)\\
      \\
      (Condition)
      &\(b\)
      &::= \(a_1 \vb{ = } a_2\) &
      True if \(\sigma(a_1) = \sigma(a_2)\);\par
      false otherwise\\
      &&\ $|$\ \(a_1 \vb{ > } a_2\)
      & True if \(\sigma(a_1), \sigma(a_2) \in \mathbb{Z}\) and \par
      \(\sigma(a_1) > \sigma(a_2)\); \par
      false otherwise\\
      \\
      (list)
      &\(\mathit{list}\)
      &::= \vb{empty} & Empty list\\
      &&\ $|$\ \(\mathit{num}\) & Integers are lists of length 1\\
      &&\ $|$\ \(\mathit{list} \vb{:} \mathit{list}\) & Concatenation\\
    \end{tabular}
      }          
  \end{center}
  \hrulefill
  \caption{Syntax of HIRAM programs}
  \label{table:hiram-syntax}
\end{figure}

\subsection{Implementing a stack machine and a compiler in LMNtal}

Kokubo has implemented a stack machine and a compiler in graph rewriting language, LMNtal\cite{lmnstack}.
The syntax of the source language in this research is given in \Cref{table:lmnstack-syntax}%
\footnote{
Again, we have changed some notations.
Since the syntax in the original paper seems to be using some notations incorrectly.
}.
Now, the source language has function definitions/calls:
it is certainly a programming language.
However, it lacks the structured data (e.g. list).
And moreover, it has not accomplished the \emph{higher order function}.
Functions in the language is not the 1st class citizen.
That is, we cannot apply nor return any function.

\begin{figure}
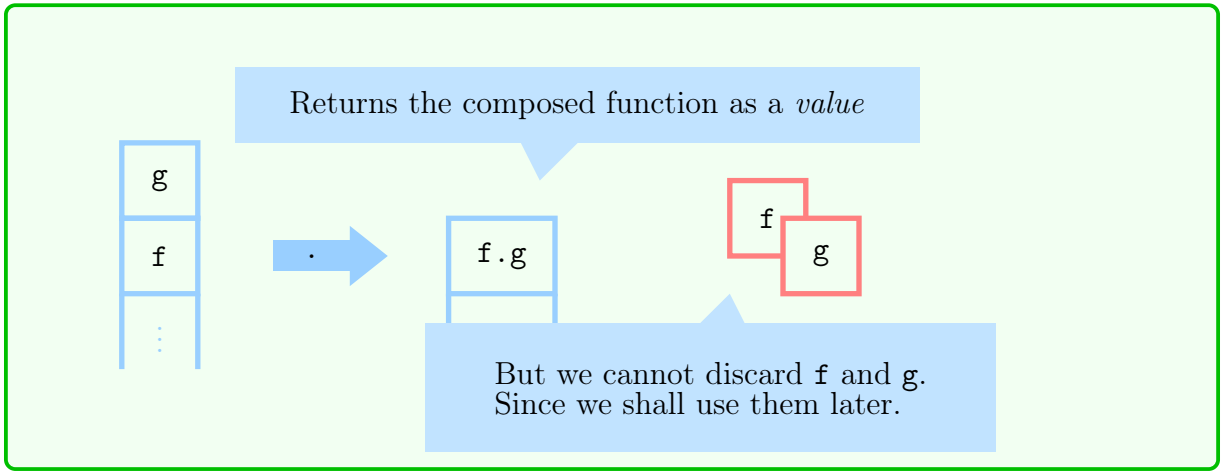

  \hrulefill
  \begin{center}
    \scalebox{0.85}{
    \begin{tabular}{llll} 
      (Program)
      \hfill\(\mathit{prog}\)
      &::= \(\mathit{def}_1 \vb{;} \ldots \vb{;} \mathit{def}_n \vb{;} \mathit{expr}\)
      & \(n \geq 0\)\\
      \\
      (Function Definition)
      \hfill\(\mathit{def}\)
      &::= \(\vb{def }\mathit{var} \vb{(} \mathit{var}_1, \ldots, \mathit{var}_n \vb{) = } \mathit{expr}
      \) & \(n \geq 0\)\\
      \\
      (Expression)
      \hfill\(\mathit{expr}\)
      &::= \(\mathit{num}\) && Integer numeral\\
      &\ $|$\ \(\mathit{var}\) && Variable \\
      &\ $|$\ \(\mathit{var} \vb{(} \mathit{expr}_1, \ldots, \mathit{expr}_n \vb{)}\)
      & \(n \geq 0\) & Function call\\
      &\ $|$\ \(\vb{var }\mathit{var} \vb{ = }\mathit{expr} \vb{ in }\mathit{expr}\) && Local definition\\
      &\ $|$ \(\mathit{expr}\ \mathit{binop}\ \mathit{expr}\)
      && Infix binary application \\
      &\ $|$ \(\mathtt{if\ } \mathit{expr}
      \mathtt{\ then\ } \mathit{expr} \mathtt{\ else\ } \mathit{expr}\)
      && If expression \\
      \\
      (Binary Operators) \hfill $\mathit{binop}$ & ::=
      \(\mathit{arithop}\ |\ \mathit{relop}\)
      \\
      \hfill $\mathit{arithop}$ & ::= \(
      \textnormal{\texttt{ + }}
      | \textnormal{\texttt{ - }}
      | \textnormal{\texttt{ * }}
      | \textnormal{\texttt{ / }}
      \)
      && Arithmetic \\
      \hfill $\mathit{relop}$ & ::= \(
      \textnormal{\texttt{ = }}
      | \textnormal{\texttt{ < }}
      | \textnormal{\texttt{ > }}
      \)
      && Comparison \\
    \end{tabular}
    }
  \end{center}
  \hrulefill
  \caption{Syntax of the language implemented in \cite{lmnstack}}
  \label{table:lmnstack-syntax}
\end{figure}

In short, the source languages in these previous researches are not functional language.

\subsection{Graph reduction}

Functional languages feature a \emph{higher order function} and thus they are certainly powerful than the imperative languages.
However, because we have to accomplish this, we cannot just use 1 simple stack to implement them.
Rather, we need a \emph{graph}.
This problem is known as the \emph{funarg problem}\cite{funarg-lisp}.

\begin{example}{Funarg problem}{}
  We can discard the argument of the function
  if the function has returned its value
  in the stack of the (call by value) imperative language.

  \begin{tikzpicture}[scale=1]
    \useasboundingbox (-1,0) rectangle(5, 5.5);
    
    \resource{0}{3}{2};
    \resource{0}{2}{1};
    \draw [line width = 2pt, color = myblue!50] (0, 2) -- (0, 1);
    \draw [line width = 2pt, color = myblue!50] (1, 2) -- (1, 1);        
    \node[] at (0.5, 1.5){\color{myblue!50}\(\vdots\)};
    
    \RightArrow{2}{2.5};
    \node[] at (2.5, 2.5){\texttt +};
    
    \resource{4.5}{2}{3};
    \draw [line width = 2pt, color = myblue!50] (4.5, 2) -- (4.5, 1);
    \draw [line width = 2pt, color = myblue!50] (5.5, 2) -- (5.5, 1);        
    \node[] at (5, 1.5){\color{myblue!50}\(\vdots\)};
    
    \draw [line width = 2pt, color = red!50, dashed]  
    (8,2) rectangle (9,3);
    \node[] at (8.5, 2.5){\color{black!40}\texttt 1};
    
    \filldraw [line width = 2pt, color = red!50, fill = green!5!white, dashed]  
    (8.7, 1.5) rectangle (9.7, 2.5);
    \node[] at (9.2, 2){\color{black!40}\texttt 2};
    
    \fill[ color = myblue!30 ] (5, 4) rectangle (11, 5);
    \fill[ color = myblue!30 ] 
    (8, 4.5)--(8.5, 3.5)--(9.5, 4.5)--cycle;
    \node[] at (8, 4.5){ 1 and 2 can be discarded };
    
  \end{tikzpicture}

  On the other hand, higher order functions may use the applied values
  even after it returned the value (including function).

  \begin{tikzpicture}[scale=1]
    \renewcommand{\baselinestretch}{0.8}
    \useasboundingbox (-1,0) rectangle(5, 5.5);

    \resource{0}{3}{g};
    \resource{0}{2}{f};
    \draw [line width = 2pt, color = myblue!50] (0, 2) -- (0, 1);
    \draw [line width = 2pt, color = myblue!50] (1, 2) -- (1, 1);        
    \node[] at (0.5, 1.5){\color{myblue!50}\(\vdots\)};
    
    \RightArrow{2}{2.5};
    \node[] at (2.5, 2.5){\texttt .};
    
    \draw [line width = 2pt, color = myblue!50]  
    (4.3, 2) rectangle (5.7, 3);
    \node[] at (5, 2.5){\texttt{f.g}};
    \draw [line width = 2pt, color = myblue!50] (4.3, 2) -- (4.3, 1);
    \draw [line width = 2pt, color = myblue!50] (5.7, 2) -- (5.7, 1);        
    \node[] at (5, 1.5){\color{myblue!50}\(\vdots\)};
    
    \draw [line width = 2pt, color = red!50]  
    (8, 2.5) rectangle (9, 3.5);
    \node[] at (8.5, 3){\texttt f};
    
    \filldraw [line width = 2pt, color = red!50, fill = green!5!white]  
    (8.7, 2) rectangle (9.7, 3);
    \node[] at (9.2, 2.5){\texttt g};
    
    \fill[ color = myblue!30 ] (1.5, 4) rectangle (10.5, 5);
    \fill[ color = myblue!30 ] 
    (5, 4.5)--(5.5, 3.5)--(6.5, 4.5)--cycle;
    \node[] at (6, 4.5){
      Returns the composed function as a \emph{value}
    };
    
    \fill[ color = myblue!30 ] (4, -0.1) rectangle (11.5, 1.6);
    \fill[ color = myblue!30 ] 
    (7, 1)--(8, 2)--(8.5, 1)--cycle;        
    \node[align=left] at (7.75, 0.75){
      But we cannot discard \texttt{f} and \texttt{g}.\\
      Since we shall use them later.
    };
    
  \end{tikzpicture}
\end{example}

A strict (call by value) functional language can possibly be implemented with just trees, although it would be a very na\"ive implementation.
However, especially for the lazy (call by need) functional language, we need a more general graph more than trees so that the evaluator can share subgraphs (subexpressions).
This, \emph{graph reduction}/\emph{graph rewriting}, is a very popular research area and there is a lot of studies as \cite{ongraphrewriting}.
However, there was no attempt to implement a compiler and a runtime environment for a (lazy) functional language in a graph rewriting language.
Therefore, in this research, we implemented G-machine\cite{g-machine}, the virtual machine for lazy evaluation, in hypergraph rewriting language, HyperLMNtal\cite{hyperlmntal}.

\section{Contributions}
The contributions of this research are basically the following 2:
\begin{itemize}
\item
  We formalized the syntax and the semantics of HyperLMNtal
\item
  We implemented G-machine and a compiler for it in HyperLMNtal
\end{itemize}

\subsection{Formalizing the syntax and the semantics of HyperLMNtal }
HyperLMNtal is extended from the graph rewriting language/calculus model LMNtal.
However, it lacked the rigid definition: it was more an extension of the implementation than on the calculus model.
The semantics of LMNtal features fine-grained concurrency based on local rewriting.
However, since we can not determine the locality of the hyperlink in HyperLMNtal, we couldn't incorporate it into the LMNtal semantics.
Thus, we first introduced a scope (link creation) adopted from the \(\pi\)-calculus\cite{pi} and defined the locality of a hyperlink to formalize the syntax/semantics.
Now, HyperLMNtal is not just a programming language extended from the basic calculation model, but also a concurrent calculation model based on strict and formal definitions.

%
%
\subsection{Implementing G-machine in HyperLMNtal}
The G-machine is a virtual machine that performs lazy evaluation, which is the basis of implementation of lazy functional programming languages such as Haskell.
The implementation of G-machine requires a heap, which is a more general graph than just a tree so as can share subgraphs.
Therefore, HyperLMNtal is ideal to implement this.

In our research, we implemented a compiler which translates the source language, the core language, into the execution code for G-machine and G-machine, in HyperLMNtal.
We have succeeded to implement the compiler in 404 lines and G-machine in 570 lines and showed that we can implement the language processing system that handles complex data structures in graph rewriting language tersely.
In addition, we achieved to visualize G-machine using the HyperLMNtal visualizer\cite{graphene}\cite{lavit}.

%
%
%

\section{Structure of this paper}
\par
The structure of this paper is as the following:

\begin{description}
\item[Chapter 2]\mbox{}\\
  We introduce the hypergraph rewriting language HyperLMNtal.
\item[Chapter 3]\mbox{}\\
  We introduce the new formal syntax and semantics of HyperLMNtal.
\item[Chapter 4]\mbox{}\\
  We introduce the design of G-machine.
\item[Chapter 5]\mbox{}\\
  We introduce the G-machine we have implemented in HyperLMNtal and some examples that run on the G-machine.
\item[Chapter 6]\mbox{}\\
  We give some discussions about HyperLMNtal as a programming language.
\item[Chapter 7]\mbox{}\\
  We give a conclusion of this paper and discuss the future task.
\end{description}

\chapter{HyperLMNtal: An introduction}
\label{cha:lmntal}
\par

HyperLMNtal\cite{hyperlmntal} is a hypergraph rewriting language, which is 
an extension of the graph rewriting language LMNtal.
In this chapter, we first briefly introduce the core syntax and the semantics of
the graph rewriting language LMNtal\cite{lmntal2004, lmntal2008}
in \Cref{sec:lmntal-syntax} and \Cref{sec:lmntal-op-sem} respectively and
its extension run on the current processing system SLIM\cite{slim}
in \Cref{sec:lmntal-ext}.
And then explain the current implementation of HyperLMNtal in \Cref{sec:hyperlmntal-old}.

%


\section{Syntax of LMNtal}
\label{sec:lmntal-syntax}

LMNtal is comprised of two kinds of identifiers:
\emph{Link names}, denoted by \(X\), which can be identifiers starting with capital letters in the concrete syntax
and \emph{atom names}, denoted by \(p\), which can be identifiers that are distinct with link names in the concrete syntax
(i.e. identifiers starting with lower letters, numbers, special symbols, etc).
The only reserved atom name is ``\texttt{=}'', which is the name for a \emph{connector} of links. 

The syntax of LMNtal is given in \Cref{table:lmntal-syntax}.
Intuitively, a \emph{process} is the LMNtal program it self.
Surprisingly there is no such thing as \emph{function/procedure calls}:
the calculation process proceeds with the \emph{matching} of a sub-graph (process) with the \emph{process templates} on left-hand side (LHS/Head) of a \emph{rule} 
and pushing the sub-graph based on process templates (and some substitution rules \(\theta\) described in \cite{lmntal2004}) on right-hand side (RHS/Body) of the matched rule.
The curly braces describe the grouping of the processes, which can compose hierarchies.
However, we can think of a \emph{Flat LMNtal} that does not feature this by omitting those with daggers(\mydagger).
\emph{Process context} and \emph{rule context} represent the rest process/rules of a membrane.
\emph{Residual} is the free links of a process context.

\begin{figure}
  \hrulefill
  \begin{center}
    \begin{tabular}{rlclll} 
      (Process) & $P$ &$::=$& $\zero$    && Null \\
      &&$|$& \(p (X_1, \ldots, X_m)\) & \(m \geq 0\) & Atom \\
      &&$|$& \((P, P)\)               && Molecule \\
      &&$|$& \(\{ P \}\)              && Cell \mydagger \\
      &&$|$& \((T \means T)\)         && Rule
      \vspace{1em}\\                                                     
      (Process Template) & $T$ &$::=$& $\zero$ && Null \\
      &&$|$& \(p (X_1, \ldots, X_m)\) & \(m \geq 0\) & Atom \\
      &&$|$& \((T, T)\)               && Molecule \\
      &&$|$& \(\{ T \}\)              && Cell \mydagger  \\
      &&$|$& \((T \means T)\)         && Rule \\
      &&$|$& \(\rulevar p\)           && Rule context \mydagger \\
      &&$|$& \(\procvar p\,[X_1, \ldots, X_m|A]\)
      & \(m \geq 0\) & Process context \mydagger 
      \vspace{1em}\\
      (Residual) & $A$ &$::=$& \([\,]\) && Empty \mydagger \\
      &&$|$& \(\fvstar X\)      && Bundle \mydagger   
    \end{tabular}
  \end{center}
  \hrulefill
  \caption{Syntax of LMNtal}
  \label{table:lmntal-syntax}
\end{figure}


Link names outside rules should occur at most twice: the endpoints of a link should be not more than two.
If a link name occurs once in a process, then it represents a \emph{free link} of the process.
And if a link name occurs twice in a process, then it represents a \emph{local link} of the process. As like the other calculus models, local links can be \(\alpha\)-converted: the name of the links convey no \emph{data} at all.
As above, we can distinguish local and free links (c.f. local variable and free variable in other calculus models such as \(\lambda\) calculus or \(\pi\) calculus) just by counting the number of the occurrences of link names without some notations like \(\lambda\) or \(\nu\). 

Link names in left/right-hand sides of a rule must
occur twice on left/right-hand side
or occur once on each of the hands sides of a rule. 
That is, the free links of the left and the right-hand sides of a rule must be the same. Thus rewriting of the process won't delete/yield free links.
This is also a notable point; if we had a link that has exactly two endpoints, then the link should always store the property, and there is no chance that a \emph{null}/\emph{dangling} pointer (link) could occur.

\subsection{Abbreviations}

\begin{enumerate}
\item
  \texttt{()} (the parenthesis of the nullary atom) and \texttt{[\,]} can be omitted.
\item
  \(p(s_1, \ldots, s_m), q(t_1, \ldots, t_n)\)
  can be abbreviated as \(q(t_1, \ldots, p(s_1, \ldots, s_{m-1}), \ldots, t_n))\)
  if the $s_m$ and the $t_i$ has the same link name.
\end{enumerate}

Also, we define the operator precedence as ``\(.\)'' \(<\) ``\(\means\)'' \(<\) ``\(,\)'' where period is just an another form of a comma but with lower connectivity.
We can omit parentheses if there is no syntactic ambiguities. 

\section{Operational Semantics of LMNtal}
\label{sec:lmntal-op-sem}

We first describe the structural congruence $(\equiv)$ rules of a process
(i.e. the definition of the equivalence of the process) and the reduction relation $(\longrightarrow)$ (i.e. the calculation step).

\subsection{Structural congruence}

The relation $\equiv$ on processes is defined as the minimal equivalence relation
which satisfies the rules in \Cref{table:lmntal-equiv},
where
$P[Y / X]$ denotes a \textit{link substitution}: substitution of $X$ with $Y$ in \(P\).

(E1)--(E3) let molecules to be multisets.
(E4) is a rule for the $\alpha$-conversion of local link names.
Though we have to choose the new link name \(Y\), we don't have to worry about \emph{variable capturing} in a scope with \(\lambda\) or \(\nu\).
Again, this gives a notable simplicity to LMNtal rather than other calculus models.
(E5)--(E6) makes $\equiv$ a congruence; the smaller parts should always be equivalent if the comprised processes are equivalent.
(E7)--(E10) are the rules for connectors.

We prove that (E8) is admissible, in some sense, it is redundant, in \Cref{th:symmetry-of-connector}.

\begin{figure}
  \hrulefill
  \begin{center}
    \begin{tabular}{ lll }
      (E1) & \(\zero, P \equiv P\)        & \\
      (E2) & \(P, Q \equiv Q, P\)           & \\
      (E3) & \(P, (Q, R) \equiv (P, Q), R\) & \\
      (E4) & \(P \equiv P[Y/X]\)            & if $X$ if a local link of $P$ \\
      (E5) & \(P \equiv P' \Rightarrow P, Q \equiv P', Q\) & \\
      (E6) & \(P \equiv P' \Rightarrow \{P\} \equiv \{P'\}\) &\\
      (E7) & \(X = X \equiv \zero \)        & \\
      (E8) & \(X = Y \equiv Y = X \)        & \\
      (E9) & \(X = Y, P \equiv P[Y/X]\)     &
      if $P$ is an atom and $X$ occurs free in $P$ \\
      (E10)& \(\{X = Y, P\} \equiv X = Y, \{P\}\) &
      if exactly one of $X$ and $Y$ occurs free in $P$ \\
    \end{tabular}
  \end{center}
  \hrulefill
  \caption{Structural congruence on LMNtal processes}
  \label{table:lmntal-equiv}
\end{figure}

\begin{theorem}{Symmetry of $=$}{symmetry-of-connector}
  \[X = Y \equiv Y = X\]
\end{theorem}

\begin{proof}{}{}
  \begin{align*}
    & Z = X, Z = Y
    \\
    \equivby{E9}
    & X = Y & \because (Z = Y)[X/Z] = (X = Y)
  \end{align*}
  and
  \begin{align*}
    & Z = X, Z = Y
    \\
    \equivby{E2}
    & Z = Y, Z = X
    \\
    \equivby{E9}
    & Y = X & \because (Z = X)[Y/Z] = (Y = X)
  \end{align*}
  
  Therefore, \(X = Y \equiv Y = X\).

\end{proof}

\subsection{Reduction relation}

The reduction relation $\longrightarrow$ (calculation step) on processes are defined as the minimal relation which satisfies the rules given in \Cref{table:lmntal-trans}.

(R1) states that the reductions can proceed in a local process (besides the other processes together compromise a molecule).
This characterize LMNtal as a \emph{concurrent} language.
(R2) states that the reductions in a cell (membrane) can proceed by it own, even if there is no help from outer processes.
(R3) introduce the structural congruence described before. A matching of a process and a LHS of a rule will be done with changing the process (not LHS of a rule) using the congruence rules.
(R4) and (R5) are used when extracting/entering connectors from/in membranes.
(R6) is the most important rule. 
The detailed explanation of the meaning of the \(\theta\) is given in \cite{lmntal2004} and We won't explain this fully:
in short, it describes the matching of process/rule contexts; the matching of the rest process/rules.

\begin{figure}
  \hrulefill
  \begin{center}
    \renewcommand{\arraystretch}{1.2}
    \begin{tabular}{ lll } 
      (R1) & \(\dfrac{P \longrightarrow P'}{P, Q \longrightarrow  P', Q}\) & 
      \vspace{1em} \\
      (R2) & \(\dfrac{P \longrightarrow P'}{\{P\} \longrightarrow  \{P'\}} \) & 
      \vspace{1em} \\
      (R3) &
      \( \dfrac{Q \equiv P \hspace{1em} P \longrightarrow P' \hspace{1em} P' \equiv Q'}{Q \longrightarrow Q'} \) &
      \vspace{0.5em} \\
      (R4) & \( \{X = Y, P\} \longrightarrow X = Y, \{P\} \) &
      if $X$ and $Y$ occur free in \(\{X = Y, P\}\) \\
      (R5) & \(X = Y, \{P\} \longrightarrow \{X = Y, P\}\) &
      if $X$ and $Y$ occur free in $P$ \\
      (R6) & \(T\theta, (T \vdash U) \longrightarrow U\theta, (T \vdash U)\) & \\
    \end{tabular}
  \end{center}
  \hrulefill
  \caption{Reduction relation on LMNtal processes}
  \label{table:lmntal-trans}
\end{figure}

\section{Extensions of LMNtal}
\label{sec:lmntal-ext}

As like the other programming languages originated from calculus models (e.g. LISP is based on \(\lambda\)-calculus), LMNtal has some extensions for practical programming.

Rules can be denoted with their names as ``\(\mathit{rulename}\ @@\ H \means B\)''.
Rules can have \emph{guards} as ``\(\mathit{Head} \means \mathit{Guard} \verb+|+ \mathit{Body}\)''.
The features of guard are for 
(\romannumeral 1) type checking (as shown in \Cref{table:lmntal-types})
(\romannumeral 2) comparing two processes (as shown in \Cref{table:lmntal-comparison})
and
(\romannumeral 3) performing some primitive instructions (as shown in \Cref{table:lmntal-primitive-operators}).
Typed process can be copied and removed.

\begin{figure}
  \hrulefill
  \begin{center}
    \renewcommand{\arraystretch}{1.2}
    \begin{tabular}{ lp{0.7\textwidth} } 
      \texttt{ground(\$p)}
      & checks that a process \texttt{\$p} constituting a connected sub-graph
      that has one or more free links \\
      \texttt{unary(\$p)}
      & checks that a process \texttt{\$p} is an atom with one link \\
      \texttt{int(\$p)}
      & checks that a process \texttt{\$p} is an atom with one link
      whose name is an integer \\
      \texttt{string(\$p)}
      & checks that a process \texttt{\$p} is an atom with one link
      whose name is an string \\
    \end{tabular}
  \end{center}
  \hrulefill
  \caption{Primitive types in LMNtal}
  \label{table:lmntal-types}
\end{figure}

\begin{figure}
  \hrulefill
  \begin{center}
    \begin{tabular}{ l l } 
      \verb|$p + $q| & addition \\
      \verb|$p - $q| & subtraction \\
      \verb|$p * $q| & multiplication \\
      \verb|$p / $q| & division \\
    \end{tabular}
  \end{center}
  \hrulefill
  \caption{Primitive operators in LMNtal}
  \label{table:lmntal-primitive-operators}
\end{figure}

\begin{figure}
  \hrulefill
  \begin{center}
    \begin{tabular}{ llll }      
      ground & unary & int & meanings\\
      \verb|$p = $q| & \verb|$p == $q| & \verb|$p =:= $q| 
      & The structures of \verb|$p| and \verb|$q| are the same \\
      \verb|$p \= $q| & \verb|$p \== $q| & \verb|$p =\= $q|
      & The structures of \verb|$p| and \verb|$q| are not the same \\ 
       - & - & \verb|$p < $q|
      & \verb|$p| is less than \verb|$q| \\
       - & - & \verb|$p > $q|
      & \verb|$p| is greater than \verb|$q| \\ 
       - & - & \verb|$p =< $q|
      & \verb|$p| is less than or equals to \verb|$q| \\
       - & - & \verb|$p >= $q|
      & \verb|$p| is greater than or equals to \verb|$q| \\ 
   \end{tabular}
  \end{center}
  \hrulefill
  \caption{Comparison operators in LMNtal}
  \label{table:lmntal-comparison}
\end{figure}

\section{The implementation of HyperLMNtal}
\label{sec:hyperlmntal-old}
\par
HyperLMNtal\cite{hyperlmntal} is an extension of LMNtal.
HyperLMNtal allows the existence of the \emph{hyper links}, which can be connected to arbitrary number of ports (endpoints).
For example, they can be used to represent a connection that the number of endpoints changes dynamically through the calculation process.


\subsection{The syntax and the manipulation of hyperlinks in the current implementation}
\label{subsec:hyperlmntal-manipulation}

We briefly explain the newly added syntax and the manipulations for hyperlinks.

Hyperlinks are expressed as ``\texttt{\$x}'' (i.e. an extension of the process context) or as ``\(!X\)''.

\texttt{new} guard construct creates a new hyperlink with a fresh local id.
In the following example, the new hyperlink \texttt{\$x} is created by \texttt{new}.
\[
H \means \mathtt{ new(\$x)}\ |\ B
\]
Using this, all hyperlinks should have distinct names so that they are distinguishable.
Current compiler implementation allows users to write a new hyperlink with exclamation mark on RHS without explicitly writing \texttt{new}. At that situation, compiler will automatically add it.

\texttt{hlink} guard construct checks whether the given link is a hyperlink.
\[
H \means \mathtt{hlink(\$x)}\ |\ B
\]

\texttt{><} on RHS of a rule is called a \emph{fusion}, which perform the merge of the given two hyperlinks.
\[
H \means \mathtt{!X ><\ !Y}, B
\]

\texttt{num} guard construct binds the number of the endpoints of the given hyperlink at the first argument to the second argument.
\[
H \means \mathtt{num(!X, \$n)}\ |\ B
\]

Hyperlinks on LHS with the same name should match the hyperlinks with the same id.
Notice the hyperlinks on LHS with different names can also match the same hyperlinks
(i.e. non-injective matching).
For example,
\begin{lstlisting}[language=lmntal]
init.
init :- new($x) | a($x), b($x).
a(!X), b(!Y) :- c.
\end{lstlisting}
will be reduced to \texttt{c} and the remaining rules.

Our current implementation cannot handle hyperlinks in the initial state.
Therefore Hyperlinks must not appear outside of rules.
Also, our current implementation does not allow to connect hyperlinks with a connector (=).

\chapter{Formalizing HyperLMNtal}
\label{cha:formalize-hlmntal}
\par

We have implemented HyperLMNtal in spite of the fact that the hyperlinks can be (inaccurately) modeled with membranes
for 
(\romannumeral 1) the efficiency and
(i.e. using membrane costs too much when we just want to represent a multi-connected link)
(\romannumeral 2) the convenience for the programmer
(i.e. hyperlinks are much easier to read/write rather than membranes).
Therefore, in the past, we have not focused on its semantics that much:
HyperLMNtal was a more an extension of the implementation rather than a calculus model.

Though the implementation of HyperLMNtal is rather simple and the concept can be easily understood,
we found that the semantics cannot be defined easily.
We will firstly briefly discuss the difficulty in the semantics in \Cref{sec:background-of-hyperlmntal-sem}.
And then introduce the new semantics which we propose in \Cref{sec:flat-hyperlmntal-syntax}, \Cref{sec:flat-hyperlmntal-op-sem} and \Cref{sec:hyperlmntal-semantics}.

\Cref{sec:flat-hyperlmntal-syntax} and \Cref{sec:flat-hyperlmntal-op-sem} describe the syntax and the semantics of Flat HyperLMNtal.
This Flat HyperLMNtal has no normal links but hyperlinks.
This is just because we want to discuss the properties of hyperlinks.
In \Cref{sec:hyperlmntal-semantics}, we extend it so that it will be able to deal with membranes, process/rule contexts, etc.

\section{Introduction to the formalization of HyperLMNtal}
\label{sec:background-of-hyperlmntal-sem}

\subsection{The difficulty in HyperLMNtal semantics}
\label{subsec:hyperlmntal-difficulty}

The links in LMNtal are only allowed to appear once or twice.
This enables us to distinguish the former to be the local link and the latter to be the free link:
the locality of the normal link can be determined by just counting the number of the appearance.
However, since hyperlinks can appear more than twice, it is difficult to determine whether the given link is a local link or a free link in a process.

If we cannot ensure the locality of the link, we cannot ensure the connector (fusion) has ended the link substitution by looking the local process.

\begin{example}{The difficulty in defining fusion}{}
  For example,
  \[ (a(!X), (!X \bowtie\ !Y, b(!X, !Y))) \]
  should be congruent with
  \[ (a(!X), b(!X, !X)) \]
  but
  \[ (a(!X), b(!Y, !Y)) \]

  Therefore, we cannot just define the structural congruence rule as
  \[!X \bowtie\ !Y, P \equiv P[!Y/!X] \mbox{ if \(!X\) occurs free in \(P\)}\]
  
  We must ensure that there is no \(!X\) occurs outside of the \(P\).
  That is, \(!X\) should be the local link of ``\(!X \bowtie\ !Y, P\)''.
  
  In (not-hyper) LMNtal, a link appeared twice is a local link,
  therefore the condition of (E9), ``\textit{if \(X\) occurs free in \(P\)}'',
  requires the \(X\) to be a local link in ``\(X = Y, P\)''
  (\(\because\) \(X\) occurs in both \(X = Y\) and \(P\)).

  This problem would not happen if we always ``left'' the $!X \bowtie\ !Y$.
  For example, as 
  \[!X \bowtie\ !Y, P \equiv\ !X \bowtie\ !Y, P[!Y/!X] \]
  
  However, if we cannot ensure the locality of the hyperlink $!X$,
  then we cannot determine whether we can diminish the \(!X \bowtie\ !Y\) or not.

  If we cannot diminish the \(!X \bowtie\ !Y\),
  then we cannot create it neither (applying the congruence rule reversely).
  Therefore, the existence of the \(!X \bowtie\ !Y\) would let the process
  different from that does not have it,
  which is certainly not the desired behavior based on the current implementation.
\end{example}

This force the semantics to always care about the whole program,
which spoils the concurrency based on local reducibility
and makes it hard to introduce the congruence rules
since the congruence of a process relies on the congruence of its constitutions.
That is, we must be able to determine the congruence on smaller processes.

The current implementation does not just look for a local process through the computation process,
it perform rewriting of the graphs \emph{globally},
thus this did not draw so much interest before;
the semantics was regarded as something obvious according to our implementation.
However, strict definition is necessary especially when we want to
(\romannumeral 1) introduce a new idea (e.g. capability typing) and
(\romannumeral 2) alter/extend the calculus model for the desired demands
(e.g. program transformation for more efficient implementation, the extension on \emph{hlground}, altering the calculus model to deal with directed hyper-graphs).

\subsection{Brief overview of the proposed semantics of HyperLMNtal}
Since we need to distinguish whether the hyperlink in a given process is a local link (occurs only in the process) or a free link (may occur outside of the process),
we introduce a \emph{name restriction} mechanism
as ``\(\lambda X.E\)'' (\emph{abstraction}) in \(\lambda\) calculus and ``\((\nu X) P\)'' (\emph{channel creation}) in \(\pi\) calculus.

In the proposed Flat HyperLMNtal semantics,
\(\nu X.P\) ensures that the hyperlink \(X\) (in the proposed abstract syntax and semantics of Flat HyperLMNtal, the names of hyperlinks are denoted with \(X\) rather than \(!X\) and we shall call them just ``link'' for the simplicity) is a local link in the process \(P\).

Then the correspondence of the (E9) of the LMNtal,
\[
X = Y, P \equiv P[Y/X]
\ \mbox{
  \begin{minipage}{0.5\textwidth}
  if \(X\) occurs free in \(P\) and \(P\) is an atom\\
  \emph{And \(X\) is a local link in \(X = Y, P\)}
  \end{minipage}
}
\]
 is 
\[
\nu X.(X \bowtie Y, P) \equiv \nu X.P[Y/X]
\ \mbox{ where \(X\) or \(Y\) occurs free in \(P\)}
\]

From the next section, we explain this formalized syntax and semantics in more detail.

\section{Syntax of Flat HyperLMNtal}
\label{sec:flat-hyperlmntal-syntax}

\subsection{Processes}
As well as LMNtal, HyperLMNtal is comprised of two kinds of identifiers:
\begin{itemize}
\item 
  $X$ denotes a link name. 
  
\item 
  $p$ denotes an atom name. 
  The only reserved name is $\bowtie$.
\end{itemize}

The syntax is given in \Cref{table:flat-hyperlmntal-syntax}. 

\begin{figure}
  \hrulefill
  \begin{center}
    \begin{tabular}{ lrclll } 
      (Process) & $P$&$::=$& \(\zero\) && Null \\
      &&$|$& \(p (X_1, \ldots ,X_m)\) & \(m \geq 0\) & Atom \\
      &&$|$& \((P, P)\) && Molecule \\
      &&$|$& \(\nu X.P\) && Link creation \\
      &&$|$& \((P \vdash P)\) && Rule \\
    \end{tabular}
  \end{center}
  \hrulefill
  \caption{Syntax of Flat HyperLMNtal}
  \label{table:flat-hyperlmntal-syntax}
\end{figure}

An atom \(X \bowtie Y\) is called a (link) \emph{fusion}.

The set of the free link names in a process $P$ is denoted as \fn{P} and is defined inductively in \Cref{table:free-names}.

\begin{figure}
  \hrulefill
  \[
  \begin{aligned}
      \mathit{fn}(\mathbf{0})          &= \emptyset \\
      \mathit{fn}(p(X_1, \ldots ,X_m)) &= \bigcup_{i=1}^m \{X_i\} \\
      \mathit{fn}((P, Q))              &= \mathit{fn}(P) \cup \mathit{fn}(Q) \\
      \mathit{fn}(\nu X.P)             &= \mathit{fn}(P) \setminus \{X\} \\
      \mathit{fn}((P \vdash Q))        &= \emptyset
  \end{aligned}
  \]
  \hrulefill
  \caption{A set of free link names of a Flat HyperLMNtal process}
  \label{table:free-names}
\end{figure}

\subsection{Rules}
Given a rule $(P \vdash Q)$, $P$ is called the left-hand side and $Q$ is called the right-hand side of the rule.
A rule $(P \vdash Q)$ must satisfy the following conditions.

\begin{enumerate}
\item Rules must not appear in $P$.
\item \(\mathit{fn}(P) \supseteq \mathit{fn}(Q)\).
\end{enumerate}

Intuitively, the latter condition indicates that we have to denote a \emph{new} hyperlink in a scope of a \(\nu\) (new) on RHS, which we believe follows a common sense.

\section{Operational Semantics of Flat HyperLMNtal}
\label{sec:flat-hyperlmntal-op-sem}

We first define structural congruence $(\equiv)$ and then define the reduction relation $(\longrightarrow)$ on processes.

\subsection{Structural congruence}

We define the relation $\equiv$ on processes as the minimal equivalence relation satisfying the rules shown in \Cref{table:equiv}.
Where $P[Y/X]$ is a link substitution that replaces all free occurrences of $X$ with $Y$ as defined in \Cref{table:hyperlink-substitution}.
Notice if a free occurrence of $X$ occurs in a location where $Y$ would not be free, $\alpha$-conversion may be required.
Here, we use $=$ to denote the syntactic equivalence of the links and processes.

The proposed congruence rules lack the corresponding rules for
(E4) \(P \equiv P[Y/X]\) (if $X$ occurs free in $P$) and
(E8) \(X = Y \equiv Y = X\)
in (not-hyper) LMNtal.
This is because that they can be derived from the other rules.
In other word, they are admissible.
We prove this in \Cref{th:alpha-equiv} and \Cref{th:symmetry-of-bowtie}.

Also, we prove that the sets of the free link names in congruent processes are the same in \Cref{th:freelinks-of-equiv}.

\begin{figure}
  \hrulefill
  \begin{center}
    \begin{tabular}{l}
      \( \zero[Y/X] \overset{def}{=} \zero \) \vspace{0.5em}\\ 
      \( p(X_1, \ldots, X_m)[Z/Y] \overset{def}{=} p(X_1[Z/Y], \ldots, X_m[Z/Y]) \)
      \vspace{0.3em}\\
      \hspace{11em} where \(
      X_i[Z/Y] \overset{def}{=}
      \left\{
      \begin{array}{ll}
        Z   & \mbox{if } X_i = Y \\
        X_i & \mbox{if } X_i \neq Y
      \end{array}
      \right.
      \) \vspace{0.5em}\\ 
      \((P, Q)[Y/X] \overset{def}{=} (P[Y/X], Q[Y/X])\) \vspace{0.5em}\\
      \(
      (\nu X.P)[Z/Y] \overset{def}{=}
      \left\{
      \arraycolsep = 0pt
      \begin{array}{ll}
        \nu X.P & \mbox{ if } X = Y \\
        \nu X.P[Z/Y] & \mbox{ if } X \neq Y \land X \neq Z \\
        \nu W.(P[W/X])[Z/Y] & \mbox{ if }
        X \neq Y \land X = Z
        \land W \notin \mathit{fn}(P) \land W \neq Z
      \end{array}
      \right.
      \) \vspace{0.5em}\\
      \( (P \vdash Q)[Y/X] \overset{def}{=} (P \vdash Q) \) \\
    \end{tabular}
  \end{center}
  \hrulefill
  \caption{Link substitution of Flat HyperLMNtal}
  \label{table:hyperlink-substitution}
\end{figure}

\begin{figure}
  \hrulefill
  \begin{center}
    \begin{tabular}{ ll } 
      (E1)  & \((\mathbf{0}, P) \equiv P\) \\
      (E2)  & \((P, Q) \equiv (Q, P)\) \\
      (E3)  & \((P, (Q, R)) \equiv ((P, Q), R)\) \\
      (E4)  & \(P \equiv P' \Rightarrow (P, Q) \equiv (P', Q)\) \\
      (E5)  & \(P \equiv Q \Rightarrow \nu X.P \equiv \nu X.Q\) \\
      (E6)  & \(\nu X.(X \bowtie Y, P) \equiv \nu X.P[Y / X]\)\\
      & where \(X \in \mathit{fn}(P) \lor Y \in \mathit{fn}(P)\) \\
      (E7)  & \(\nu X.\nu Y.X \bowtie Y \equiv \zero\) \\
      (E8)  & \(\nu X.\zero \equiv \zero\)\\
      (E9)  & \(\nu X.\nu Y.P \equiv \nu Y.\nu X.P\)\\
      (E10) & \(\nu X.(P,Q) \equiv (\nu X.P,Q)\)\\
      & where \(X \notin \mathit{fn}(Q)\) \\
    \end{tabular}
  \end{center}
  \hrulefill
  \caption{Structural congruence on Flat HyperLMNtal processes}
  \label{table:equiv}
\end{figure}

\begin{lemma}{Absorption of a futile link creation}{absorb-link-creation}
  \[\nu X.P \equiv P \mbox{ where } X \notin \mathit{fn}(P)\]
\end{lemma}

\begin{proof}{}{}
  \[\begin{array}{llr}
  & \nu X.P &
  \\
  \equiv_{\mbox{\scriptsize E5}}
  & \nu X.(\zero, P)
  & \mbox{where } P \equiv_{\scriptsize \mbox{E1}} (\zero, P)
  \\
  \equiv_{\mbox{\scriptsize E10}}
  & (\nu X.\zero, P)
  & \mbox{where } X \notin \mathit{fn}(P)
  \\
  \equiv_{\mbox{\scriptsize E4}}
  & (\zero, P)
  & \mbox{where } \nu X.\zero \equiv_{\mbox{\scriptsize E8}} \zero
  \\
  \equiv_{\mbox{\scriptsize E1}} & P
  \end{array}\]    
\end{proof}

\begin{lemma}{Absorption of a futile link substitution}{futile-link-substitution}
  $P[X/X]$ is syntactically equivalent with $P$
\end{lemma}

\begin{proof}{}{}
  We prove this by structural induction on processes.

  \begin{itemize}[label=\(\lozenge\), itemsep=2ex]
  \item \emph{Case \zero} :
    \begin{quote}
      \(\zero[X/X] \overset{def}{=} \zero\)
    \end{quote}
  \item \emph{Case \(p(X_1, \ldots, X_m)\)} :
    \begin{quote}
      \( p(X_1, \ldots, X_m)[X/X] \overset{def}{=} p(X_1[X/X], \ldots, X_m[X/X])
      = p(X_1, \ldots, X_m)\) 

      \hspace{1em} since \( X_i[X/X] = X_i \) where \(
      X_i[X/X] \overset{def}{=}
      \left\{
      \begin{array}{ll}
        X   & \mbox{if } X_i = X \\
        X_i & \mbox{if } X_i \neq X
      \end{array}
      \right.
      \)
    \end{quote}
  \item \emph{Case \((P, Q)\)} :
    \begin{quote}
      We have \(P[X/X] = P\) and \(Q[X/X] = Q\) by induction hypothesis. 

      Therefore,
      \((P, Q)[X/X] \overset{def}{=} (P[X/X], Q[X/X]) = (P, Q)\)
    \end{quote}
  \item \emph{Case \(\nu Y.P\)} :
    \begin{quote}
      \(
      (\nu Y.P)[X/X] \overset{def}{=}
      \left\{
      \begin{array}{ll}
        \nu Y.P & \mbox{if } Y = X \vspace{1em}\\
        \nu Y.P[X/X] & \mbox{if } Y \neq X \\
        = \nu Y.P & \because \mbox{
          \begin{minipage}[t]{15em}
            \(P[X/X] = P\) \\ 
            by induction hypothesis
          \end{minipage}
          }
      \end{array}
      \right.
      \)
      
      Since \(Y \neq X \land Y = X\) could never happen,
      there is no chance for ``variable capturing'' and
      $\alpha$-conversion for its avoidance
      (the third option of the link substitution scheme for the link creation),
      which possibly makes the process not syntactically equivalent
      (\(\alpha\)-equivalent though), won't happen.
    \end{quote}
  \item \emph{Case \((P \vdash Q)\)} :
    \begin{quote}
      \( (P \vdash Q)[X/X] \overset{def}{=} (P \vdash Q)\)
    \end{quote}
  \end{itemize}
\end{proof}

\begin{theorem}{$\alpha$-equivalence}{alpha-equiv}
  \[\nu X.P \equiv \nu Y.P[Y / X] \mbox{ where } Y \notin \mathit{fn}(P)\]
\end{theorem}

\begin{proof}{}{}
  We prove this using \Cref{lem:absorb-link-creation} and \Cref{lem:futile-link-substitution}.
  \begin{itemize}[label=\(\lozenge\), itemsep = 2ex]
  \item \emph{Case \(X \in \mathit{fn}(P)\)} :
    \begin{quote}
      \(\begin{array}{ll}
      \nu X. \nu Y. (Y \bowtie X, (X \bowtie Y, P)) &
      \\
      \equivby{E5, E6}
      \nu X. \nu Y. (X \bowtie X, P)
      & \because P[X/Y] = P \mbox{ since } Y \notin \mathit{fn}(P)
      \\
      \equivby{E5, \Cref{lem:absorb-link-creation}}
      \nu X.(X \bowtie X, P)
      & \because Y \notin \mathit{fn}((X \bowtie X, P))
      \\
      \equivby{E6}
      \nu X.P
      & \because P[X/X] = P \mbox{ by \Cref{lem:futile-link-substitution}}
      \\
      \end{array}\)

      and

      \(\begin{array}{ll}
      \nu X. \nu Y. (Y \bowtie X, (X \bowtie Y, P)) &
      \\
      \equivby{E2, E3, E5, E9}
      \nu Y. \nu X. (X \bowtie Y, (Y \bowtie X, P)) &
      \\
      \equivby{E5, E6}
      \nu Y. \nu X. (Y \bowtie Y, P[Y/X]) & 
      \\
      \equivby{E5, \Cref{lem:absorb-link-creation}}
      \nu Y. (Y \bowtie Y, P[Y/X])
      & \because X \notin \mathit{fn}((Y \bowtie Y, P[Y/X]))
      \\
      \equivby{E6}
      \nu Y.P[Y/X]
      & \because (P[Y/X])[Y/Y] = P[Y/X] \\
      & \ \ \mbox{ by \Cref{lem:futile-link-substitution}}
      \end{array}\)
    \end{quote}
  \item \emph{Case \(X \notin \mathit{fn}(P)\)} :
    \begin{quote}
      In this case, we first forcibly include free link \(X\)
      in order to exploit the former proof. 

      \(\begin{array}{ll}
      \nu X. P
      \\
      \equivby{\Cref{lem:absorb-link-creation}}
      P
      \\
      \equivby{E1}
      (\zero, P)
      \\
      \equivby{E4, E7}
      (\nu X. \nu X. X \bowtie X, P)
      \\
      \equivby{E4, \Cref{lem:absorb-link-creation}}
      (\nu X. X \bowtie X, P)
      & \because X \notin \mathit{fn}(\nu X. X \bowtie X)
      \\
      \equivby{E10}
      \nu X. (X \bowtie X, P)
      & \because X \notin \mathit{fn}(P)
      \\
      \equivby{The former proof}
      \nu Y. (X \bowtie X, P)[Y/X]
      & \because X \in \mathit{fn}((X \bowtie X, P))
      \\
      =
      \nu Y. (Y \bowtie Y, P[Y/X])
      \\
      \equivby{E6}
      \nu Y. (P[Y/X])[Y/Y]
      & \because Y \in \mathit{fn}((Y \bowtie Y, P[Y/X]))
      \\
      \equivby{\Cref{lem:futile-link-substitution}}
      \nu Y. P[Y/X]
      \end{array}\)
    \end{quote}
  \end{itemize}
  
  Therefore, \(\nu X.P \equiv \nu Y.P[Y/X]\) if \(Y \notin \mathit{fn}(P)\).

\end{proof}

\begin{theorem}{Symmetry of $\bowtie$}{symmetry-of-bowtie}
  \[X \bowtie Y \equiv Y \bowtie X\]
\end{theorem}

\begin{proof}{}{}
  \begin{quote}
    \(\begin{array}{ll}
    \nu Z. (Z \bowtie X, Z \bowtie Y) &
    \\
    \equivby{E6}
    \nu Z. (X \bowtie Y)
    & \because (Z \bowtie Y)[X/Z] = X \bowtie Y
    \\
    \equivby{\Cref{lem:absorb-link-creation}}
    X \bowtie Y
    & 
    \\
    \end{array}\)
  \end{quote}
  and
  \begin{quote}
    \(
    \begin{array}{ll}
      \nu Z. (Z \bowtie X, Z \bowtie Y) &
      \\
      \equivby{E2, E5}
      \nu Z. (Z \bowtie Y, Z \bowtie X) &
      \\
      \equivby{E6}
      \nu Z. (Y \bowtie X)
      & \because (Z \bowtie X)[Y/Z] = Y \bowtie X
      \\
      \equivby{\Cref{lem:absorb-link-creation}}
      Y \bowtie X
      & 
      \\
    \end{array}
    \)
  \end{quote}
  
  Therefore, \(X \bowtie Y \equiv Y \bowtie X\).

\end{proof}

\begin{theorem}{The sets of the free links of congruent processes}{freelinks-of-equiv}
  \[\mathit{fn}(P) = \mathit{fn}(Q) \mbox{ if } P \equiv Q\]
\end{theorem}

\begin{proof}{}{}
  We prove this by structural induction on processes (and structural congruent rules).

  \begin{itemize}[label=\(\lozenge\), itemsep=2ex]
  \item \emph{Case \((\zero, P) \equiv P\)} :
    \begin{quote}
      \(\mathit{fn}((\zero, P))
      = \mathit{fn}(\zero) \cup \mathit{fn}(P)
      = \mathit{fn}(P)\)
    \end{quote}
  \item \emph{Case \((P, Q) \equiv (Q, P)\)} :
    \begin{quote}
      \(
      \mathit{fn}((P, Q))
      = \mathit{fn}(P) \cup \mathit{fn}(Q)
      = \mathit{fn}(Q) \cup \mathit{fn}(P)
      = \mathit{fn}((Q, P)) 
      \)
    \end{quote}
  \item \emph{Case \((P, (Q, R)) \equiv ((P, Q), R)\)} :
    \begin{quote}
      \(
      \mathit{fn}((P, (Q, R)))
      = \mathit{fn}(P) \cup \mathit{fn}(Q) \cup \mathit{fn}(R)
      = \mathit{fn}(((P, Q), R))
      \)
    \end{quote}
  \item \emph{Case \(P \equiv P' \Rightarrow (P, Q) \equiv (P', Q)\)} :
    \begin{quote}
      We have \(\mathit{fn}(P) = \mathit{fn}(P')\) if \(P \equiv P'\)
      by induction hypothesis.

      Therefore,
      \(
      \mathit{fn}((P, Q))
      = \mathit{fn}(P) \cup \mathit{fn}(Q)
      = \mathit{fn}(P') \cup \mathit{fn}(Q)
      = \mathit{fn}((P', Q))
      \)
    \end{quote}
  \item \emph{Case \(P \equiv P' \Rightarrow \nu X.P \equiv \nu X.P'\)} :
    \begin{quote}
      We have \(\mathit{fn}(P) = \mathit{fn}(P')\) if \(P \equiv P'\)
      by induction hypothesis.
      
      Therefore,
      \(
      \mathit{fn}(\nu X.P)
      = \mathit{fn}(P) \setminus \{X\}
      = \mathit{fn}(P') \setminus \{X\}
      = \mathit{fn}(\nu X.P')
      \)
    \end{quote}
  \item \emph{Case \(\nu X.(X \bowtie Y, P) \equiv \nu X.P[Y / X]\)
  where \(X \in \mathit{fn}(P) \lor Y \in \mathit{fn}(P)\)} :
    \begin{quote}
      Since \(X \in \mathit{fn}(P) \lor Y \in \mathit{fn}(P)\),
      \(\mathit{fn}(P[Y/X])\) should be equivalent with
      \((\mathit{fn}(P) \setminus \{X\}) \cup \{Y\} \)
      .
      
      Therefore, \\
      \(\begin{array}{l}
      \mathit{fn}(\nu X.(X \bowtie Y, P))
      = (\mathit{fn}(P) \cup \{Y\}) \setminus \{X\} \\
      = ((\mathit{fn}(P) \setminus \{X\}) \cup \{Y\}) \setminus \{X\}
      = \mathit{fn}(P[Y/X]) \setminus \{X\} 
      = \mathit{fn}(\nu X.P[Y/X])
      \end{array}\)
    \end{quote}
  \item \emph{Case \(\nu X.\nu Y.X \bowtie Y \equiv \zero\)} :
    \begin{quote}
      \(
      \mathit{fn}(\nu X.\nu Y.X \bowtie Y)
      = \mathit{fn}(X \bowtie Y) \setminus \{X, Y\}
      = \{X, Y\} \setminus \{X, Y\}
      = \emptyset
      = \mathit{fn}(\zero)
      \)
    \end{quote}
  \item \emph{Case \(\nu X.\zero \equiv \zero\)} :
    \begin{quote}
      \(
      \mathit{fn}(\nu X.\zero)
      = \emptyset \setminus \{X\}
      = \emptyset
      = \mathit{fn}(\zero)
      \)
    \end{quote}
  \item \emph{Case \(\nu X.\nu Y.P \equiv \nu Y.\nu X.P\)} :
    \begin{quote}
      \(
      \mathit{fn}(\nu X.\nu Y.P)
      = (\mathit{fn}(P) \setminus \{Y\}) \setminus \{X\}
      = (\mathit{fn}(P) \setminus \{X\}) \setminus \{Y\}
      = \mathit{fn}(\nu Y.\nu X.P)
      \)
    \end{quote}
  \item \emph{Case \(\nu X.(P,Q) \equiv (\nu X.P,Q)\)
  where \(X \notin \mathit{fn}(Q)\)} :
    \begin{quote}
      \(\begin{array}{ll}
      \mathit{fn}(\nu X.(P,Q)) &\\
      = (\mathit{fn}(P) \cup \mathit{fn}(Q)) \setminus \{X\} &\\
      = (\mathit{fn}(P) \setminus \{X\}) \cup \mathit{fn}(Q)
      & \because X \notin \mathit{fn}(Q) \\
      = \mathit{fn}((\nu X.P,Q)) &\\
      \end{array}\)
    \end{quote}
  \end{itemize}
\end{proof}

\subsection{Reduction relation}

We define the reduction relation $\longrightarrow$ on processes as the minimal relation satisfying the rules in \Cref{table:trans}.
We proved that the free links of a process, decrease monotonously through a reduction in \Cref{th:freelinks-of-trans}.
That is, we won't face a new free link after the reduction, which is a very natural property of a calculus model.

\begin{figure}
  \hrulefill
  \begin{center}
    \begin{tabular}{ ll } 
      (R1) & \(\dfrac{P \longrightarrow P'}{(P, Q) \longrightarrow  (P', Q)} \)
      \vspace{1em}\\
      (R2) & \(\dfrac{P \longrightarrow P'}{\nu X.P \longrightarrow  \nu X.P'} \)
      \vspace{1em} \\
      (R3) & \(\dfrac{Q \equiv P \hspace{16pt} P \longrightarrow P' \hspace{16pt} P' \equiv Q'}{Q \longrightarrow Q'} \)
      \vspace{0.5em} \\
      (R4) & \( (P, (P \vdash Q)) \longrightarrow (Q, (P \vdash Q)) \) \\
    \end{tabular}
  \end{center}
  \hrulefill
  \caption{Reduction relation on Flat HyperLMNtal processes}
  \label{table:trans}
\end{figure}

\begin{example}{Non-injective matching}{}
  Can the rule 
  \[
  (p(X, Y) \vdash q(X, Y))
  \]
  rewrite an atom $p(X, X)$ ?

  More precisely, can the process 
  \[
  (p(X, X), (p(X, Y) \vdash q(X, Y)))
  \]
  reduces to something?

  The rule cannot be $\alpha$-converted to the form
  \[ (p(X, X) \vdash \ldots) \]

  However, the atom $p(X, X)$ can be converted to $\nu Y.(Y \bowtie X, p(X, Y))$ using (E7) and \Cref{lem:absorb-link-creation}.

  Therefore, it can be rewritten as
  \begin{align*}
    & (p(X, X), (p(X, Y) \vdash q(X, Y))) \\
    & \equivby{\Cref{lem:absorb-link-creation}}
    \nu Y.(p(X, X), (p(X, Y) \vdash q(X, Y)))\\
    & \equivby{E7}
    \nu Y.(Y \bowtie X, (p(X, Y), (p(X, Y) \vdash q(X, Y))))\\
    & \longrightarrow
    \nu Y.(Y \bowtie X, (q(X, Y), (p(X, Y) \vdash q(X, Y))))\\
    & \equivby{E7}
    \nu Y.(q(X, X), (p(X, Y) \vdash q(X, Y)))\\
    & \equivby{\Cref{lem:absorb-link-creation}}
    (q(X, X), (p(X, Y) \vdash q(X, Y)))
  \end{align*}

  As the above, we can match non-injective free links using congruence rule on the link fusion. 
\end{example}

\begin{theorem}{The set of the free links of a process through reduction}{freelinks-of-trans}
  \[\mathit{fn}(P) \supseteq \mathit{fn}(Q) \mbox{ if } P \longrightarrow Q\]
\end{theorem}

\begin{proof}{}{}
  We prove this by structural induction on processes (and reduction relations).

  \begin{itemize}[label=\(\lozenge\), itemsep=2ex]
  \item \emph{Case \(\dfrac{P \longrightarrow P'}{(P, Q) \longrightarrow  (P', Q)}\)} :
    \begin{quote}
      We have \(\mathit{fn}(P) \supseteq \mathit{fn}(P')\) if \(P \longrightarrow P'\)
      by induction hypothesis.

      Therefore,
      \(
      \mathit{fn}((P, Q))
      = \mathit{fn}(P) \cup \mathit{fn}(Q)
      \supseteq \mathit{fn}(P') \cup \mathit{fn}(Q)
      = \mathit{fn}((P', Q))
      \)
    \end{quote}
  \item \emph{
  Case \(\dfrac{P \longrightarrow P'}{\nu X.P \longrightarrow  \nu X.P'}\)} :
    \begin{quote}
      We have
      \(\mathit{fn}(P) \supseteq \mathit{fn}(P')\) where \(P \longrightarrow P'\)
      by induction hypothesis.

      Therefore,
      \(
      \mathit{fn}((P, Q))
      = \mathit{fn}(P) \setminus \{X\}
      \supseteq \mathit{fn}(P') \setminus \{X\}
      = \mathit{fn}((P', Q))
      \)
    \end{quote}
  \item \emph{Case \(\dfrac{P \longrightarrow P'}{\nu X.P \longrightarrow  \nu X.P'}\)} :
    \begin{quote}
      We have
      \(\mathit{fn}(P) \supseteq \mathit{fn}(P')\) where \(P \longrightarrow P'\)
      by induction hypothesis.

      Therefore,
      \(
      \mathit{fn}(\nu X.P)
      = \mathit{fn}(P) \setminus \{X\}
      \supseteq \mathit{fn}(P') \setminus \{X\}
      = \mathit{fn}(\nu X.P')
      \)
    \end{quote}
  \item \emph{Case \(
  \dfrac{Q \equiv P \hspace{16pt} P \longrightarrow P' \hspace{16pt} P' \equiv Q'}
        {Q \longrightarrow Q'}
        \)} :
    \begin{quote}
      We have
      \(\mathit{fn}(P) \supseteq \mathit{fn}(P')\) where \(P \longrightarrow P'\)
      by induction hypothesis and
      \(\mathit{fn}(Q) = \mathit{fn}(P)\) where \(Q \equiv P\) and
      \(\mathit{fn}(P') = \mathit{fn}(Q')\) where \(P' \equiv Q'\) 
      by \Cref{th:freelinks-of-equiv}.
      
      Therefore,
      \(
      \mathit{fn}(Q)
      = \mathit{fn}(P)
      \supseteq \mathit{fn}(P')
      = \mathit{fn}(Q')
      \)
    \end{quote}
  \item \emph{Case \((P, (P \vdash Q)) \longrightarrow (Q, (P \vdash Q))\)} :
    \begin{quote}
      We have \(\mathit{fn}(P) \supseteq \mathit{fn}(Q)\) 
      by the second syntactical condition of a rule.

      Therefore,
      \(
      \mathit{fn}((P, (P \vdash Q)))
      = \mathit{fn}(P)
      \supseteq \mathit{fn}(Q) 
      = \mathit{fn}((Q, (P \vdash Q)))
      \)
    \end{quote}
  \end{itemize}
\end{proof}

\section{HyperLMNtal semantics}
\label{sec:hyperlmntal-semantics}

In this section, we will briefly look at the extended syntax and semantics of Flat HyperLMNtal to suit our current implementation, in which introduce a hierarchy of processes and localization of rules; membranes.

\subsection{Abstract syntax and semantics of HyperLMNtal}
\Cref{table:hyperlmntal-impl-syntax} gives the syntax of HyperLMNtal based on the current implementation.
The newly added syntax is denoted with a dagger (\mydagger).
The reserved names are \(=\) and \verb|><|.
Notice the hyperlink is defined in a very similar way to an atom.
This is because that the hyperlinks are implemented as an atom with a name \verb|!| and a hidden pointer to an id (this corresponds with the name of the hyperlink in the proposed syntax and semantics) in our current system.
Therefore, hyperlinks can be connected to normal link
and a port of an atom that is connected to a hyperlink can be matched with a normal link.
For example, 
\begin{lstlisting}[language=lmntal]
init.
init :- a(!H).
a(X) :- b(X).
\end{lstlisting}
will be reduced to \texttt{b(!H)} and the remaining rules in our implementation.
We allow applying the same abbreviation rules that has used in LMNtal.
Also, we can apply the abbreviation schemes for atoms to hyperlinks.
For example, \(!H_1(X), X \:\verb|><|\: Y, !H_2(Y)\) can be abbreviated as \(!H_1 \:\verb|><|\: !H_2\).

\Cref{table:hyperlmntal-impl-free-names} gives the set of the free hyperlink names of a HyperLMNtal process.
Notice that the namespaces of normal links and hyperlinks are completely different.
When we say ``$X$ occurs free in $P$'', then it is the normal link that occurs free (once) in a process $P$ and we have not mentioned about the hyperlinks in that process at all.

The syntactic conditions of both LMNtal and the proposed Flat HyperLMNtal must be satisfied: the number of the end points of a normal link should be kept in at most two, a rule \(P \means Q\) must satisfy \(\mathit{fn}(P) \supseteq \mathit{fn}(Q)\), etc.

\Cref{table:hyperlmntal-impl-equiv} shows the structural congruence on HyperLMNtal processes.
The newly added congruence relation is denoted with a dagger (\mydagger).
Here, $P[Y/X]$ is a normal link substitution which replaces a normal link $X$ with a normal link $Y$ and
\(P \angled{H_2/H_2}\) is a hyperlink substitution, which is defined in the \Cref{table:hyperlmntal-impl-hyperlink-substitution}.
Notice that \(\alpha\)-conversion might be required in hyperlink substitutions as we have already discussed in the Flat HyperLMNtal semantics.

\Cref{table:hyperlmntal-impl-trans} shows the reduction relation on HyperLMNtal process.
The newly added reduction relation is denoted with a dagger (\mydagger).
The process contexts can also match a hyperlink but a link creation.
That is, the copying of a process won't divide a hyperlink (in other words, create a new hyperlink).

\begin{figure}
  \hrulefill
  \begin{center}
    \begin{tabular}{rlclll} 
      (Process) & $P$ &$::=$& $\zero$     & & Null \\
      &&$|$& \(p (X_1, \ldots, X_m)\) & \(m \geq 0\) & Atom \\
      &&$|$& \((P, P)\)               & & Molecule \\
      &&$|$& \(\{ P \}\)              & & Cell \\
      &&$|$& \(\nu H.P\)              & & Hyperlink creation \mydagger \\  
      &&$|$& \(!H(X)\)                & & Hyperlink \mydagger \\  
      &&$|$& \((T \means T)\)         & & Rule
      \vspace{1em}\\                                                     
      (Process Template) & $T$ &$::=$& $\zero$ & & Null \\
      &&$|$& \(p (X_1, \ldots, X_m)\) & \(m \geq 0\) & Atom \\
      &&$|$& \((T, T)\)               & & Molecule \\
      &&$|$& \(\{ T \}\)              & & Cell \\
      &&$|$& \(\nu H.T\)              & & Hyperlink creation \mydagger \\ 
      &&$|$& \(!H(X)\)                & & Hyperlink \mydagger\\  
      &&$|$& \((T \means T)\)         & & Rule \\
      &&$|$& \(\rulevar p\)           & & Rule context \\
      &&$|$& \(\procvar p\,[X_1, \ldots, X_m|A]\) & \(m \geq 0\) & Process context
      \vspace{1em}\\
      (Residual) & $A$ &$::=$& \([\,]\) & & Empty \\
      &&$|$& \(\fvstar X\)      & & Bundle         
%
    \end{tabular}
  \end{center}
  \hrulefill
  \caption{Syntax of HyperLMNtal based on the current implementation}
  \label{table:hyperlmntal-impl-syntax}
\end{figure}

\begin{figure}
  \hrulefill
  \[
  \begin{aligned}
    \mathit{fn}(\mathbf{0})          &= \emptyset \\
    \mathit{fn}(p(X_1, \ldots ,X_m)) &= \emptyset \\
    \mathit{fn}((P, Q))              &= \mathit{fn}(P) \cup \mathit{fn}(Q) \\
    \mathit{fn}(\{P\})               &= \mathit{fn}(P) \\
    \mathit{fn}(\nu H.P)             &= \mathit{fn}(P) \setminus \{!H\} \\
    \mathit{fn}(!H(X))               &= \{!H\} \\
    \mathit{fn}((P \vdash Q))        &= \emptyset
  \end{aligned}
  \]
  \hrulefill
  \caption{Free hyperlink names of a HyperLMNtal process based on the current implementation}
  \label{table:hyperlmntal-impl-free-names}
\end{figure}

\begin{figure}
  \hrulefill
  \begin{center}
    \begin{tabular}{l}
      \( \zero \angled{H_2/H_1}  \overset{def}{=} \zero \)
      \vspace{1em}\\ 
      \( p(X_1, \ldots, X_m) \angled{H_2/H_1}
      \overset{def}{=} p(X_1, \ldots, X_m) \)
      \vspace{1em}\\ 
      \((P, Q) \angled{H_1/H_2}
      \overset{def}{=} (P \angled{H_1/H_2} , Q \angled{H_1/H_2} )\)
      \vspace{1em}\\
      \(
      (\nu H.P) \angled{H_1/H_2}  
      \) \\
      \hspace{2em}
      \(
      \overset{def}{=}
      \left\{
      \begin{array}{ll}
        \nu H.P & \mbox{if } H = H_1 \\
        \nu H.P \angled{H_1/H_2} 
        & \mbox{if } H \neq H_1 \land H \neq H_2 \\
        \nu H_3.(P \angled{H_3/H} ) \angled{H_1/H_2} 
        & \mbox{if } \mbox{\begin{minipage}[t]{11.5em}
            \(H \neq H_1 \land H = H_2\) \\
            \(\land\ !H_3 \notin \mathit{fn}(P) \land H_3 \neq H_2\)
          \end{minipage}
        }
      \end{array}
      \right.
      \) \vspace{1em}\\
      \( !H(X) \angled{H_1/H_2} \overset{def}{=} 
      \left\{
      \begin{array}{ll}
        !H_2(X) & \mbox{if } H = H_1 \\
        !H(X)   & \mbox{if } H \neq H_1
      \end{array}
      \right.
      \) \vspace{1em}\\ 
      \( (P \vdash Q) \angled{H_1/H_2}  \overset{def}{=} (P \vdash Q) \) \\
    \end{tabular}
  \end{center}
  \hrulefill
  \caption{Hyperlink substitution of HyperLMNtal based on the current implementation}
  \label{table:hyperlmntal-impl-hyperlink-substitution}
\end{figure}

\begin{figure}
  \hrulefill
  \begin{center}
     \scalebox{0.91}{
      \begin{tabular}{lll} 
        (E1) & \(\mathbf{0}, P \equiv P\) \\
        (E2) & \(P, Q \equiv Q, P\) \\
        (E3) & \(P, (Q, R) \equiv (P, Q), R\) \\
        (E4) & \(P \equiv P[Y/X]\) 
        & if $X$ if a local link of $P$ \\
        (E5) & \(P \equiv P' \Rightarrow P, Q \equiv P', Q\) \\
        (E6) & \(P \equiv P' \Rightarrow \{P\} \equiv \{P'\}\) \\
        (E7) \mydagger & \(P \equiv Q \Rightarrow \nu H.P \equiv \nu H.Q\) \\
        (E8) & \(X = X \equiv \zero \)                  \\
        (E9) & \(X = Y, P \equiv P[Y/X]\)
        & if $P$ is an atom and $X$ occurs free in $P$ \\
        (E10) & \(\{X = Y, P\} \equiv X = Y, \{P\}\)
        & if exactly one of $X$ and $Y$ occurs free in $P$ \\
        (E11) \mydagger &
        \(\nu H_1.(!H_1 \:\textnormal{\texttt{><}}\: !H_2, P) \equiv \nu H_1.P \angled{H_1/H_2}\)
        & where \(!H_1 \in \mathit{fn}(P)\ \lor\ !H_2 \in \mathit{fn}(P)\) \\
        (E12) \mydagger & \(\nu H_1.\nu H_2. (!H_1 \:\textnormal{\texttt{><}}\: !H_2) \equiv \zero\) \\
        (E13) \mydagger & \(\nu H.\zero \equiv \zero\)\\
        (E14) \mydagger & \(\nu H_1.\nu H_2.P \equiv \nu H_2.\nu H_1.P\)\\
        (E15) \mydagger & \(\nu H.\{P\} \equiv \{\nu H.P\}\) \\
        (E16) \mydagger & \(\nu H.(P, Q) \equiv (\nu H.P, Q)\)
        & where \(!H \notin \mathit{fn}(Q)\)\\
      \end{tabular}
      }
  \end{center}
  \hrulefill
  \caption{Structural congruence on HyperLMNtal processes based on the current implementation}
  \label{table:hyperlmntal-impl-equiv}
\end{figure}

\begin{figure}
  \hrulefill
  \begin{center}
    \renewcommand{\arraystretch}{1.2}
    \begin{tabular}{ l l l } 
      (R1) & \(\dfrac{P \longrightarrow P'}{P, Q \longrightarrow  P', Q}\) & 
      \vspace{1em} \\
      (R2) & \(\dfrac{P \longrightarrow P'}{\{P\} \longrightarrow  \{P'\}} \) & 
      \vspace{1em} \\
      (R3) \mydagger
      & \(\dfrac{P \longrightarrow P'}{\nu H.P \longrightarrow \nu H.P'} \) &
      \vspace{1em} \\
      (R4) &
      \( \dfrac{Q \equiv P \hspace{1em} P \longrightarrow P' \hspace{1em} P' \equiv Q'}{Q \longrightarrow Q'} \) & 
      \vspace{0.5em} \\
      (R5) & \( \{X = Y, P\} \longrightarrow X = Y, \{P\} \) &
      if $X$ and $Y$ occur free in \(\{X = Y, P\}\) \\
      (R6) & \(X = Y, \{P\} \longrightarrow \{X = Y, P\}\) &
      if $X$ and $Y$ occur free in $P$ \\
      (R7) & \(T\theta, (T \vdash U) \longrightarrow U\theta, (T \vdash U)\) & \\
    \end{tabular}
  \end{center}
  \hrulefill
  \caption{Reduction relation on HyperLMNtal processes based on the current implementation}
  \label{table:hyperlmntal-impl-trans}
\end{figure}

\subsection{Translation to the concrete syntax}

Any rule $(P \vdash Q)$ can be rewritten as 
\((\nu X_1. \cdots. \nu X_n. P' \means \nu Y_1. \cdots. \nu Y_m. Q')\)
\((n \geq 0, m \geq 0)\)
, a structurally congruent rule, where no link creation appears in $P'$ and $Q'$
and \(Y_i \notin \mathit{fn}(P')\) .

Then, rewrite it as 

\[
(P'
\means
\verb|num|(!X_1) \verb| =:= | N_1, \cdots, \verb|num|(!X_n) \verb| =:= | N_n
\verb+ | + Q'')
\]

where $N_i$ is a number of the occurrence of the hyperlink $X_i$ in $P$ (except those have normal links connected to the fusion atom) and
$Q''$ is a inductively translated concrete syntax form of $Q'$.
The current implementation does not support to write as \verb|!H(X)|.
We have to write all the hyperlinks inside of the ports of the atoms that are connected to.
For example, \verb|a(!H)| or \verb|!H = X|.

\subsection{Further observation}
Basically, implementing a hyperlink as a special atom was a great decision.
It allows both implementation and semantics relatively easy.
However, it results with a little annoying observation when we use hyperlinks with membranes.
That is, hyperlinks can be separated with a membrane.
For example,the following program 
\begin{lstlisting}[language=lmntal]
init.
init :- {a(!H)}.
{a(X), $p[X]} :- a(X), {$p[X]}. 
a(!H) :- b.
\end{lstlisting}
will be reduced to \lstinline[language=lmntal]|a(X), {!H(X)}| and the forth rule \lstinline[language=lmntal]|a(!H) :- b| would not be applied since the hyperlink \verb|!H(X)| is not at the same hierarchy with the rule.

Both our current implementation and the proposed semantics support this idea.
Thus we may do not need to worry about this.
However, if we think of a hyperlink as an extension of a normal link, this is certainly a strange behavior.
Well, for the case above, we can apply the rule 
\lstinline[language=lmntal]|a(X), {!H = X} :- b|
to the program and if we do so, then the program will be reduced to
\lstinline[language=lmntal]|b|.
But still, we do not think this is a natural behavior based on our intuition.

Therefore, we may like to modify our implementation to let all the hyperlinks to be in the same hierarchy with the atoms connected to them at every transition state.
Then we should modify our semantics to allow this.
This can be done in a very simple way; just adding the following congruence rule, which allows a hyperlink to move anywhere to go along with the connected atom.

\[
\begin{array}{ll}
  (E18) & \{!H(X), P\} \equiv\ !H(X), \{P\}
\end{array}
\]

\chapter{Syntax/Semantics of G-Machine}
\label{cha:gm}
\par
G-Machine is a virtual machine for non-strict functional languages which uses \emph{graph-reduction}.
Since LMNtal is a language that has focused on graph rewriting, it is important to check how easily we can achieve this (or if not, why we cannot) .
In this chapter, we will introduce the functional language we implemented, core language, and the implemented G-Machine based on \cite{implementingfl}.

\section{General introduction to lazy evaluation and graph-reduction}
\label{sec:graph-reduction-intro}

Strict evaluation strategy requires all the arguments to be evaluated before performing beta-reduction (application).
Where non-strict doesn't require this. 
Thus, the easiest way to accomplish beta reduction in a non-strict evaluation is just replacing the variable with the corresponding applied term: Call By Name.
For example, \((\lambda x.x + x) (1 + 2)\) will be reduced to \((1 + 2) + (1 + 2)\).
However, this is not efficient since we should calculate \(1 + 2\) twice.
To improve efficiency, we can consider the reducing expressions as ``graph'' not just tree which allows sharing of expressions.
For example, the former expression can be represented as the graph in \Cref{thunk}.
Notice that the expression \(1 + 2\) (blue-colored subgraph) is shared.
By sharing expressions and evaluating once for a each, we can reduce the steps of calculation process.

\begin{figure}
  \begin{center}
    \includegraphics[width=0.2\textwidth]{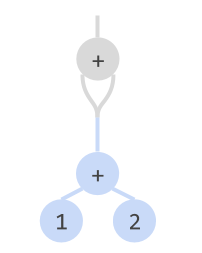}
  \end{center}
  \caption{Graph representation of an expression}
  \label{thunk}
\end{figure}

Though not-hyper Flat LMNtal can handle graph, the end points of links in LMNtal are restricted to be just 2. Therefore, it is not easy to represent sharing.
Using membrane is one solution to this problem but it is rather hard to write these kinds of programs with membranes.
Instead, HyperLMNtal is a very suitable language to deal with this.
Since hyperlinks are allowed to connect one or more ports.

\section{An introduction to G-Machine}
\label{sec:gm-intro}

The components of G-Machine is a code, stack, dump heap and a global environment

Heap is a multiset of nodes.
Where nodes are basically one of ``application node'', ``global (i.e. super combinator/function) node'' and ``a node for a primitive value (e.g. integer)''.
Notice that the global node is completely different from the global environment.
These nodes form graph since the application node has addresses (i.e. directed edges/pointers) to its argument nodes.
Figure.~\ref{NodesOfGM} shows the example of nodes in G-Machine.
Here, the label for the application node is written as \verb|@|.
Each nodes might have or not have edges whose head is the node (the gray-colored arrow).

\begin{figure}
  \begin{center}
    \includegraphics[width=0.7\textwidth]{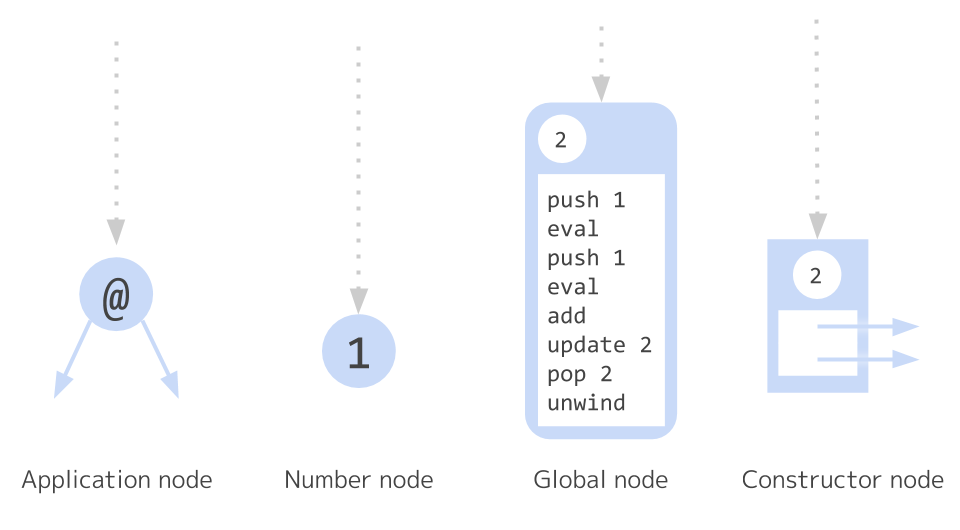}
  \end{center}
  \caption{Example of nodes in G-Machine}
  \label{NodesOfGM}
\end{figure}

Stack is a stack of an addresses to nodes of a chained applications.
The process of evaluating application is called \emph{unwinding},
which is basically done by pushing the address of the right child of application node on top of the stack.

Dump is a stack of the stack.
This is used to evaluate an expression pointed by the top of the stack until it reaches to Weak Head Normal Form.
To accomplish this, we have to store the current stack and the code to dump.
This is mainly for evaluating an expression with a primitive operator.
For example, we can not add expressions but numbers:
we have to evaluate the both arguments until they are to be numbers.

Code is a list of instructions.
The evaluation is done by following these instructions.
Codes are obtained by compiling (top-level) functions and kept in global nodes.

Global environment is a function from the name of the top-level functions to the addresses of the corresponding global nodes.

\section{Syntax of the core language}
\label{sec:gm-corelang-syntax}

The core language is an lazy functional language.
The syntax of the core language is given in \Cref{table:gm-fl-syntax}.
This is an abstract syntax.
Where \(\mathit{var}\) is a variable (identifier) and \(\mathit{num}\) is an integer.
A more concrete syntax can be obtained from \cite{implementingfl}.

A core language program is consisted of a list of functions.
We call them super-combinators.
Since these functions are allowed to have a free variable, a variable which is not a locally defined variable, in their body expressions (i.e. the right-hand side of the ``='') these are certainly not combinators and thus not super-combinators technically.
However, in the reference \cite{implementingfl}, these are noted as super-combinators, therefore we shall call them so, like wisely.

It is not designed to be written by programmers directly, but is designed to be suitable for an input of the compiler.
For example, this language lacks lambda abstraction (anonymous function), which can be achieved by ``program transformation'' technique, lambda lifting.
And user defined data types must be written as 
\lstinline'Pack{2, 2} 1 Pack{1, 0}'
instead of like
\lstinline'Cons 1 Nil'
(with a predefinition
\lstinline'data List a = Cons List a | Nil'
). 
The first argument of \verb|Pack| is the \emph{tag} of the data type.
For example, we may assign 1 for \texttt{Nil} and 2 for \texttt{Cons}.
The second argument is an arity of the data type (the number of the arguments).
In the former example, we need 0 to be the second argument of the \verb|Pack{1, 0}| since \texttt{Nil} should have no argument and 2 to be the second argument of the \verb|Pack{2, 2}| since the \texttt{Cons} should have 2 arguments.
The decomposition of these kinds of data structures in the case expression are done with the tags we assigned for each types.
For example, \lstinline'case list of Nil -> ... | Cons h t -> ...' can be done with
\lstinline'case list of <1> -> ... ; <2> h t -> ...' in the core language.

\begin{figure}
  \hrulefill
  \begin{center}
    \scalebox{0.85}{
      \begin{tabular}[t]{lrcll} 
        (Program) &$\mathit{program}$ &$::=$& 
        \(sc_1 ; \ldots ; sc_n\) & \(n \geq 1\)
        \vspace{0.5em}\\
        (Super Combinator) &$\mathit{sc}$ &$::=$&
        \(\mathit{var}\ \mathit{var_1} \ldots \mathit{var_n}
        \textnormal{\texttt{ = }} \mathit{expr}\)
        & \(n \geq 0\)
        \vspace{0.5em}\\
        (Expression) &$\mathit{expr}$ &$::=$& \(\mathit{var}\) & Variable \\
        &&$|$& \(\mathit{expr}\ \mathit{expr}\) & Application \\
        &&$|$& \(\mathit{expr}\ \mathit{binop}\ \mathit{expr}\)
        & Infix binary application \\
        &&$|$& \(\textnormal{\texttt{let }} \mathit{defns}
        \textnormal{\texttt{ in }} \mathit{expr}\)
        & Local definition \\
        &&$|$& \(\mathtt{let\ rec\ } \mathit{defns} \mathtt{\ in\ } \mathit{expr}\)
        & Local recursive definition \\
        &&$|$& \(\mathtt{case\ } \mathit{expr} \mathtt{\ of\ } \mathit{alts}\)
        & Case expression \\
        &&$|$& \(\mathtt{if\ } \mathit{expr}
        \mathtt{\ then\ } \mathit{expr} \mathtt{\ else\ } \mathit{expr}\)
        & If expression \\
        &&$|$& \(\mathit{num}\) & Number \\
        &&$|$&
        \(\textnormal{\texttt{Pack\{}} \mathit{num}
        \textnormal{\texttt{,}} \mathit{num}
        \textnormal{\texttt{\}}}\)
        & Constructor
        \vspace{0.5em}\\
        (Definitions) &$\mathit{defns}$ &$::=$&
        \(\mathit{defn}_1; \ldots; \mathit{def}_n \)
        & \(n \geq 1\) \\
        &$\mathit{defn}$ &$::=$&
        \(\mathit{var} \textnormal{\texttt{ = }} \mathit{expr}\)
        \vspace{0.5em}\\
        (Alternatives) &$\mathit{alts}$ &$::=$&
        \(\mathit{alt}_1 ; \ldots ; \mathit{alt}_n\)
        & \(n \geq 1\) \\
        &$\mathit{alt}$ &$::=$&
        \(\textnormal{\texttt{<}} \mathit{num} \textnormal{\texttt{> }}
        \mathit{var}_1 \ldots \mathit{var}_n
        \textnormal{\texttt{ -> }} \mathit{expr} \)
        & \(n \geq 0\) 
        \vspace{0.5em}\\
        (Binary Operators) &$\mathit{binop}$ &$::=$&
        \(\mathit{arithop}\ |\ \mathit{relop}\)
        & \\
        &$\mathit{arithop}$ &$::=$& \(
        \textnormal{\texttt{ + }}
        | \textnormal{\texttt{ - }}
        | \textnormal{\texttt{ * }}
        | \textnormal{\texttt{ / }} \)
        & Arithmetic \\
        &$\mathit{relop}$ &$::=$& \(
        \textnormal{\texttt{ < }}
        | \textnormal{\texttt{ /= }} \)
        & Comparison \\
      \end{tabular}
    }
  \end{center}
  \hrulefill
  \caption{Syntax of the core language}
  \label{table:gm-fl-syntax}
\end{figure}

\section{Syntax of G-Machine}
\label{sec:gm-code}

A set of instructions of G-Machine in given in \Cref{table:gm-code}.
where \(i\) is an integer \(\mathit{instr}\) is an instruction \(\mathit{name}\) is a name of top-level super-combinator and \(\mathit{code}\) is a code (a list of instructions).
The operation semantics of these instructions will be given in \Cref{sec:gm-op-sem}.

\begin{figure}
  \hrulefill
  \begin{center}
    \renewcommand{\arraystretch}{1.2}
    \begin{tabular}[t]{lrcll}
      (Code) & \(code\)
      &$::=$&
      \([\mathit{instr_1}, \ldots, \mathit{instr_n}]\) & \(n \geq 0\)\\
      \\
      (Instruction) & $\mathit{instr}$ &$::=$& \texttt{Slide} \(\mathit{num}\) \\
      &&$|$& \texttt{Alloc} \(\mathit{num}\) \\
      &&$|$& \texttt{Update} \(\mathit{num}\) \\
      &&$|$& \texttt{Pop} \(\mathit{num}\) \\
      &&$|$& \texttt{Unwind} \\
      &&$|$& \texttt{PushGlobal} \(\mathit{var}\) \\
      &&$|$& \texttt{PushInt} \(\mathit{num}\) \\
      &&$|$& \texttt{Push} \(\mathit{num}\) \\
      &&$|$& \texttt{Mkap} \\
      &&$|$& \texttt{Eval} \\
      &&$|$& \texttt{Cond} \(code\ code\) \\
      &&
      $|$& \texttt{Add}
      $|$\ \texttt{Sub}
      $|$\ \texttt{Mul}
      $|$\ \texttt{Div} \\
      &&
      $|$& \texttt{Lt}
      $|$\ \texttt{Neq} \\
      &&$|$& \texttt{Pack}
      \(\mathit{num}\) \(\mathit{num}\) \\
      &&$|$& \texttt{Casejump}
      \(\begin{bmatrix}
      \mathit{num}_1 \textnormal{\texttt{ -> }} \mathit{code}_1\\
      \vdots \\
      \mathit{num}_n \textnormal{\texttt{ -> }} \mathit{code}_n\\
      \end{bmatrix}\) & \(n \geq 1\)\\
      &&$|$& \texttt{Split} \(\mathit{num}\) \\
    \end{tabular}        
  \end{center}
  \hrulefill
  \caption{Instructions of G-Machine}
  \label{table:gm-code}
\end{figure}

\begin{figure}
  \hrulefill
  \begin{center}
    \renewcommand{\arraystretch}{1.2}
    \begin{tabular}{lrclll} 
      (Node) & $\mathit{node}$ &$::=$& \texttt{NNum} \(\mathit{num}\) && Number \\
      &&$|$& \texttt{NGlobal} \(\mathit{num}\) \(\mathit{code}\)
      && Global \\
      &&$|$& \texttt{NApp} \(a\) \(a\) && Application \\
      &&$|$& \texttt{NConstr} \(\mathit{num}\) \([a_1, \ldots, a_n]\)
      &\ \(n \geq 0\) & Constructor \\
      &&$|$& \texttt{NInd} \(a\) && Indirection \\
    \end{tabular}
  \end{center}
  \hrulefill
  \caption{Nodes of G-Machine}
  \label{table:gm-node}
\end{figure}

\begin{figure}
  \hrulefill
  \begin{center}
    \renewcommand{\arraystretch}{1.2}
    \begin{tabular}{lrcll} 
      (Stack)
      & \(\mathit{stack}\)
      &$::=$& \([a_1, \ldots, a_n]\) &\ \(n \geq 0\) \\
      (Dump)
      & \(\mathit{dump}\)
      &$::=$& \([ \langle \mathit{stack}_1, \mathit{code}_1 \rangle,
        \ldots, 
        \langle \mathit{stack}_n, \mathit{code}_n \rangle ]\)
      &\ \(n \geq 0\) \\
      (Global Environment)
      & \(\mathit{env}\)
      &$::=$& \(\{\mathit{var}_1 \mapsto a_1, \ldots, \mathit{var}_n \mapsto a_n\}\)
      &\ \(n \geq 1\) \\
      (Heap)
      & \(\mathit{heap}\)
      &$::=$& \(\{a_1 \mapsto \mathit{node}_1, \ldots, a_n \mapsto \mathit{node}_n\}\)
      &\ \(n \geq 0\) \\
    \end{tabular}
  \end{center}
  \hrulefill
  \caption{The Stack, dump, global environment and the heap of G-Machine}
  \label{table:gm-components}
\end{figure}

\section{Compilation scheme of G-Machine}

The compilation to the codes are done with the following rules.
\newcommand{\dbra}[1]{\llbracket #1 \rrbracket\ }
Where \(l_1 \doubleplus l_2\) describes a concatenation of lists \(l_1\) and \(l_2\).

\begin{mdframed}

  \(\mathcal{SC}\) compiles a super-combinator \(f\ x_1 \ldots x_n = e\).
  The output is a G-Machine code for a global node.

  \[\mathcal{SC} \dbra{f\ x_1 \ldots x_n = e}
  = \mathcal{R} \dbra{e}
  [x_1 \mapsto 0, \ldots, x_n \mapsto n - 1]\ n
  \]
\end{mdframed}

\begin{mdframed}

  \(\mathcal{R}\) is for a compilation of the body expression of super-combinators $e$ (i.e. the right-hand side of $=$).
  Where $\rho$ is an environment (the function from the names of variables to their indices), $d$ is an arity of the supercombinator.
  After evaluating the $e$ following the code obtained from
  \(\mathcal{E} \dbra{e} \rho\ d \), we need to 
  (\romannumeral 1) \texttt{update} the pointer which previously pointed the supercombinator to point the address of the node newly obtained by evaluating the body expression of the supercombinator,
  (\romannumeral 2) \texttt{pop} out the addresses which correspond the local variables (the right-hand side of the supercombinator definition) and
  (\romannumeral 3) then perform evaluation, which is basically \texttt{unwinding}, as before.

  \[
  \mathcal{R} \dbra{e} \rho\ d
  = \mathcal{E} \dbra{e} \rho
  \doubleplus [ \mathtt{Update\ } d,
    \mathtt{Pop\ } d,
    \mathtt{Unwind}
  ]
  \]

\end{mdframed}

\begin{mdframed}
  \(\mathcal{E}\) compiles an expression in a strict-context for a better efficiency. 
  This is designed not to offend the general laziness of the core language.
  Whether the given expression is in the strict context or not can be determined by the derivations shown in p.126 of \cite{implementingfl}.

  \vspace{2em}
  If the given expression is an integer, we need a code to generate the number node with the integer (i.e. \texttt{push} \texttt{int}eger to the heap)

  \[\mathcal{E} \dbra{i} \rho = [\mathtt{PushInt}\ i]\]

  \vspace{2em}
  
  If the given expression is a local definition,
  we need to evaluate all of those expressions we have defined locally in a non-strict context since we may not require these
  and we can evaluate the body expression in a strict context because we must evaluate this invariably.
  Where \(\rho^{+i}\) is defined to meet \(\rho^{+n} x = (\rho x) + n\).
  This is because the distances of the variables from the top of the stack increases by pushing the each local definitions.
  Also, we need to clear the $n$ local definitions and \texttt{slide} the top of the stack at the very last process.

  \vspace{1em}
  
  {\renewcommand{\arraystretch}{1}
    \begin{center}
      \begin{tabular}{ rll } 
        \multicolumn{2}{l}{
          \(\mathcal{E} \dbra{
            \mathtt{let\ } x_1 = e_1; \ldots; x_n = e_n \mathtt{\ in\ } e
          } \rho\)
        }\\
        = & \(\mathcal{C}\dbra{e_1} \rho^{+ 0} \doubleplus \ldots \doubleplus\) \\
        & \(\mathcal{C}\dbra{e_n} \rho^{+ (n - 1)} \doubleplus \) \\
        & \(\mathcal{E}\dbra{e} \rho' \doubleplus [\mathtt{Slide\ } n]\) \\
        & \hspace{2em}
        where \(\rho' = \rho^{+ n}[x_1 \mapsto n - 1, \ldots, x_0 \mapsto 0]\)\\
      \end{tabular}
    \end{center}
  }
  \vspace{2em}
  
  The evaluation/execution and compilation of a local recursive definition is basically the same as the local (not-recursive) definition.
  The difference is that we need the environment with all the locally defined variables not just when we are compiling the body expression but also when we are compiling the right-hand side of definitions and
  that we firstly need to \texttt{allocate} dummy pointers on the top $n$ of the stack when we evaluate.
  After the evaluation of a local definition, the allocated dummy pointers are \texttt{updated} to point the node which is a result of the evaluation
  so as not to traverse these dummy pointers, which would cause ``segmentation fault'' or ``null pointer exception'' in such languages as ``c'' or ``Java''.

  \vspace{1em}

  {\renewcommand{\arraystretch}{1}
    \begin{center}
      \begin{tabular}{ rll } 
        \multicolumn{3}{l}{
          \(\mathcal{E} \dbra{
            \mathtt{let\ rec\ } x_1 = e_1; \ldots; x_n = e_n \mathtt{\ in\ } e
          } \rho\)
        }\\
        = & \([ \mathtt{alloc\ } n] \doubleplus\) \\
        & \(\mathcal{C}\dbra{e_1} \rho' \doubleplus [ \mathtt{Update\ } (n - 1)]
        \doubleplus \ldots \doubleplus\)
        \\
        & \(\mathcal{C}\dbra{e_n} \rho' \doubleplus [ \mathtt{Update\ } 0] \doubleplus
        \) \\
        & \(\mathcal{E}\dbra{e} \rho' \doubleplus [\mathtt{Slide\ } n] \) \\
        & \hspace{2em}
        where \(\rho' = \rho^{+ n}[x_1 \mapsto n - 1, \ldots, x_0 \mapsto 0]\)\\
      \end{tabular}
    \end{center}
  }
  \vspace{2em}

  If the given expression is a binary application with a primitive operator, we will evaluate the both arguments since we definitely need them.
  Therefore, the compilation of the arguments can be done with the same scheme \(mathcal{E}\), in a strict context.
  However, the environment of the first arguments (actually, either one of both) must be converted as \(\rho^{+ 1}\) since the location of the variables change by pushing the second argument to the top of the stack as we described in the compiling scheme of local definition.
  We need to perform the primitive operation, addition for this time, after evaluating the arguments, therefore we put the operator at the very last position of the generated code.
  The expressions with the other primitive operators (such as \verb| - |, \verb| * |, \verb| < |) are also dealt with the same manner.
  
  \[
  \mathcal{E} \dbra{e_0 + e_1} \rho =
  \mathcal{E}\dbra{e_1} \rho \doubleplus \mathcal{E}\dbra{e_0} \rho^{+ 1}
  \doubleplus [ \mathtt{Add} ]
  \]

  \vspace{2em}
  
  An expression with an unary primitive operator is compiled as the following.
  Notice that the argument can be compiled in a strict context as we have already discussed in the binary operator.
  
  \[
  \mathcal{E} \dbra{\mathtt{negate\ } e} \rho =
  \mathcal{E}\dbra{e} \rho \doubleplus [ \mathtt{Neg} ]
  \]

  \vspace{2em}

  If the given expression is a case expression \(\verb|case | e \verb| of | \mathit{alts}\), we can compile the \(e\) in a strict context since we will evaluate this definitely.
  After the evaluation of the \(e\), we will be executing the \texttt{Casejump} instruction, which has one or more codes and  simply looks at the tag of the structured data, the evaluated expression, and determines which code to execute in the execution phase.
  The compilation scheme of the \(\mathit{alts}\) will be described later.

  \[
  \mathcal{E} \dbra{\mathtt{case\ } e \mathtt{\ of\ } \mathit{alts}} \rho =
  \mathcal{E} \dbra{e} \rho
  \doubleplus [ \verb|Casejump | \mathcal{D} \dbra{\mathit{alts}}\ \rho ]
  \]
  
  \vspace{2em}

  If the given expression is a constructor and sufficient applications can be compiled with the following rule.
  If the application is in-sufficient then we should use program conversion technique to make it sufficient as described in \cite{implementingfl}.
  Notice that we should compile all the arguments in a non-strict context since we don't want to always evaluate the arguments of the structured data.
  For example, an infinite list can achieved by not evaluating the second argument of a \verb|Cons| somewhere.
  
  \[
  \mathcal{E} \dbra{\mathtt{Pack\{} t \mathtt{,} a \mathtt{\}}\ e_1 \ldots e_a} \rho =
  \mathcal{C}\dbra{e_a} \rho^{+ 0} \doubleplus \ldots 
  \doubleplus \mathcal{C}\dbra{e_1} \rho^{+ (a - 1)}
  \doubleplus [ \verb|Pack | t\ a ]
  \]
  
  \vspace{2em}
  
  If the given expression is an if expression \(\verb|if | e_0 \verb| then | e_1 \verb| else | e_2\), we will evaluate \(e_0\) but we won't evaluate one of \(e_1\) or \(e_2\).
  However, this will be accomplished by the \texttt{Cond} instruction, which has two codes and will look at the top of the stack and determine which code to execute in the evaluating process.
  Therefore, we can compile the every three argument in a strict context.
  Notice that we haven't convert the environment as \(\rho^{+ 1}\)
  since we will be discarding the result of the \(e_0\) at the top of the stack and we put the address to the node which consists he expression \(e_1\) or \(e_2\) at the top of the stack, the exact same location where \(e_0\) was before.

  \[
  \mathcal{E} \dbra{
    \mathtt{if\ } e_0 \mathtt{\ then\ } e_1 \mathtt{\ else\ } e_2
  } \rho  =
  \mathcal{E}\dbra{e_0} \rho
  \doubleplus [ \mathtt{Cond\ }
    (\mathcal{E}\dbra{e_1} \rho)
    (\mathcal{E}\dbra{e_2} \rho)
  ]
  \]
  
  \vspace{2em}

  The default cases including an application, a variable and a constructor are compiled in a non-strict context.
  Since the evaluation of an application should be done in a lazy manner so does the application to a constructor too.
  The evaluation of a variable and a unary constructor can be done in a step and the generated code should have no difference regardless of whether we compiled it in a strict context or non-strict context.
  Therefore, we would like to compile them in a non-strict context so we would not have to repeat the same compilation rule both in compilation of strict and non-strict context.
  We also need to \texttt{eval}uate the expression until it reaches to the weak head normal form.  
  \[
  \mathcal{E} \dbra{e} \rho =
  \mathcal{C}\dbra{e} \rho \doubleplus [ \mathtt{Eval} ]
  \]
  
\end{mdframed}

\begin{mdframed}
  \(\mathcal{C}\) compiles expressions in a non-strict context to accomplish the laziness of the core language.
  The compilation scheme is basically the same as the \(\mathcal{E}\) so we won't explain them in detail.

  \[
  \mathcal{C} \dbra{i} \rho  = [\mathtt{PushInt}\ i]
  \]
  
  \vspace{0.5em}

  {\renewcommand{\arraystretch}{1}
    \begin{center}
      \begin{tabular}{rl} 
        \multicolumn{2}{l}{
          \(\mathcal{C} \dbra{
            \mathtt{let\ } x_1 = e_1; \ldots; x_n = e_n \mathtt{\ in\ } e
          } \rho\)
        }\\
        = & \(\mathcal{C}\dbra{e_1} \rho^{+ 0} \doubleplus \ldots \doubleplus\) \\
        & \(\mathcal{C}\dbra{e_n} \rho^{+ (n - 1)} \doubleplus \) \\
        & \(\mathcal{C}\dbra{e} \rho' \doubleplus
        \dbra{\mathtt{Slide\ } n}\) \\
        & \hspace{2em}
        where \(\rho' = \rho^{+ n}[x_1 \mapsto n - 1, \ldots, x_0 \mapsto 0]\)\\
      \end{tabular}
    \end{center}
  }

  \vspace{1em}

  {\renewcommand{\arraystretch}{1}
    \begin{center}
      \begin{tabular}{rl} 
        \multicolumn{2}{l}{
          \(\mathcal{C} \dbra{
            \mathtt{let\ rec\ } x_1 = e_1; \ldots; x_n = e_n \mathtt{\ in\ } e
          } \rho\)
        }\\
        = & \([ \mathtt{Alloc\ } n] \doubleplus\) \\
        & \(\mathcal{C}\dbra{e_1} \rho' \doubleplus [ \mathtt{Update\ } (n - 1)]
        \doubleplus \ldots \doubleplus\) \\
        & \(\mathcal{C}\dbra{e_n} \rho' \doubleplus [ \mathtt{Update\ } 0] \doubleplus
        \) \\
        & \(\mathcal{C}\dbra{e} \rho' \doubleplus
        \dbra{\mathtt{Slide\ } n}\) \\
        & \hspace{2em}\hfill
        where \(\rho' = \rho^{+ n}[x_1 \mapsto n - 1, \ldots, x_0 \mapsto 0]\)\\
      \end{tabular}
    \end{center}
  }
  \vspace{1em}

  \[
  \mathcal{E} \dbra{\mathtt{Pack\{} t \mathtt{,} a \mathtt{\}}\ e_1 \ldots e_a} \rho
  = 
  \mathcal{C}\dbra{e_a} \rho^{+ 0} \doubleplus \ldots 
  \doubleplus \mathcal{C}\dbra{e_1} \rho^{+ (a - 1)}
  \doubleplus [ \verb|Pack | t\ a ]
  \]

  {\renewcommand{\arraystretch}{1}
    \vspace{0.5em}
    \begin{center}
      \begin{tabular}{ rll } 
        \(\mathcal{C} \dbra{e_0\ e_1} \rho\) & =
        & \(\mathcal{C}\dbra{e_1} \rho \doubleplus \mathcal{C}\dbra{e_0} \rho^{+ 1}
        \doubleplus [ \mathtt{Mkap} ] \) \\
        \\
        \(\mathcal{C} \dbra{f} \rho\)
        & = & \([\mathtt{PushGlobal}\ f]\) \\
        && \hfill where \(f\) is a name of a supercombinator \\
        \\
        \(\mathcal{C} \dbra{x} \rho\)
        & = & \([\mathtt{Push}\ (\rho\ x)]\)
        \hfill where \(x\) is a local variable \\
      \end{tabular}
    \end{center}
  }

  Notice that the compiling with the primitive operator \(\odot\) is simply generating the code \([\mathtt{PushGlobal}\ \odot]\).
  For example, \(\verb|if | e_1 \verb| then | e_2 \verb| else | e_3\) is firstly transpiled into \(\verb|if | e_1\ e_2\ e_3\) (notice that the if expression is now an expression just applying all the three arguments \(e_1, e_2, e_3\) to the \texttt{if} primitive) and then compiled as \([\ldots, \verb|PushGlobal if|]\).
  
\end{mdframed}

\begin{mdframed}
  The alternatives in a \texttt{case} expression are compiled as the following.
  \[
  \mathcal{D} \dbra{\mathit{alt}_1, \ldots, \mathit{alt}_n}\ \rho
  = [ \mathcal{A} \dbra{\mathit{alt}_1} \rho,
    \ldots, \mathcal{A} \dbra{\mathit{alt}_n} \rho ]
  \]
\end{mdframed}

\begin{mdframed}
  An alternative in a \texttt{case} expression is compiled as the following.
  The generated code is basically a tuple of tag and a code.
  The tag is used to match the constructor.
  If the tag has matched, then the combined code will be executed.
  We firstly need to \texttt{split} our structured data and push all of them on the stack.
  By doing so, we can access our variable with the indices obtained from our new environment \(\rho'\).
  At the very last phase of the evaluation, we should \texttt{slide} the addresses on the stack  which were assigned for the newly defined variables since we would no longer need them anymore.

  \[
  \mathcal{A}
  \llbracket
  \verb|<| t \verb|>|\ x_1 \ldots x_n \verb| -> | \mathit{body}
  \rrbracket\ \rho =
  t \verb| -> | [ \verb|Split | n ] \doubleplus \mathcal{E} \dbra{\mathit{body}} \rho'
  \doubleplus [ \verb|Slide | n ]
  \]
  \hfill where \(\rho' = \rho^{+ n}[ x_1 \mapsto 0 \ldots x_n \mapsto n - 1 ]\)
\end{mdframed}

\section{Operational Semantics of G-Machine}
\label{sec:gm-op-sem}

As we have seen in \Cref{sec:gm-intro}, G-Machine has code, stack, dump, heap and global environment.
Here in this section, we will look at the state transition rules for G-Machine.
The transition rules are described as the following.

\vspace{1em}
\fbox{
  \begin{tabular}{ rrrrll } 
    & old code & old stack & old dump & old heap & old global environment \\
    $\Longrightarrow$
    & new code & new stack & new dump & new heap & new global environment \\
  \end{tabular}
}\vspace{1em}

The list of the instructions, code, can be represented as
\(\mathit{instr}_1 : \ldots : \mathit{instr}_n \) as well as
\([\mathit{instr}, \ldots, \mathit{instr}]\) as like in Haskell language.

The existence and a creation of a node \(\mathit{node}\) at the address \(a\) in a heap \(h\) is represented as \(h[a : \mathit{node}]\).

Also, existence of the relation of the name of a supercombinator \(\mathit{name}\) and an address of the corresponding global node \(a\) in a global environment \(m\) is represented as \(m[\mathit{name} : a]\).

\vspace{1em}
\fbox{
  \(\begin{array}{ rrrrll } 
  & \mathtt{PushGlobal\ } f : i & s & d & h & m[f : a] \\
  \Longrightarrow & i & a : s & d & h & m
  \end{array}\)
}\vspace{2em}

\fbox{
  \(\begin{array}{ rrrrll } 
  & \mathtt{PushInt\ } n : i & s & d & h & m \\
  \Longrightarrow & i & a : s & d & h[a : \mathtt{NNum\ } n] & m
  \end{array}\)
}\vspace{2em}

\fbox{
  \(\begin{array}{ rrrrll } 
  & \mathtt{Mkap} : i & a_1 : a_2 : s & d & h & m \\
  \Longrightarrow & i & a : s & d & h[a : \mathtt{NApp\ } a_1\ a_2] & m
  \end{array}\)
}\vspace{2em}

\fbox{
  \(\begin{array}{ rrrrll } 
  & \mathtt{Push\ }n : i & a_0 : \ldots : a_{n + 1} : s & d
  & h[a_{n + 1} : \mathtt{NApp}\ a_n\ a_n'] & m \\
  \Longrightarrow & c & a_n' : a_0 : \ldots : a_{n + 1} : s & d & h & m
  \end{array}\)
}\vspace{2em}

\fbox{
  \(\begin{array}{ rrrrll } 
  & \mathtt{Slide\ }n : i & a_0 : \ldots : a_n : s & d & h & m \\
  \Longrightarrow & c & a_0 : s & d & h & m
  \end{array}\)
}\vspace{2em}

\fbox{
  \(\begin{array}{ rrrrll } 
  & \mathtt{Update\ } n : i & a : a_0 : \ldots : a_n : s & d & h & m \\
  \Longrightarrow & i & a_0 : \ldots : a_n : s & d & h[a_n : \mathtt{NInd\ } a] & m \\
  \end{array}\)
}\vspace{2em}

\fbox{
  \(\begin{array}{ rrrrll } 
  & \mathtt{Pop\ }n : i & a_0 : \ldots : a_n : s & d & h & m \\
  \Longrightarrow & i & s & d & h & m
  \end{array}\)
}\vspace{2em}

\fbox{
  \(\begin{array}{ rrrrll } 
  & \mathtt{Alloc\ } n : i & s & d & h & m \\
  \Longrightarrow & i & a_1 : \ldots : a_n : s & d
  & h
  \begin{bmatrix}
    a_1 : \mathtt{NInd\ }\mathtt{hNull} \\
    \ldots \\
    a_n : \mathtt{NInd\ }\mathtt{hNull} \\
  \end{bmatrix}
  & m \\
  \end{array}\)
}\vspace{2em}

\fbox{
  \(\begin{array}{ rrrrrll } 
  & \odot : i & a_0 : a_1 : s & d
  & h[ a_0 : \mathtt{NNum}\ n_0, a_1 : \mathtt{NNum}\ n_1] & m \\
  \Longrightarrow & i & a : s & d & h[ a : \mathtt{NNum}\ (n_0 \odot n_1 )] & m \\
  \end{array}\)
}\vspace{0.5em}

Where \(\odot\) is a binary operator such as \texttt{Add} or \texttt{Neq}.

\vspace{2em}

\fbox{
  \(\begin{array}{ rrrrrll } 
  & \mathtt{Neg} : i & a : s & d & h[ a : \mathtt{NNum}\ n] & m \\
  \Longrightarrow & i & a' : s & d & h[ a' : \mathtt{NNum}\ (- n)] & m \\
  \end{array}\)
}\vspace{2em}

\fbox{
  \(\begin{array}{ rrrrrll } 
  & \mathtt{Cond}\ i_1\ i_2 : i & a : s & d & h[ a : \mathtt{NNum}\ 1] & m \\
  \Longrightarrow & i_1 \doubleplus i & s & d & h & m \\
  \end{array}\)
}\vspace{2em}

\fbox{
  \(\begin{array}{ rrrrrll } 
  & \mathtt{Cond}\ i_1\ i_2 : i & a : s & d & h[ a : \mathtt{NNum}\ 0] & m \\
  \Longrightarrow & i_2 \doubleplus i & s & d & h & m \\
  \end{array}\)
}\vspace{2em}

\fbox{
  \(\begin{array}{ rrrrrll } 
  & \mathtt{Eval} : i & a : s & d & h & m \\
  \Longrightarrow & [\mathtt{Unwind}] & [ a ] & \langle i, s \rangle : d & h & m \\
  \end{array}\)
}\vspace{2em}

\fbox{
  \(\begin{array}{ rrrrll } 
  & \mathtt{Pack}\ t\ n : i & a_1 : \ldots : a_n : s & d & h & m \\
  \Longrightarrow
  & i & s & d & h[a : \mathtt{NConstr}\ t\ [a_1 : \ldots : a_n] ] & m \\
  \end{array}\)
}\vspace{2em}

\fbox{
  \(\begin{array}{ rrrrll } 
  & \mathtt{Casejump}\ [\ldots, t \textnormal{\texttt{ -> }} i', \ldots] : i
  & a : s & d & h[a : \mathtt{NConstr}\ t\ \mathit{ss}] & m \\
  \Longrightarrow
  & i' \doubleplus i & a : s & d & h & m \\
  \end{array}\)
}
\vspace{2em}

\fbox{
  \(\begin{array}{ rrrrll } 
  & [ \mathtt{Split} ] : i & a : s & d
  & h[a : \mathtt{NConstr}\ t\ [a_1, \ldots, a_n]] & m \\
  \Longrightarrow
  & i & a_1 : \ldots : a_n : s & d & h & m
  \end{array}\)
}\vspace{2em}

\fbox{
  \(\begin{array}{ rrrrll } 
  & [ \mathtt{Unwind} ] & a : s
  & \langle i', s' \rangle : d
  & h[a : \mathtt{NNum\ } n] & m \\
  \Longrightarrow & i' & a : s' & d & h & m
  \end{array}\)
}\vspace{2em}

\fbox{
  \(\begin{array}{ rrrrll } 
  & [ \mathtt{Unwind} ] & a : s
  & \langle i', s' \rangle : d
  & h[a : \mathtt{NConstr}\ n\ \mathit{as}] & m \\
  \Longrightarrow
  & i' & a : s & d & h & m
  \end{array}\)
}\vspace{2em}

\fbox{
  \(\begin{array}{ rrrrll } 
  & [ \mathtt{Unwind} ] & a_0 : s & d & h[a_0 : \mathtt{NInd\ } a] & m \\
  \Longrightarrow & [ \mathtt{Unwind} ] & a : s & d & h & m
  \end{array}\)
}\vspace{2em}

\fbox{
  \(\begin{array}{ rrrrll } 
  & [ \mathtt{Unwind} ] & a_0 : \ldots : a_n : s & d & h[a_0 : \mathtt{NGlobal}\ n\ c] & m \\
  \Longrightarrow & c & a_0 : \ldots : a_n : s & d & h & m
  \end{array}\)
}\vspace{2em}

\fbox{
  \(\begin{array}{ rrrrll } 
  & [ \mathtt{Unwind} ] & a : s & d & h[a : \mathtt{NApp\ } a_1\ a_2] & m \\
  \Longrightarrow & [ \mathtt{Unwind} ] & a_1 : a : s & d & h & m
  \end{array}\)
}\vspace{2em}

When we have a code \([\mathtt{Unwind}]\) but have an empty dump, then we are now at a terminal state.
The stack should now contain exactly 1 address pointing to the desired result.

\chapter{The implementation of G-Machine in HyperLMNtal}
\label{cha:gm-on-lmntal}
\par
In this chapter, we will explain briefly about the G-machine and the compiler we implemented.
The source code for each are in \Cref{appendix:lmn-gm}.
Then we will see some examples to show we have implemented G-machine properly.

\section{Project overview}
\label{cha:lmn-gm-overview}

The program is separated into three parts;
the testing module (test15.lmn in \Cref{lst:gm-lmn-testing-mod}),
the compiler (compiler15.lmn in \Cref{sec:gm-compiler-source})
and G-machine (gm15.lmn in \Cref{sec:gm-source})%
\footnote{
The number, 15, just denotes the version and is not important in this paper.
}
The program is executed with the following command:
\begin{lstlisting}[language = bash]
  slim --hl -t test15.lmn gm15.lmn compiler15.lmn
\end{lstlisting}

\begin{lstlisting}[language=lmntal,
  caption = The testing module,
  label = {lst:gm-lmn-testing-mod}
]
% Testing module

{ test.
    
    % Testing code
    Program = [
        ["id", "x"] = "x",
        ["main"   ] = ("id", 3)
    ].
    
    compileSC(Program). % To starts compiling
    
    % The initial code and the stack.
    % (The initial heap (global nodes) and global environment 
    %  will be produced by compiling the program (super combinators)).
    code = [pushGlobal("main"), unwind].
    stack = [].
    dump = [].

    % Predifined primitive operators.
    primitives @@
    :- uniq |
    	m("+", !Add), 
        !Add = 
    	    nGlobal(2, [push(1), eval, push(1), eval, add, update(2), pop(2), unwind]),
    	m("-", !Sub), 
    	!Sub = 
    	    nGlobal(2, [push(1), eval, push(1), eval, sub, update(2), pop(2), unwind]),
    	m("*", !Mul), 
    	!Mul = 
    	    nGlobal(2, [push(1), eval, push(1), eval, mul, update(2), pop(2), unwind]),
    	m("/", !Div), 
    	!Div = 
    	    nGlobal(2, [push(1), eval, push(1), eval, div, update(2), pop(2), unwind]),
    	m("negate", !Neg), 
    	!Neg = nGlobal(1, [push(0), eval, neg, update(1), pop(1), unwind]),
    	m("<", !LT), 
    	!LT = 
    	    nGlobal(2, [push(1), eval, push(1), eval, lt, update(2), pop(2), unwind]),
    	m("\=", !Neq), 
    	!Neq = 
    	    nGlobal(2, [push(1), eval, push(1), eval, neq, update(2), pop(2), unwind]),
    	m("if", !IF), 
    	!IF = 
    	    nGlobal(3, [ push(0), eval, cond([push(1)], [push(2)]), 
    	    		 update(3), pop(3), unwind]
    	    ).
}.


compile @@
{compiler. @rCompiler}, {test. $p. @rTest}
  :- {compiling. @rCompiler. $p. @rTest}.


exec @@
{gm. @rGM}, {compiling. @rCompiler. $p}/
  :- {@rGM. $p}.


% Terminates if the dump is empty.
terminateWithNum @@
{code = [unwind], stack = [!A], !A = nNum(N), dump = [], $p[], @r}
  :- int(N)
     | result = N.

% Terminates if the dump is empty.
terminateWithConstr @@
{code = [unwind], stack = [!A], !A = nConstr(Tag, Args), dump = [], $p[], @r}
  :- ground(Args), int(Tag)
     | result = nConstr(Tag, Args).
\end{lstlisting}

Here, in the testing module \lstinline|test15.lmn|, the testing code is
\begin{lstlisting}[language=lmntal]
  Program = [
    ["id", "x"] = "x",
    ["main"   ] = ("id", 3)
  ].
\end{lstlisting}
.
This testing code should satisfy the LMNtal Shape Type\cite{lmntalshapetype2016} defined as \Cref{shapetype-of-program}.
We will introduce more examples later in \Cref{sec:examples}.

By executing this with a trace option \lstinline|-t|, we firstly see this to be compiled as the following:%
\footnote{
This is a simplified output.
}
\begin{lstlisting}[language=lmntal]
  m("id",!H37).
  nGlobal(1,[push(0),eval,update(1),pop(1),unwind],!H37).
  m("main",!Hf4).
  nGlobal(0,[pushInt(3),pushGlobal("id"),mkap,eval,update(0),pop(0),unwind],!Hf4).
\end{lstlisting}
Where the \lstinline[language=lmntal]|m("id", !H37)| and \lstinline[language=lmntal]|m("main", !Hf4)| are the global environment, that maps the name of the supercombinators, ``id'' and ``main'', to the address of the whose global nodes exist.
Then we are to execute our virtual machine.

\begin{figure}
\begin{lstlisting}[language=lmntal]
  defshape scs(Prog) {
    Prog = scs :- Prog = [sc|scs].
    Prog = scs :- Prog = [sc].
    SC = sc :- string(Name) | SC = ([Name|vars] = exp).
    Vars = vars :- Vars = [].
    Vars = vars :- string(Var) | Vars = [Var|vars]
    E = exp :- int(N) | E = N.
    E = exp :- E = (exp, exp).
    E = exp :- string(X) | E = X.
    E = exp :- E = if(E, E, E).
    E = exp :- E = case(E, alts).
    E = exp :- E = exp + exp.
    E = exp :- E = exp - exp.
    E = exp :- E = exp * exp.
    E = exp :- E = exp / exp.
    E = exp :- E = exp < exp.
    E = exp :- E = exp \= exp.
    E = exp :- E = negate(exp).
    E = exp :- int(Tag), int(Arity) | E = eConstr(Tag, Arity).
    Alts = alts :- [altNT|alts].
    Alts = alts :- [altNT].
    Alt = altNT :- int(Tag) | Alt = alt(Tag, vars, exp).
  } nonterminal {
    exp(E), sc(SC), scs(SCs), vars(Vars), alts(Alts), altNT(Alt)
  }
\end{lstlisting}
\caption{The shape type of the program}
\label{shapetype-of-program}
\end{figure}

\section{Examples}

\label{sec:examples}

This section gives some core language examples run on the G-machine we have implemented.

\subsection{Recursive functions}

Here is a program that calculates the factorial number and a program that calculates the Fibonacci number.

\begin{lstlisting}[language=lmntal, caption=Factorial]
  Program = [
    ["fac", "n"] = if(1 < "n", ("n" * ("fac", "n" - 1)), 1),
    ["main"    ] = ("fac", 5)
  ].
\end{lstlisting}
This example above will result with the \lstinline|result(120)| which is
\(1 \times 2 \times 3 \times 4 \times 5\) in 166 steps.

\begin{lstlisting}[language=lmntal, caption=nFib]
  Program = [
    ["fib", "n"] = if(2 < "n",
                      ("fib", "n" - 1) + ("fib", "n" - 2),
                      1),
    ["main"    ] = ("fib", 7)
  ].
\end{lstlisting}
The example above should result \lstinline|result(13)| in 793 steps.

Here is a \emph{nFib} program.
This function is famous since the returned value denotes the number of the function called.
And therefore it is a popular program in benchmark tests.
\begin{lstlisting}[language=lmntal, caption=nFib]
  Program = [
    ["nfib", "n"] = if(0 < "n",
                       1 + ("nfib", "n" - 1) + ("nfib", "n" - 2),
                       1),
    ["main"     ] = ("nfib", 10)
  ].
\end{lstlisting}

We have observed that the number of the steps for each \lstinline|nfib| function calls are about 34.5.
That is, we need about 34.5 steps to calculate the
\lstinline|if 0 < n then 1 + nfib (n - 1) + nfib (n - 2) else 1|
and the function call for \lstinline|nfib|.


\begin{figure}
  \begin{center}
    \includegraphics[width=\textwidth]{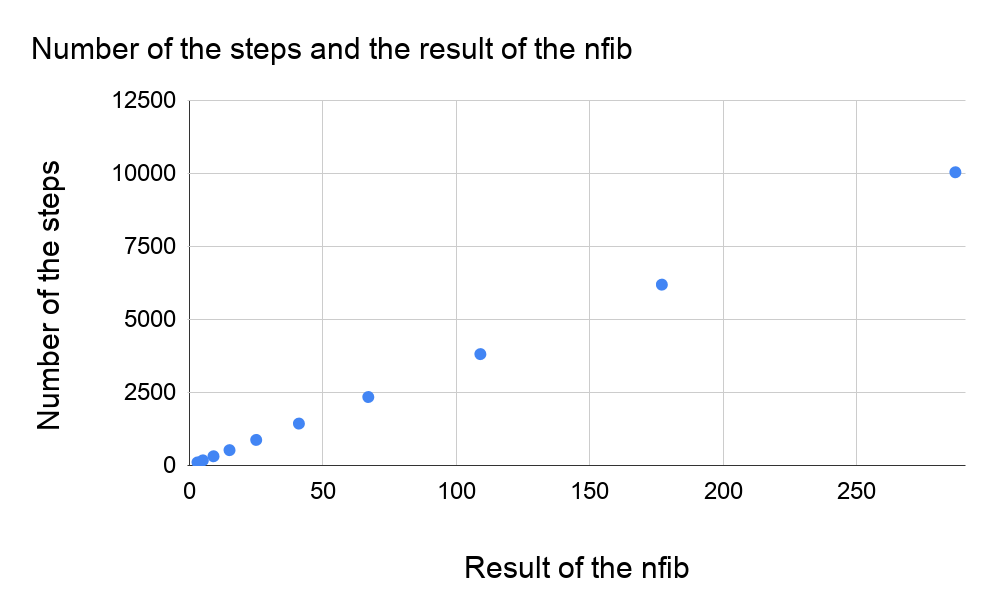}
  \end{center}
  \caption{Number of steps and the result of the \texttt{nfib} in G-machine}
\end{figure}

\subsection{Higher order functions}

Then we are to check that our G-machine have achieved the \emph{higher order function}.
Here is a program that exploit SKI combinators:
\begin{lstlisting}[language=lmntal, caption=SKI combinators]
  Program = [
    ["k", "x", "y"     ] = "x",
    ["s", "x", "y", "z"] = (("x", "z"), ("y", "z")),
    ["id"              ] = (("s", "k"), "k"),
    ["main"            ] = ("id", 3)
  ].
\end{lstlisting}
After the execution of the program, we will get the result \lstinline|result(3)| since the we have applied \lstinline|3| to the ideal function \lstinline|id|.
Notice that we actually did not need an I combinator:
we didn't define the function \lstinline|id| as \lstinline|id x = x| (though it is certainly possible).
Instead, we defined the \lstinline|id| as \lstinline|id = s k k|
since the I combinator is equivalent with the \lstinline|SKK| (or \lstinline|SKS|).
Furthermore we can observe that 
the \lstinline|id| function has no argument and is composed of only \lstinline|s| and \lstinline|k|s.
The \lstinline|id| is certainly a function.
Thus, we achieved to \emph{return a function}.
Also, notice the \lstinline|s| function is taking two \lstinline|k| in the body of the \lstinline|id| function.
Therefore we achieved to \emph{take a function as an argument}.
Finally, at the same place, we can observe that the \lstinline|s| is not taking 3 arguments, there are only two \lstinline|k|s, thus there we performed a \emph{partial application}.
This could be only achieved if we \emph{curried} the function.
Thus, we can see that our G-machine (and a compiler) have successfully implemented the currying.

Here is a more realistic example with higher order functions:
\begin{lstlisting}[language=lmntal, caption=Partial application to the addition function]
  Program = [
    ["incr"] = ("+", "x"),
    ["main"] = ("incr", ("incr", 3))
  ].
\end{lstlisting}
Running this program, we shall obtain \lstinline|result(5)| as the result.
Notice that we defined the function \lstinline|incr| not as \lstinline|incr x = x + 1| but \lstinline|incr = (+) 1|.
This returns the function that takes one argument and increase it by 1.
This shows the power of the functional programming.
We can define a general function, the addition in this case, and apply some specific value, which is 1 in this case, and get a function for a specific purpose.

The example above applies the \lstinline|incr| twice but how about applying it more times ?
We definitely do not want to write \lstinline|incr| more.
Then we may use the \lstinline|twice| function as the following example:
\begin{lstlisting}[language=lmntal, caption=Partial application to the addition]
  Program = [
    ["twice", "f", "x"] = ("f", ("f", "x")),
    ["incr"           ] = ("+", "x"),
    ["main"           ] = ((("twice", "twice"), "incr"), 3)
  ].
\end{lstlisting}
This program yield \lstinline|7| as the result.
Since we have double the effect of the function, that doubles the effect of the function (\lstinline|incr|).
That is, the function \lstinline|twice twice incr| is equivalent with \lstinline|(+) 4|.

\subsection{Non-strict evaluation}

The examples we have seen so far can be all achieved with a strict languages.
However, we have implemented a lazy evaluator.
To prove this, we can do the following experiment.

Firstly, we shall consider the example, results with a runtime failure.
\begin{lstlisting}[language=lmntal, caption=Division by zero]
  Program = [
    ["main"] = 1/0
  ].
\end{lstlisting}

This program can be compiled 
(Actually, whole this program is compiled in a \emph{strict} context as
\lstinline|[pushInt(0),pushInt(1),div,update(0),pop(0),unwind]| (this is a code in a main global node).
Notice that the division is in-lined as \lstinline|div|
instead of calling the function as \lstinline|pushGlobal("/")|.
).
and executed but since we will eventually perform the division, we will never succeed to get the final result.
The program terminates in the following state%
\footnote{
This is based on the output of the implemented G-machine but is simplified
}:
\begin{lstlisting}[language=lmntal, caption=Runtime Error]
  code([div,update(0),pop(0),unwind]).
  stack([!A1,!A0,!Main]).
  dump([]).
  nNum(0,!A0).
  nNum(1,!A1).
  nGlobal(0,[pushInt(0),pushInt(1),div,update(0),pop(0),unwind],!Main).
\end{lstlisting}
We now have an empty dump, therefore we have finished to evaluate the arguments for a primitive operators, they have already reached to the Weak Head Normal Form, this time 1 and 0.
Thus, we are to perform the division, only to find that it is impossible
since we cannot divide the \lstinline|1|, pointed by the top of the stack \lstinline|!A1|, by \lstinline|0|, pointed by the second element of the stack.

Then, how about the following program?
Does this program reaches the final state and returns some value or stops in the execution phase ?
\begin{lstlisting}[language=lmntal, caption=Applying the division-by-zero expression]
  Program = [
    ["k", "x", "y"] = "x",
    ["main"       ] = (("k", 2), 1/0)
  ].
\end{lstlisting}
If we evaluate this in a strict strategy, call by value, then we should perform the devision by 0 since the expression \lstinline|1/0| is applied.
However, if we evaluate this in a non-strict strategy, call by name or call by need, then we do not have to execute the division.
Since the \lstinline|K| combinator discards the second argument.
Thus, we can obtain the final result \lstinline|2| executing this program on the G-machine we have implemented.
We can check this running the machine.

Here is a list of rules used in the transition of G-machine in the former program.
\begin{lstlisting}
---->pushGlobal
---->unwindGlobalWithEmptyDump
---->pushInt
---->pushInt
---->pushGlobal
---->mkap
---->mkap
---->pushInt
---->pushGlobal
---->mkap
---->mkap
---->eval
---->unwindApp
---->unwindApp
---->unwindGlobalWithNonEmptyDump
---->push
---->eval
---->unwindInt
---->update
---->pop0
---->unwindInt
---->update
---->pop0
---->terminateWithNum
\end{lstlisting}

Notice that there is no \lstinline|div| seen in the transition step.
We certainly did not perform the division.
This shows an advantage of the non-strict evaluation.

\subsection{Call By Name evaluation strategy}
In the previous example, we have seen that we have achieved the \emph{non-strict} evaluation
but we have not yet checked whether we did accomplished the \emph{call by need} evaluation other than the \emph{call by name}.
To show this, we can experiment with the following program:
\begin{lstlisting}[language=lmntal, caption=Double]
  Program = [
    ["double", "x"] = "x" + "x",
    ["main"       ] = ("double", 1 + 2)
  ].
\end{lstlisting}
which will result \lstinline|result(6)|.

Actually, this program is the same program we have seem in \Cref{sec:graph-reduction-intro}.
In short, we are to share the sub-expression \lstinline|1 + 2| and should not calculate this twice.
Thus, we will be performing the addition twice; in the \lstinline|double| function, \lstinline|3 + 3|, and for the applied expression, \lstinline|1 + 2|.
Here is a list of the transition rules of G-machine.
\begin{lstlisting}
---->pushGlobal
---->unwindGlobalWithEmptyDump
---->pushInt
---->pushInt
---->pushGlobal
---->mkap
---->mkap
---->pushGlobal
---->mkap
---->eval
---->unwindApp
---->unwindGlobalWithNonEmptyDump
---->push
---->eval
---->unwindApp
---->unwindApp
---->unwindGlobalWithNonEmptyDump
---->push
---->eval
---->unwindInt
---->push
---->eval
---->unwindInt
---->add          % performing the "1 + 2"
---->update
---->pop0
---->unwindInt
---->push
---->eval
---->unwindInt
---->addSameNode  % performing the "3 + 3"
---->update
---->pop0
---->unwindInt
---->update
---->pop0
---->terminateWithNum
\end{lstlisting}
Notice that we have performed the addition only twice, at the \lstinline|add| and the \lstinline|addSameNode|.
We have defined the rules differently since the former will match the different nodes and the latter will match the same nodes.
The rules are the following:
\begin{lstlisting}[language=lmntal, caption=The transition rule for the addition in G-machine]
  % Perform the addition with the different nodes as the arguments.
  add @@
  code = [add|Is], stack = [!A0, !A1|S], !A0= nNum(N0), !A1 = nNum(N1)
    :- N = N0 + N1
      | code = Is, stack = [!A|S], !A0= nNum(N0), !A1 = nNum(N1), !A = nNum(N).
  
  % Perform the addition with the same node as the arguments.
  addSameNode @@
  code = [add|Is], stack = [!A0, !A0|S], !A0= nNum(N0)
    :- N = N0 + N0
      | code = Is, stack = [!A|S], !A0 = nNum(N0), !A = nNum(N).
\end{lstlisting}
We can observe that the G-machine was not sharing an expression in the addition of the first time, when it is performing \lstinline|1 + 2|, but it was sharing the expression when we perform the second addition \lstinline|3 + 3|.
This can be also checked with the visualizer, \emph{Graphene}\cite{graphene}.
Here is a shot of the visualizer.

\begin{center}
  \begin{tabular}{ccc}
    \includegraphics[clip, width=6cm]{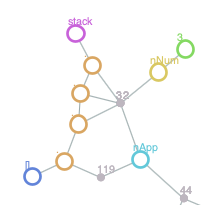}
    \begin{tikzpicture}[overlay, baseline = 0pt ]
      \useasboundingbox (0, 0) rectangle(1, 1);
      \draw [ color = black!25, line width = 7pt, opacity = 0.5]
      (- 1.5, 4.5) circle (40pt);
    \end{tikzpicture}
    &
    \begin{tikzpicture}[baseline = 0pt, scale = 0.3 ]
      \draw [line width = 3pt, color = blue!30!black!20]
      (0, 8) --++ (3, 0);
      \draw [line width = 3pt, color = blue!30!black!20]
      (2, 7) --++ (1, 1) --++ (-1, 1);
    \end{tikzpicture}
    &
    \includegraphics[clip, width=6cm]{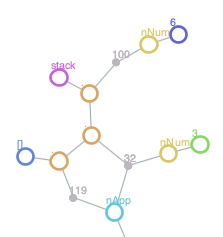}
    \begin{tikzpicture}[overlay, baseline = 0pt ]
      \useasboundingbox (0, 0) rectangle(1, 1);
      \draw [ color = black!25, line width = 7pt, opacity = 0.5]
      (- 1.5, 5) circle (40pt);
    \end{tikzpicture}
    \\
    \begin{minipage}{0.44\textwidth}
      \begin{center}
        The result of the \(1 + 2\), \(3\) is shared\\
      \end{center}
    \end{minipage}
    &&
    \begin{minipage}{0.44\textwidth}
      \begin{center}
        The next instruction (\lstinline|addSameNode|) generates the 
        \(6\ (= 3 + 3)\) and push its address on the top of the stack.
      \end{center}
    \end{minipage}
  \end{tabular}
\end{center}

\subsection{Structured Data: Infinite List}
So far, we have observed that we have accomplished the non-strict evaluation.
Here comes a more practical examples: \emph{an infinite list}.

The following example calculates the factorial number using an infinite list
\lstinline|[1, 2, 3, 4, 5, ...]|.
The function \lstinline|from| generates the infinite list.
And the \lstinline|timesN list n| multiples the first \lstinline|n| of the list \lstinline|list|.
Here, \lstinline|eConstr(2, 2)| is a data \emph{constructor} and can be read as \lstinline|Cons| for this time as we have already described in \Cref{sec:gm-corelang-syntax}.
The
\lstinline|alt(2, ["h", "t"], "h" * (("timesN", "t"), "n" - 1))|
performs the \emph{pattern matching}.
This checks the \emph{tag} of the structured data and binds the components of the data to the variables, for this time, \lstinline|"h"| and \lstinline|"t"|.

\begin{lstlisting}[language=lmntal, caption=Calculating the factorial using an infinite list]
  Program = [
    ["from", "n"] = ((eConstr(2, 2), "n"), ("from", "n" + 1)), 
    ["timesN", "list", "n"] = if( 0 < "n",
    case( "list", 
    [ alt(1, [], 1),     % Nil
      alt(2, ["h", "t"], % Cons
      "h" * (("timesN", "t"), "n" - 1)
      ) ]
    ),
    1
    ),
    ["main"] = (("timesN", ("from", 1)), 5)
  ].
\end{lstlisting}
Again, we can obtain \lstinline|result(120)|.

A more interesting example exploiting an infinite list is the \emph{Sieve of Eratosthenes}.
This is a well-known algorithm to obtain prime numbers.
The program is the following:
\begin{lstlisting}[language=lmntal, caption=The Sieve of Eratosthenes]
  Program = [
    ["main"] = (("getNth", 5), ("sieve", ("from", 2))),
    ["from", "n"] = ((eConstr(2, 2), "n"), ("from", "n" + 1)),
    ["sieve", "xs"] = case("xs", [
      alt(1, [], eConstr(1, 0)),
      alt(2, ["p", "ps"], 
        ((eConstr(2, 2), 
          "p"), 
         ("sieve", (("filter", ("nonMultiple", "p")), "ps")))
      )]),
    ["filter", "predicate", "xs"] = case("xs", [
      alt(1, [], eConstr(1, 0)),
      alt(2, ["p", "ps"], 
        let(["rest" = (("filter", "predicate"), "ps")],
          if(("predicate", "p"), 
             ((eConstr(2, 2), "p"), "rest"),
             "rest"
          ))
      )]),
    ["nonMultiple", "p", "n"] = (("n" / "p") * "p" \= "n"),
    ["getNth", "n", "xs"] = case("xs", [ 
      alt(1, [], -1), % Error
      alt(2, ["p", "ps"], 
        if(0 < "n",
           (("getNth", "n" - 1), "ps"),
           "p"
        )
      )])
  ].
\end{lstlisting}
Where \lstinline|getN| returns the \lstinline|5|th element of the sieve, \lstinline|from| generates an infinite list \lstinline|[2, 3, 4, 5, 6, ...]|, \lstinline|sieve| excludes the multiple of the head from the given list, \lstinline|filter| filters an list with a function \lstinline|predicate| and \lstinline|nonMultiple| determines whether the second argument is a multiple of the first.

We have also counted the number of the steps of our G-machine took in the former program.
Notice that this is not the number of the steps of HyperLMNtal.
It took 36504 steps to obtain the 33th prime number, 131.
Here we can observe that this algorithm takes about \(\mathcal{O}(N \log(\log N))\) steps in \Cref{eratos-n-p}.

\begin{figure}
  \begin{center}
    \includegraphics[width=\textwidth]{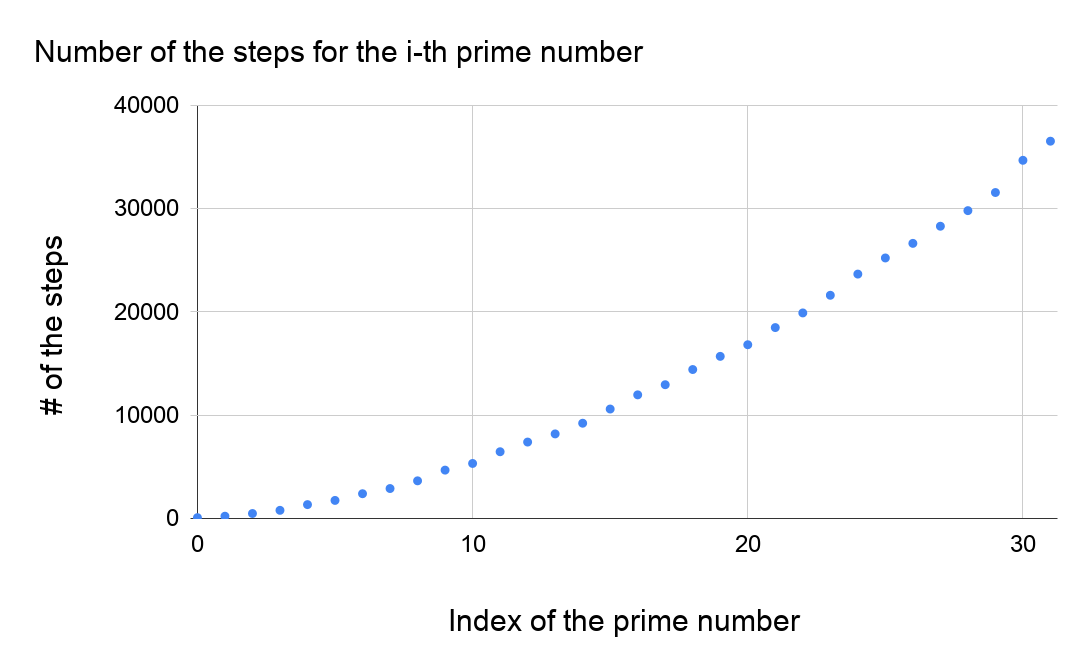}
  \end{center}
  \caption{Number of the steps of the implemented G-machine and the index of the prime number}
  \label{eratos-n-i}
\end{figure}

\begin{figure}
  \begin{center}
    \includegraphics[width=0.95\textwidth]{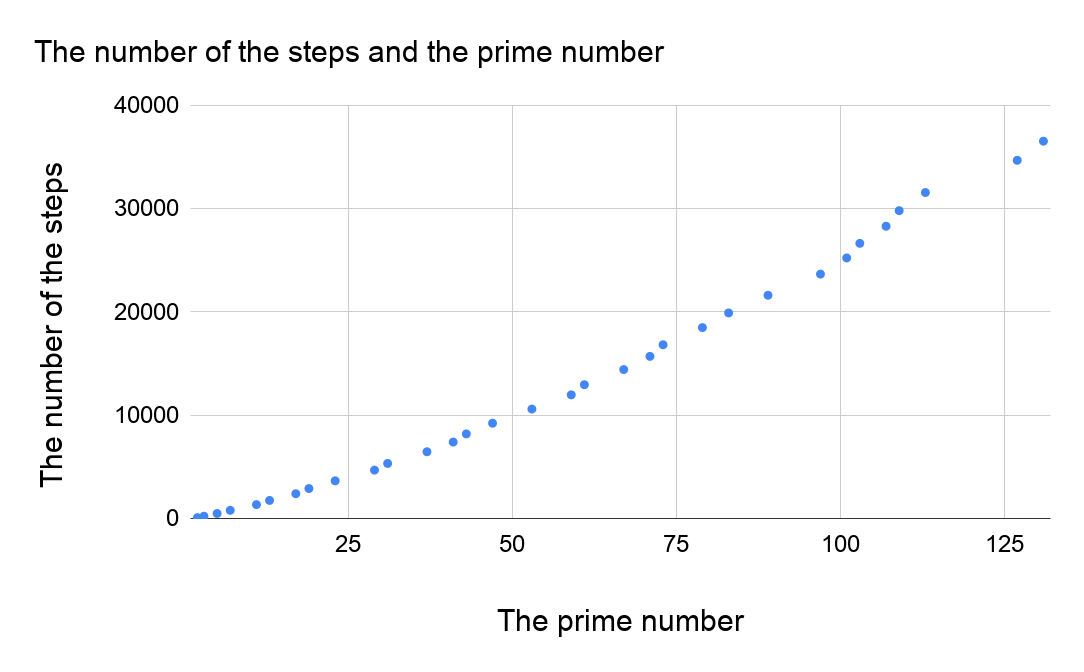}
  \end{center}
  \caption{Number of the steps of the implemented G-machine and the prime number obtained}
  \label{eratos-n-p}
\end{figure}





\chapter{Further Discussion}
\par

In this chapter, we will give a further discussion based on the experience
of the implementation in the former chapter.

\section{Proposal for abbreviation rules for a hyperlink}

We observed that the hyperlinks often connect 2 atoms in the practice.
In LMNtal, a local link, a link with just 2 end points, can be abbreviated.
For example, \lstinline|a(X, Y, Z), b(W, Y)| can be abbreviated as
\lstinline|a(X, b(W), Z)|.
Thus, we may introduce the same strategy into hyperlinks.
For example, as follows:
\begin{mdframed}
  In the proposed HyperLMNtal syntax,
  \[
  \nu !H.(p(X_1, \ldots, !H, \ldots, X_n),
  q(Y_1, \ldots, Y_{n - 1}, !H))
  \]
  can be abbreviated as
  \[
   p(X_1, \ldots, X_{i-1}, !q(Y_1, \ldots, Y_{n - 1}), X_{i+1}, \ldots, X_n)
  \]
\end{mdframed}

Using this abbreviation, for example,
a program \lstinline|a(!X), b(!X)| can also be written as \lstinline|a(!b)|.
Furthermore, a hyperlink connected list
\begin{lstlisting}[language=lmntal]
 !X1 = [A1|!X2],
 !X2 = [A2|!X3], 
 !X3 = [A3|!X4],
 !X4 = [].
\end{lstlisting}
can be abbreviated as 
\begin{lstlisting}[language=lmntal]
 !X1 = [A1|![A2|![A3|![]]]].
\end{lstlisting}

We may introduce a more ambitious abbreviation scheme for a list.
For example, we may like to write the former list as 
\begin{lstlisting}[language=lmntal]
 !X = [A1!A2!A3![]]
\end{lstlisting}

\section{Abstraction in the visualizer and the output of the HyperLMNtal runtime environment}

In the previous chapter, we have seen the screenshot of the visualizer\cite{graphene}
and the output of the HyperLMNtal runtime environment\cite{slim}.
However, they were actually simplified for this purpose.
It was very hard to observe the desired information when using these tools.
To tell the truth, we did not use the visualizer when implementing this.
Here is the screenshot of the Graphene in \Cref{graphene-fac4}.
Its is a state in the calculation process of factorial of 4.

\begin{figure}
  \begin{center}
    \includegraphics[width=\textwidth]{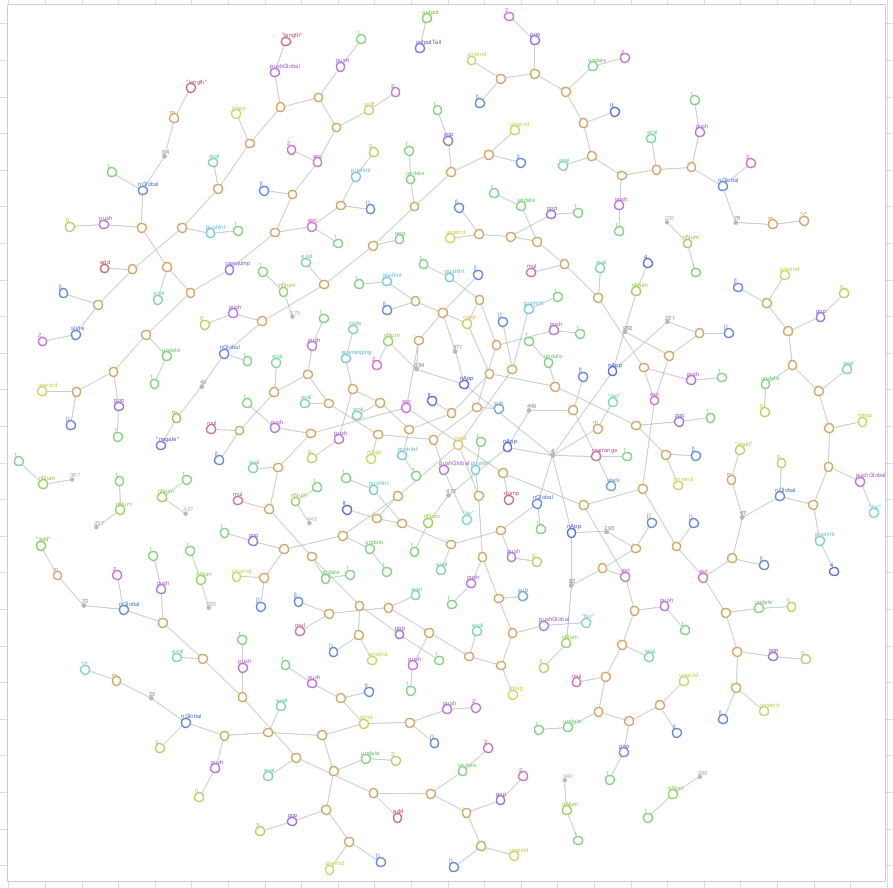}
  \end{center}
  \caption{A state in the calculation process of factorial of 4}
  \label{graphene-fac4}
\end{figure}
As you can see, this is that easy-to-understand output.
This is simply because that it output too much information.
Thus, we can improve this by abstracting the state and reducing the information.

One possible way to achieve this is to write some rules for the abstraction and apply them to each of the states in the execution steps.
For example, in G-machine, we have a global node,
which has a pretty long code, a list of instructions, which does not changes through the execution process.
Therefore, we may like to discard all the information about the global nodes.
Then for example, we can have a rule as the following:
\begin{lstlisting}[language=lmntal]
  clearGlobalNode @@
  !A = nGlobal(Arity, Code)
    :- int(Arity), ground(Code)|.
\end{lstlisting}
Applying the rule will let the output much clearer.
However, we may want to restore the previous state.

For example, we may do not want to delete the global node entirely but want to simplify to one atom with the following rule and want to restore it if the user clicked the simplified atom in the visualizer.
\begin{lstlisting}[language=lmntal]
  abstractGlobalNode @@
  !A = nGlobal(Arity, Code)
    :- int(Arity), ground(Code)|
      | !A = nGlobal.
\end{lstlisting}
Of course we cannot restore the previous graph if we simply discarded it.
We should probably store the previous graph.
However, even we did so, it is not obvious to restore the graph that has degenerated to the atom.
Well there is no certain way in my thought but maybe the technique in the \emph{Reversible computing} help.
We may like to restrict the degenerating rule syntactically, for example, the grammar used in the LMNtal Shape Type\cite{lmntalshapetype2016} or CSLMNtal\cite{cslmntal}, in order to achieve this.






\section{Else If statement in a rule}

In LMNtal, we cannot use the \emph{otherwise} pattern.
That is, there does not exists the \emph{else if statements}.
The runtime environment SLIM\cite{slim} does attempt the matching and the rewriting from the top of the rules to the bottoms.
And this, else if statement, can be actually achieved by just denoting rules downwards.
However, this is not guaranteed as a language feature:
this is just a behavior of the processing system.


Therefore, in the implementation of the compiler for G-machine, we have written all the \emph{otherwise} patterns.
However, this is very inconvenient.

Also, we have observed that the existence of the rule that matches partially affects the execution time heavily.
The program used in the experiment if the following:

\begin{lstlisting}[language=lmntal, caption=The benchmark program]
incr(n(1)).

% Partially matches
irrelevantRule1 @@
incr(a) :- . 

% Does not matches at all
irrelevantRule2 @@
decr(a) :- . 

increment @@
incr(n(N)) :- N < 1000000, N1 = N + 1 | incr(n(N1)).
\end{lstlisting}

In our experiment, the program is compiled with the options
\verb|--slimcode -O3 --use-swaplink|.
And experimented in the environment shown in the \Cref{tab:environment}.
\Cref{partial-match} shows the execution time of the program.

\begin{table}[htb]
  \begin{center}
    \caption{Experiment Conditions}
    \begin{tabular}{l|l} \hline
      Processor & 2.2 GHz Dual-Core Intel Core i7 \\ 
      Memory & 8 GB \ 1600 MHz \ DDR3 \\
      LMNtal Version & LMNtal Compiler 1.50 \\ 
      SLIM Version & 2.4.0 \\ \hline
    \end{tabular}
    \label{tab:environment}
  \end{center}
\end{table}

\begin{figure}
  \begin{center}
    \includegraphics[width=\textwidth]{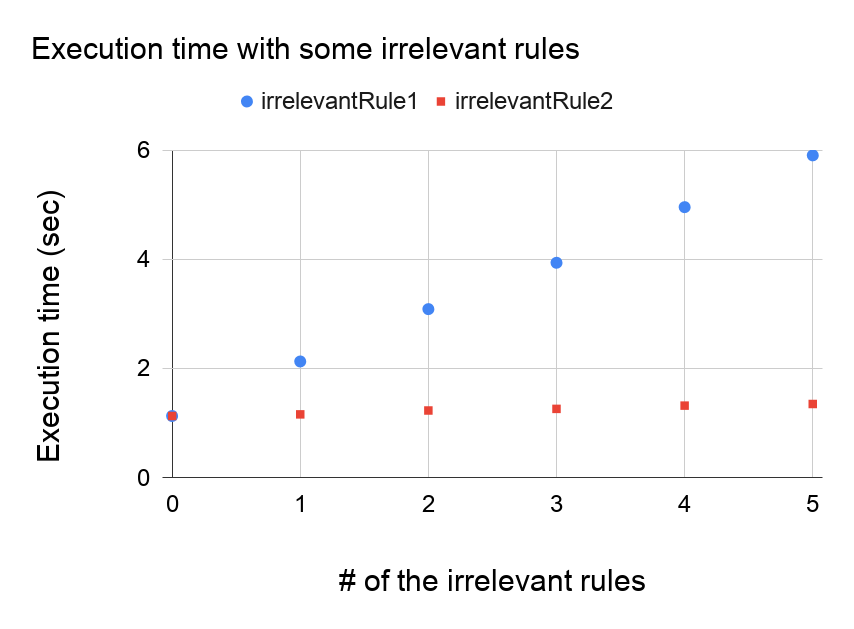}
  \end{center}
  \caption{Execution time with the arbitrary number of the irrelevant rules}
  \label{partial-match}
\end{figure}

We can clearly see that the number of the partially matching rule make the performance worse.
This is a very big issue since we often write the partially matching rules.
For example, in the G-machine we have implemented, we firstly search for the atom \verb|code| and this eventually succeeds, since we always have the atom.
However, we might backtrack depending on the head of the list of the code, for example, \verb|unwind|, \verb|push|, \verb|add|, etc.
If we failed to apply the rule, then it was a rule that partially matched.
And this happens so many times.
Thus, this problem affects the overall performance.

Incorporating the \emph{else if statement} gives one solution to this problem.
For the example above, we do not want to search for the \verb|code| again and again.
CSLMNtal\cite{cslmntal} has a rule with multiple guards but we want to exploit the else if statement in more general matching.
Thus, we may introduce a new syntax as the \Cref{table:syntax}:

\begin{figure}
  \hrulefill
  \begin{center}
    \begin{tabular}{ lrclll } 
      (Rule) & \(\mathit{rule}\) &$::=$&
      \(P \means [\mathit{rule}_1 \verb' || ' \ldots \verb' || ' \mathit{rule}_n]\)
      & \(n > 1\) & Else If statement \\
      &&$|$& \(P \means P \verb| @@ | \mathit{RuleName}\) && Rule \\
    \end{tabular}
  \end{center}
  \hrulefill
  \caption{Syntax of Flat HyperLMNtal}
  \label{table:syntax}
\end{figure}


For example, the following program 
\begin{lstlisting}[
    language=lmntal,
    caption=makeSpine,
    basicstyle={\ttfamily \footnotesize}
  ]
  % Make application spines
  makeSpineApp @@
  H = [ makeSpine((E1, E2)) | T] 
  :- H = [ makeSpine(E1), E2 | T].
  
  % The default cases (other than application) for the "makeSpine"
  makeSpineDefaultInt @@
  R = makeSpine(Int) 
  :- int(Int) | R = Int.    
  
  makeSpineDefaultVar @@
  R = makeSpine(Var) 
  :- string(Var) | R = Var.    
  
  makeSpineDefaultCase @@
  R = makeSpine(case(E, Alts)) 
  :- R = case(E, Alts).

  makeSpineDefaultIf @@
  R = makeSpine(if(E1, E2, E3)) 
  :- R = if(E1, E2, E3).
  
  makeSpineDefaultPlus @@
  R = makeSpine(E1 + E2) 
  :- R = E1 + E2.
  
  makeSpineDefaultTimes @@
  R = makeSpine(E1 * E2) 
  :- R = E1 * E2.
  
  makeSpineDefaultMinus @@
  R = makeSpine(E1 - E2) 
  :- R = E1 - E2.
  
  makeSpineDefaultDiv @@
  R = makeSpine(E1 / E2) 
  :- R = E1 / E2.
  
  makeSpineDefaultNeg @@
  R = makeSpine(negate(E)) 
  :- R = negate(E).
  
  makeSpineDefaultLessThan @@
  R = makeSpine(E1 < E2) 
  :- R = (E1 < E2).
  
  makeSpineDefaultNeq @@
  R = makeSpine(E1 \= E2) 
  :- R = (E1 \= E2).
  
  makeSpineDefaultLet @@
  R = makeSpine(let(Defns, E)) 
  :- R = let(Defns, E).

  makeSpineDefaultLetRec @@
  R = makeSpine(letrec(Defns, E)) 
  :- R = letrec(Defns, E).
  
  makeSpineDefaultConstr @@
  R = makeSpine(eConstr(Tag, Arity)) 
  :- R = eConstr(Tag, Arity).
\end{lstlisting}

can be rewritten as the following
\begin{lstlisting}[language=lmntal, caption=makeSpine with an else if syntax]
  % Make spines
  H = [ makeSpine(Exp) | T] :- [
    % Make application spines
    Exp = (E1, E2) :- H = [ makeSpine(E1), E2 | T] @@ makeSpineApp
    % The default cases (other than application) for the "makeSpine"
    || :- H = [Exp | T] @@ makeSpineDefault
  ]
\end{lstlisting}
which is certainly easy to read and write.

We can implement this easily; just adding a \verb|branch| instruction to the HyperLMNtal runtime environment SLIM.

%
%
%
%
%
%
%

\section{Directed HyperLMNtal}
The hyperlinks in HyperLMNtal are undirected.
However, considering to deal with the directed hyperlinks is important from the view of
(\romannumeral 1) the efficiency (the one-way pointer is more efficient that the mutual links) and
(\romannumeral 2) the programmer to program with directed edges.
Thus, we introduce a language extended HyperLMNtal, Directed HyperLMNtal.

\Cref{hl-deref} is a instructions obtained by compiling the HyperLMNtal program
\begin{lstlisting}[language=lmntal]
a(!H), b(!H) :- .
\end{lstlisting}.
Notice that we should perform \verb|findatom| twice to obtain the atoms ``a'' and ``b''  which are connected by a hyperlink.
That is, we search two atoms randomly and checks whether it is connected to the same hyperlink (by the instruction \verb|samefunc|).
This is obviously inefficient.
We should let this to execute \verb|findatom| once and then \verb|deref|erence through the connected hyperlink.
To implement this, it will be much easier if we let hyperlink as an one-way pointer, which is referred as a backward hyperarc in \cite{towards}, rather than an undirected hyper edge.

Capability typing\cite{towards} gives one solution to this problem but it is too much strict.
Since the capability typing does not allow a port to be connected to a hyperlink that the number of the end-points changes through the calculation process.

\begin{example}{List traversal}
The following program is il-typed in the capability typing system but it does not yields any null/dangling pointers and should be allowed.
\begin{align*}
    & \overline{R} \mapsto walk(Prev), \overline{Prev} \mapsto cons(Elem, Next)\\
    & \vdash \overline{R} \mapsto walk(Next), \overline{Prev} \mapsto cons(Elem, Next)
\end{align*}
\end{example}

Thus, we shall firstly introduce the Directed HyperLMNtal, which simply let hyperlinks in our HyperLMNtal to be directed.
Then discuss the condition that the Directed HyperLMNtal should satisfy.
Again, we shall call a hyperlink as a link for the simplicity from now on in this section.

\begin{figure}
  \begin{center}
    \begin{lstlisting}
      spec         [1, 5]
      findatom     [1, 0, 'a'_1]
      derefatom    [3, 1, 0]
      ishlink      [3]
      isunary      [3]
      findatom     [2, 0, 'b'_1]
      derefatom    [4, 2, 0]
      isunary      [4]
      samefunc     [3, 4]
      jump         [L103, [0], [1, 2, 3, 4], []]
    \end{lstlisting}
  \end{center}
  \caption{Example of an intermediate code of HyperLMNtal}
  \label{hl-deref}
\end{figure}

\subsection{Syntax and operational semantics of Directed HyperLMNtal}
\label{subsec:dhl-syntax-opsem}

$X$ denotes a link name and $p$ denotes an atom name. 
The only reserved name is $\mapsto$.
The syntax is given in \Cref{table:dhl-syntax}. 
An atom \(\overline{X} \mapsto Y\) is called a \emph{indirection}.
Given a rule $(P \vdash Q)$, rules must not appear in $P$ and it should satisfy \fn{P} $\supseteq$ \fn{Q}.
Notice that the name of the head and the tail of the link are in the same name space.
For example, the free link names of \(p(X, Y, \overline{Z})\) is \(\{X, Y, Z\}\).

\begin{figure}
  \hrulefill
  \begin{center}
    \begin{tabular}{ rlclll } 
      (Process) &$P$&$::=$& \(\zero\) && Null \\
      &&$|$& \(p (\mathit{L}_1, \ldots ,\mathit{L}_m)\) & $m \geq 0$ & Atom \\
      &&$|$& \((P, P)\) && Molecule \\
      &&$|$& \(\nu X.P\) && Link Creation \\
      &&$|$& \((P \vdash P)\) && Rule \\
      \\
      (Link) & $L$&$::=$& \(X\) && Tail \\
      &&$|$& \(\overline{X}\) && Head \\
    \end{tabular}
  \end{center}
  \hrulefill
  \caption{Syntax of Directed HyperLMNtal}
  \label{table:dhl-syntax}
\end{figure}

We define the relation $\equiv$ on processes as the minimal equivalence relation satisfying the rules shown in \Cref{table:dhl-equiv}.
Where $P[Y / X]$ is a link name substitution that replaces all free occurrences of $X$ with $Y$ (and $\overline{X}$ with $\overline{Y}$).
If a free occurrence of $X$ occurs in a location where $Y$ would not be free, $\alpha$-conversion may be required.

\begin{figure}
  \hrulefill
  \begin{center}
    \begin{tabular}{ ll } 
      (E1) & \((\mathbf{0}, P) \equiv P\) \\
      (E2) & \((P, Q) \equiv (Q, P)\) \\
      (E3) & \((P, (Q, R)) \equiv ((P, Q), R)\) \\
      (E4) & \(P \equiv P' \Rightarrow (P, Q) \equiv (P', Q)\) \\
      (E5) & \(P \equiv Q \Rightarrow \nu X.P \equiv \nu X.Q\) \\
      (E6) & \(\nu X.(\overline{X} \mapsto Y, P) \equiv \nu X.P[Y / X]\)\\
      & where \(X \in \mathit{fn}(P) \lor Y \in \mathit{fn}(P)\) \\
      (E7) & \(\nu X.\nu Y.\overline{X} \mapsto Y \equiv \zero\) \\
      (E8) & \(\nu X.\zero \equiv \zero\)\\
      (E9) & \(\nu X.\nu Y.P \equiv \nu Y.\nu X.P\)\\
      (E10)& \(\nu X.(P,Q) \equiv (\nu X.P,Q)\)\\
      & where \(X \notin \mathit{fn}(Q)\) \\
    \end{tabular}
  \end{center}
  \hrulefill
  \caption{Structural congruence on Directed HyperLMNtal processes}
  \label{table:dhl-equiv}
\end{figure}

We define the reduction relation $\longrightarrow$ on processes as the minimal relation satisfying the rules in \Cref{table:dhl-trans}.

\begin{figure}
  \hrulefill
  \begin{center}
    \begin{tabular}{ ll } 
      (R1) & \(\dfrac{P \longrightarrow P'}{(P, Q) \longrightarrow  (P', Q)} \)
      \vspace{1em}\\
      (R2) & \(\dfrac{P \longrightarrow P'}{\nu X.P \longrightarrow  \nu X.P'} \)
      \vspace{1em} \\
      (R3) & \(\dfrac{Q \equiv P \hspace{16pt} P \longrightarrow P' \hspace{16pt} P' \equiv Q'}{Q \longrightarrow Q'} \)
      \vspace{0.5em} \\
      (R4) & \( (P, (P \vdash Q)) \longrightarrow (Q, (P \vdash Q)) \) \\
    \end{tabular}
  \end{center}
  \hrulefill
  \caption{Reduction relation on Directed HyperLMNtal processes}
  \label{table:dhl-trans}
\end{figure}

\subsection{Conditions for Directed HyperLMNtal}
\label{subsec:dhl-condition}

\subsubsection*{Functional (right-unique) condition}
To let a hyperlink to be a backward hyperarc, we should let the head of the link to occur just once.
Thus, we can add a following condition.

\begin{condition}[Functional condition]
  For all link named \(X \in \mathit{fn}(P)\) in a process $P$, 
  the \textit{head} of the link ($\overline{X}$) must not occur free more than once in $P$. 
\end{condition}

Notice the congruence rules should be applied preserving this condition.
For example, we cannot transform \(\nu X.p(\overline{X})\) to
\(\nu X.(\overline{X} \mapsto Y, p(\overline{X}))\) with the reversed rule of (E6).

By introducing this condition, the indirection (link fusion) is no longer symmetry.
Since we have proved the symmetry of the fusion via the process
\(\nu Z. (Z \bowtie X, Z \bowtie Y)\),
which is now disallowed in Directed HyperLMNtal
because \(\nu Z. (\overline{Z} \mapsto X, \overline{Z} \mapsto Y)\), in which the head of a link $Z$, $\overline{Z}$, occurs twice, is prohibited\footnote{
However, we cannot say it is asymmetry just because we cannot prove the symmetricity.
}.

\subsubsection*{Serial (left-total) condition}
We definitely do not want to allow null/dangling pointers.
To achieve this, we should add a condition as the following.

\begin{condition}[Serial condition 1 (pre-revised)]
  \label{insufficient-serial-condition}
  For all the process $\nu X.P$ where \(X \in \mathit{fn}(P)\) in a given process, the \textit{head} of the $X$ must exist in $P$ .
\end{condition}

However, this is insufficient considering the congruence rules.
For example,
\(\nu X.(\overline{X} \mapsto X, p(X))\)
satisfies the both functional and serial conditions,
but the congruent process \(\nu X. p(X)\) obtained by applying the rule (E6) does not satisfy this.
To resolve this, we should first consider a process without indirection.
The serial condition is now revised as the following.

\begin{condition}[Serial condition 2 (revised)]
  \label{revised-serial-condition}
  
  This condition is satisfied if
  the \textit{head} of the link $X$ exists in $P$
  and the indirection \(\overline{X} \mapsto Y\) (whose $X$ occurs free in $P$) does not occur in $P$
  for all the process $\nu X.P$ where \(X \in \mathit{fn}(P)\) in the given process.

  Also, $Q$ satisfies this condition if $P$ satisfies this condition where $P \equiv Q$.
\end{condition}


Intuitively, this means that if you follow the link whose scope is closed from a port of an atom except an indirection possibly via indirections, you will always reach a port (head of the link) of an atom except an indirection.
We shall call it ``reachable''.

\subsubsection*{An additional condition for a rule}

We should add an additional condition to store the serial condition after reduction.

\begin{condition}[An additional condition for a rule]
 If the \textit{head} of a link $X$ occurs free in $P$, the \textit{head} of the link $X$ must also occur free in $Q$.
\end{condition}

\begin{theorem}{The preservation of serial condition in a reduction}{}
  \(P\) where \(Q \mapsto P\) satisfies serial condition if the rule used in (R4) does not have an indirection.
\end{theorem}

\begin{proof}{}{}
  Firstly, we do not have to worry about the local link (the link but with the free link name).
  The local links on LHS should match local links outside of the rule ``completely''.
  And there should be no chance that the link would be unreachable.
  Since we cannot even start traversing.
  The local links on RHS will generate a new local links, which should have satisfied the serial condition.
  Thus it will obviously satisfy the condition after generation by the rule.
  For the free links, we have to need to pay more attention.
  Is there a chance that the head of a free link would disappear?
  Well, this is denied by the condition we added on a rule;
  We are to store the head of the links in the rewriting process.
  If we are allowed to have an indirection in the rule, then it may disappear after rewriting and the process may cannot satisfy the serial condition.
  However, we just have prohibited this and that would not occur.
\end{proof}

\subsubsection*{An example that collapses}
Firstly, we shall introduce an abbreviation rule for the convenience.
\begin{mdframed}
  \(
  \nu X.(p(\ldots, X, \ldots),
  q(\ldots, \overline{X}))
  \)
  where the number of occurrences of the link named \(X\) is 2
  can be written as
  \(
   p(\ldots, q(\ldots), \ldots)
  \)
  .

  For example, \(\nu X.(\overline{R} \mapsto X, append(L_1, L_2, \overline{X}))\) can be abbreviated as \(\overline{R} \mapsto append(L_1, L_2)\)
\end{mdframed}

We have shown that a rule without an indirection will be a safe to rewrite a process.
What about a rule with indirection ?
Sadly, it yields an annoying problem.

\begin{example}{}{}
The append of the nil and the list should just return the latter list.
\[
\overline{R} \mapsto append(nil, L) \vdash \overline{R} \mapsto L
\]
\end{example}

What happens if we apply the rule to the process ?
\[
\overline{R} \mapsto append(nil, R), p(R)
\]

Then it will be like the following. 
\begin{align*}
& \nu R. (\overline{R} \mapsto append(nil, R), p(R), (\ldots \vdash \ldots)) \\
& \equiv_{\mbox{\scriptsize (E7)}}
\nu R. \nu L. (\overline{L} \mapsto R, \overline{R} \mapsto append(nil, L), p(R), (\ldots \vdash \ldots)) \\
& \longrightarrow
\nu R. \nu L. (\overline{L} \mapsto R, \overline{R} \mapsto L, p(R), (\ldots \vdash \ldots)) \\
& \equiv_{\mbox{\scriptsize (E7)}} 
\nu R. (\overline{R} \mapsto R, p(R), (\ldots \vdash \ldots)) \\
& \equiv_{\mbox{\scriptsize (E7)}} 
\nu R. p(R), (\ldots \vdash \ldots)
\end{align*}

Here the link $R$ is now unreachable, which is surely undesirable.

\subsubsection*{The solutions}

As shown in former example, the indirection from a free link to a free link is somewhat dangerous. How should we avoid this ?

There should be at least 2 solutions.

\begin{itemize}
    \item 
    Prohibit the aliasing on the right-hand side of the rule.
    
    Then, for example, the above append-nil-rule can be rewritten with 2 rules
    \begin{align*}
      (\overline{R} \mapsto &\ append(nil, cons(H, T))
      \vdash \overline{R} \mapsto cons(H, T)), \\
      (\overline{R} \mapsto &\  append(nil, nil) \vdash \overline{R} \mapsto nil)
    \end{align*}
    
    This is the simplest solution though it limits the expressiveness power.
    
    \item 
    Check the rules (and the processes) can be well-typed with the capability-typing as following.
    
    Then, for example, the above append-nil-rule will be il-typed.
    \begin{align*}
      & \overline{R} \mapsto append(nil, R), p(R),
      &  \mbox{\textcircled{\scriptsize 1}} \\
      & (\overline{R} \mapsto append(nil, L) \vdash \overline{R} \mapsto L)
      &  \mbox{\textcircled{\scriptsize 2}}
    \end{align*}

    By KCL, this should satisfy
    \begin{align*}
    & - append/3 + append/2 + p/1 = 0 & \mbox{by \textcircled{\scriptsize 1}} \\
    & append/3 =\ \mapsto/1 & \mbox{by \textcircled{\scriptsize 2}} \\
    & append/2 =\ \mapsto/2 & \mbox{by \textcircled{\scriptsize 2}} \\
    \end{align*}

    By Conn, \(\mapsto/1 =\ \mapsto/2\), therefore, \(p/1 = 0\), which is il-typed.
    
\end{itemize}

\subsection{Implementation}

We have implemented a proof of concept of Directed HyperLMNtal.
The Github repository is \verb|https://github.com/sano-jin/vertex|.
Current implementation lacks the capability typing system and restricted so that the head of the link can only appear at the left-hand side of \verb| -> | (\(\mapsto\)).
However, we have implemented programs including G-machine and FL interpreter
and are confident with its usefulness.

Visualizer : 
\begin{center}
  \begin{tabular}{c}
    \begin{minipage}{0.5\hsize}
      \begin{center}
\begin{lstlisting}[language=lmntal]
\X Y Z W U V S T.(
   X -> a(Z, W),
   Y -> a(W),
   Z -> a(U, V),
   W -> a(V),
   U -> a(S, T),
   V -> a(T),
   S -> a(X, Y),
   T -> a(Y)
)
\end{lstlisting}
      \end{center}
    \end{minipage}
    \begin{minipage}{0.5\hsize}
      \begin{center}
        \includegraphics[clip, width=4cm]{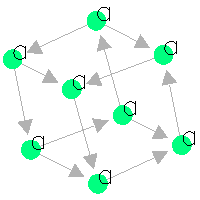}
      \end{center}
    \end{minipage}
  \end{tabular}
\end{center}

Non-deterministic execution and the visualizer for the state transition graph:
\begin{center}
  \begin{tabular}{c}
    \begin{minipage}{0.5\hsize}
      \begin{center}
\begin{lstlisting}[language=lmntal]
% Cycling program
%
% This program never terminates with the `normal execution`
% However, the `non-deterministic execution` will result the finite states
% and it shows the ring forming transitons.

% The initial graph
a.

% Cycing rules.
a :- b.
b :- c.
c :- d.
d :- e.
e :- f.
f :- a.
\end{lstlisting}
      \end{center}
    \end{minipage}
    \begin{minipage}{0.5\hsize}
      \begin{center}
        \includegraphics[clip, width=4cm]{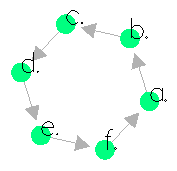}
      \end{center}
    \end{minipage}
  \end{tabular}
\end{center}

\chapter{Summary and Conclusion}

\section{Summary}

The contributions of this research are basically the following 2:
\begin{itemize}
\item
  We formalized the syntax and the semantics of HyperLMNtal
\item
  We implemented G-machine and a compiler for it in HyperLMNtal
\end{itemize}

\subsection*{Formalizing the syntax and the semantics of HyperLMNtal }
HyperLMNtal is extended from the graph rewriting language/calculus model LMNtal.
However, it lacked the rigid definition: it was more an extension of the implementation than on the calculus model.
The semantics of LMNtal features fine-grained concurrency based on local rewriting.
However, since we could not determine the locality of the hyperlink in HyperLMNtal, we couldn't incorporate it into the LMNtal semantics.
Thus, we first introduced a scope (link creation) adopted from the \(\pi\)-calculus\cite{pi} and defined the locality of a hyperlink to formalize the syntax/semantics.
Now, HyperLMNtal is not just a programming language extended from the basic calculation model, but also a concurrent calculation model based on strict and formal definitions.

\subsection*{Implementing G-machine in HyperLMNtal}

In our research, we implemented a compiler which translates the source language, the core language, into the execution code for G-machine and G-machine, in HyperLMNtal.
We have succeeded to implement the compiler in 404 lines and G-machine in 570 lines and showed that we can implement the language processing system that handles complex data structures in graph rewriting language tersely.
In addition, we achieved to visualize G-machine using the HyperLMNtal visualizer\cite{graphene}\cite{lavit}.

\section{Future tasks}

We can possibly prove the correctness of the encodings of \(\lambda\)-calculus, System F and so on in HyperLMNtal with the formalized semantics.
For example, in the encodings of \(\lambda\)-calculus, \(\alpha\)-equivalent expressions should be encoded to the structural congruent processes.
And the \(\beta\)-converted expression should be encoded into those of the reduced process in our new semantics.
We believe that we can give a more formal and simple proof without referring the graph homomorphism.


For the G-machine we have implemented, we can firstly eliminate the duplication and a deletion of a large graph and make it more efficient:
currently we are copying the ``code'' each time when we \texttt{pushGlobal} but this can be eliminated by replacing the ``code'' with a ``pointer'' to the code; instruction pointer.
Then we may compare with the implementation in other languages.
Also, there exists a more efficient G-machines, spineless G-machine\cite{spineless} and spineless tag-less G-machine\cite{stg} (this is the machine used in Glaskaw Haskell Compiler) and a Parallel G-machine \cite{hdg}
so we may like to implement these.
With a new semantics, HyperLMNtal is now a concurrent calculus model with a fine-grained concurrency.
Thus, we believe that it is a great significance to demonstrate how parallel G-machine can be modeled in HyperLMNtal.

\chapter*{Acknowledgments}
\addcontentsline{toc}{chapter}{Acknowledgments}

Throughout the writing of this paper I have received a great deal of support and assistance.

I would first like to thank my supervisor, Professor Ueda.
I could not even start my research without your patient support.

In addition, I would like to thank my parents for the encouragement.

Finally, I would like to thank my colleagues, Masaki Nakata, Kunihiro Hata, Uzawa Seishiro and Kota Fukui and also other members in our lab including Yamamoto-san, Yamada-san, Tsunekawa-san and Walter-san.



\begin{figure}[htbp]

  \begin{flushright}
January 2021, Jin Sano
\end{flushright}
	\end{figure}

\bibliographystyle{plain}
\begin{CJK}{UTF8}{ipxm}
\bibliography{ref}

\begin{thebibliography}{10}

\bibitem{gp2}
Christopher Bak.
\newblock Gp 2: Efficient implementation of a graph programming language.
\newblock September 2015.

\bibitem{spineless}
G.~L. Burn, S.~L. Peyton~Jones, and J.~D. Robson.
\newblock The spineless g-machine.
\newblock In {\em Proceedings of the 1988 ACM Conference on LISP and Functional
  Programming}, LFP '88, page 244–258, New York, NY, USA, 1988. Association
  for Computing Machinery.

\bibitem{implementingfl}
SL~Peyton Jones, DR~Lester, and Simon Peyton~Jones.
\newblock {\em Implementing functional languages: a tutorial}.
\newblock Prentice Hall, 1 1992.

\bibitem{g-machine}
Richard~B. Kieburtz.
\newblock The g-machine: A fast, graph-reduction evaluator.
\newblock In Jean-Pierre Jouannaud, editor, {\em Functional Programming
  Languages and Computer Architecture}, pages 400--413, Berlin, Heidelberg,
  1985. Springer Berlin Heidelberg.

\bibitem{hdg}
Hugh Kingdon, David~R. Lester, and Geoffrey~L. Burn.
\newblock The hdg-machine: A highly distributed graph-reducer for a transputer
  network.
\newblock {\em Comput. J.}, 34(4):290–301, August 1991.

\bibitem{funarg-lisp}
Joel Moses.
\newblock The function of function in lisp or why the funarg problem should be
  called the environment problem.
\newblock {\em SIGSAM Bull.}, (15):13–27, July 1970.

\bibitem{stg}
Simon Peyton~Jones.
\newblock Implementing lazy functional languages on stock hardware: The
  spineless tagless g-machine.
\newblock {\em J. Funct. Program.}, 2:127--202, 04 1992.

\bibitem{imp2gp2}
Detlef Plump.
\newblock From imperative to rule-based graph programs.
\newblock {\em Journal of Logical and Algebraic Methods in Programming}, 88:154
  -- 173, 2017.

\bibitem{handbookgraph}
Grzegorz Rozenberg.
\newblock {\em Handbook of Graph Grammars and Computing by Graph
  Transformation}.
\newblock WORLD SCIENTIFIC, 1997.

\bibitem{pi}
Davide Sangiorgi and David Walker.
\newblock {\em The Pi-Calculus: A Theory of Mobile Processes}.
\newblock Cambridge University Press, USA, 2001.

\bibitem{logiclmntal}
Kazunori Ueda.
\newblock Lmntal as a hierarchical logic programming language.
\newblock {\em Theoretical Computer Science}, 410(46):4784 -- 4800, 2009.
\newblock Abstract Interpretation and Logic Programming: In honor of professor
  Giorgio Levi.

\bibitem{towards}
Kazunori Ueda.
\newblock {\em Towards a Substrate Framework of Computation}, pages 341--366.
\newblock Springer Berlin Heidelber, Berlin, Heidelberg, 2014.

\bibitem{hyperlmntal}
Kazunori Ueda and Seiji Ogawa.
\newblock Hyperlmntal: An extension of a hierarchical graph rewriting model.
\newblock {\em KI - Kunstliche Intelligenz}, 26(1):27--36, 2 2012.

\bibitem{ongraphrewriting}
Ian Zerny.
\newblock On graph rewriting, reduction, and evaluation in the presence of
  cycles.
\newblock {\em Higher Order Symbol. Comput.}, 26(1–4):63–84, December 2013.

\bibitem{lmntalshapetype2016}
吉元佑介.
\newblock
  グラフ書換え言語{LMNtal}におけるグラフ型検査およびルール型保存性検査.
\newblock 修士論文,
  早稲田大学基幹理工学研究科情報理工・情報通信専攻,
  2014.

\bibitem{lmntal2004}
上田 和紀 and 加藤 紀夫.
\newblock 言語モデル{LMNtal}.
\newblock {\em コンピュータ ソフトウェア}, 21(2):126--142, 2004.

\bibitem{cslmntal}
奈良耕太.
\newblock
  再帰的な文脈パターンマッチング機能を持つグラフ書き換え言語の設計と効率的な実装手法.
\newblock 修士論文,
  早稲田大学大学院基幹理工学研究科情報理工・情報通信専攻,
  2015.

\bibitem{lmnstack}
小久保晃.
\newblock {LMNtal} 処理系を利用したコンパイラの設計と実装.
\newblock 卒業論文, 早稲田大学理工学部情報学科, 2014.

\bibitem{lmntal2008}
乾 敦行, 工藤 晋太郎, 原 耕司, 水野 謙, 加藤 紀夫, and 上田
  和紀.
\newblock
  階層グラフ書換えモデルに基づく統合プログラミング言語{LMNtal}.
\newblock {\em コンピュータ ソフトウェア}, 25(1):124--150, 2008.

\bibitem{slim}
石川力, 堀泰祐, 村山敬, 岡部亮, and 上田和紀.
\newblock 軽量な{LMNtal}実行時処理系{SLIM}の設計と実装.
\newblock {\em 情報処理学会第70回全国大会講演論文集}, pages
  153--154, 2008.

\bibitem{graphene}
谷口直輝.
\newblock
  動的階層グラフ可視化ソフトウェア{Graphene}の設計と実装.
\newblock 修士論文,
  早稲田大学大学院基幹理工学研究科情報理工・情報通信専攻,
  2014.

\bibitem{lavit}
綾野 貴之, 堀 泰祐, 岩澤 宏希, 小川 誠司, and 上田 和紀.
\newblock 統合開発環境による{LMNtal}モデル検査.
\newblock {\em コンピュータ ソフトウェア}, 27(4):197--214, 2010.

\end{thebibliography}
\end{CJK}

\appendix
\addcontentsline{toc}{chapter}{Appendices}
\chapter{Source codes}
\label{appendix:lmn-gm}

\section{Implementation of the compiler for G-Machine in HyperLMNtal}
\label{sec:gm-compiler-source}

\begin{lstlisting}[
  language=lmntal,
  basicstyle={\ttfamily\footnotesize}
]
/*
   G-Machine Compiler on LMNtal
*/


{ compiler. 
    {'$callback'(zerostep).

	% Some helper functions for the compiler.

	appendCons @@
	R = append([H|T], L)
	:- R = [H | append(T, L)].

	appendNil @@
	R = append([], L)
	:- R = L.

	% zip with an incresing number N 
	zipNCons @@
	R = zip([V|Vs], N) 
	:- Nplus1 = N + 1
	   | R = [asc(V, N) | zip(Vs, Nplus1) ].

	zipNNil @@
	R = zip([], N) 
	:- int(N)
	   | R = [].

	addNtoRightElemOfAscListCons @@
	R = addNtoRightElemOfAscList([asc(V, Vi) | T], N)
	:- ViPlusN = Vi + N
	   | R = [asc(V, ViPlusN) | addNtoRightElemOfAscList(T, N)].

	addNtoRightElemOfAscListNil @@
	R = addNtoRightElemOfAscList([], N)
	:- int(N)
	   | R = [].

	reverse @@ 
	R = reverse(List) 
	:- R = revAppend(List, []).

	revAppendCons @@
	R = revAppend([H|T], Acc)
	:- R = revAppend(T, [H|Acc]).

	revAppendNil @@
	R = revAppend([], Acc)
	:- R = Acc.

	length_init @@
	R = length(L) 
	:- R = length(L, 0).

	length_cons @@
	R = length([H|L], N)
	:- Nplus1 = N + 1, ground(H)
	   | R = length(L, Nplus1).

	length_nil @@
	R = length([], N)
	:- R = N.

	mapDefVarCons @@
	R = mapDefVar([Var = Def|T])
	:- ground(Def) 
	   | R = [Var | mapDefVar(T)].

	mapDefVarNil @@
	R = mapDefVar([])
	:- R = [].


	% Main compiling functions starts from here

	% Compiles the super combinators.
	% This turns the super combinator expression to the global node
	% with the code (the compiled body of the super combinator) embedded in it.
	% Also, adds the 
	% "<name of the super combinator> -> <the address of the global node>"
	% binding "m(Name, Address)" to the global environment.
	compileSC @@
	compileSC = [ [Name|Vs] = Body | SCs ]
	:- ground(Vs)
	   | compileSC = SCs, 
	     m(Name, !A), 
	     !A = compileR(Body, zip(Vs, 0), length(Vs)).

	compiledSC @@ 
	compileSC = [] :- .

	compileR @@
	R = compileR(Exp, Env, Arity)
	:- int(Arity)
	   | R = nGlobal(Arity, 
			 [ compileE(Exp, Env),
			   update(Arity), 
			   pop(Arity), 
			   unwind
			 ]).


	% The "compileE" compiles the expression and turns it to the code.
	% compileE(<expression to compile>, <environment>, <de bruijn index (based index)>)
	
	% If the expression is an integer, the code should be the "pushInt"
	compileEInt @@
	R = compileE(Int, Env)
	:- int(Int), ground(Env)
	   | R = pushInt(Int).
	
	% let
	compileELet @@
	H = [compileE(let(Defs, E), Env) | T ]
	:- ground(Defs), ground(Env)
	   | H = [ compileLetDefs(Defs, Env), 
 		   compileE(E, 
 			    % new environmemt
 			    append(zip(reverse(mapDefVar(Defs)), 0), 
 				   addNtoRightElemOfAscList(Env, length(Defs)))
 		   ),
 		   slide(length(Defs)) | T ].
	
	
	% letrec
	compileELetRec @@
	H = [compileE(letrec(Defs, E), Env) | T ]
	:- ground(Defs), ground(Env)
	   | H = [ alloc(length(Defs)),
 		   compileLetRecDefs(Defs, 
 				     % new environmemt
 				     append(zip(reverse(mapDefVar(Defs)), 0), 
 					    addNtoRightElemOfAscList(Env, length(Defs))),
 				     length(Defs)
 		   ), 
 		   compileE(E, 
 			    % new environmemt
 			    append(zip(reverse(mapDefVar(Defs)), 0), 
 				   addNtoRightElemOfAscList(Env, length(Defs)))
 		   ),
 		   slide(length(Defs)) | T ].
	
	
	% Compile primitive operators
	compileEAdd @@
	H = [ compileE(E0 + E1, Env) | T ]
	:- ground(Env)
	   | H = [ compileE(E1, Env), 
 		   compileE(E0, addNtoRightElemOfAscList(Env, 1)),
 		   add | T ].
	
	compileEMul @@
	H = [ compileE(E0 * E1, Env) | T ]
	:- ground(Env)
	   | H = [ compileE(E1, Env), 
 		   compileE(E0, addNtoRightElemOfAscList(Env, 1)),
 		   mul | T ].
	
	compileESub @@
	H = [ compileE(E0 - E1, Env) | T ]
	:- ground(Env)
	   | H = [ compileE(E1, Env), 
 		   compileE(E0, addNtoRightElemOfAscList(Env, 1)),
 		   sub | T ].
	
	compileEDiv @@
	H = [ compileE(E0 / E1, Env) | T ]
	:- ground(Env)
	   | H = [ compileE(E1, Env), 
 		   compileE(E0, addNtoRightElemOfAscList(Env, 1)),
 		   div | T ].
	
	compileENegate @@
	H = [ compileE(negate(E0), Env) | T ]
	:- H = [ compileE(E0, Env), 
 		 neg | T ].
	
	compileELessThan @@
	H = [ compileE(E0 < E1, Env) | T ]
	:- ground(Env)
	   | H = [ compileE(E1, Env), 
 		   compileE(E0, addNtoRightElemOfAscList(Env, 1)),
 		   lt | T ].

	compileENeq @@
	H = [ compileE(E0 \= E1, Env) | T ]
	:- ground(Env)
	   | H = [ compileE(E1, Env), 
 		   compileE(E0, addNtoRightElemOfAscList(Env, 1)),
 		   neq | T ].
	
	compileEIf @@
	H = [ compileE(if(E0, E1, E2), Env) | T ]
	:- ground(Env)
	   | H = [ compileE(E0, Env), 
 		   cond([ compileE(E1, Env) ], [ compileE(E2, Env) ]) | T ].
	
	% The default cases.
	% (The default cases are handled with compileC).
	
	compileEVariable @@
	H = [ compileE(Var, Env) | T]
	:- string(Var)
	   | H = [ compileC(Var, Env), eval | T].
	
	compileEApp @@
	H = [ compileE((E0, E1), Env) | T]
	:- H = [ compileC((E0, E1), Env), eval | T].

	compileEConstr @@
	H = [ compileE(eConstr(Tag, Arity), Env) | T]
	:- H = [ compileC(eConstr(Tag, Arity), Env),
		 eval | T].
	
	% Mark 6
	compileECase @@
	H = [ compileE(case(E, Alts), Env) | T ]
	:- ground(Env)
	   | H = [ compileE(E, Env), 
		   casejump(compileD(Alts, Env)) | T ].
	
	% Compile alternatives in the case expression
	% "D" scheme
	compileDCons @@
	R = compileD([Alt | Alts], Env)
	:- ground(Env)
	   | R = [ compileA(Alt, Env) | compileD(Alts, Env) ].

	compileDNil @@
	H = compileD([], Env)
	:- ground(Env)
	   | H = [].

	% "A" scheme
	compileACons @@
	R = compileA(alt(Tag, Vars, Body), Env)
	:- ground(Vars)
	   | R = asc(Tag, 
		     [ split(length(Vars)),
		       compileE(Body, 
				append(zip(Vars, 0), 
 				       addNtoRightElemOfAscList(Env, length(Vars)))

		       ),
		       slide(length(Vars)) ]
	).

	
	% The "compileC" compiles the expression and turns it to the code.
	% compileC(<expression to compile>, <environment>, <de bruijn index (based index)>)
	
	% Looks up and the get the de bruijn index
	% (... based index considering the change of the length of the stack
	% through pushes).
	compileLocalVar_lookup @@
	R = compileC(Var1, [asc(Var2, Vi) | Vs])
	:- string(Var1), Var1 \== Var2, int(Vi)
	   | R = compileC(Var1, Vs).
	
	% If the variable is in the environment, the variable is a local variable and
	% it can be "pushed" with the de bruijn index (based index) from the stack.
	compileLocalVar_resolve @@
	R = compileC(Var1, [asc(Var2, Vi) | Vs])
	:- string(Var1), Var1 == Var2, ground(Vs)
	   | R = push(Vi).
	
	% If the variable is NOT in the environment, the variable is a free variable and
	% it should be "pushed" from the global environment through "pushGlobal".
	compileFreeVar @@
	R = compileC(Var, [])
	:- string(Var)
	   | R = pushGlobal(Var).

	
	% Arithmetic operators
	compileCSAdd @@
	R = compileC(E0 + E1, Env)
	:- R = compileC((("+", E0), E1), Env).
	
	compileCSub @@
	R = compileC(E0 - E1, Env)
	:- R = compileC((("-", E0), E1), Env).
	
	compileCMul @@
	R = compileC(E0 * E1, Env)
	:- R = compileC((("*", E0), E1), Env).
	
	compileCDiv @@
	R = compileC(E0 / E1, Env)
	:- R = compileC((("/", E0), E1), Env).
	
	compileCNeg @@
	R = compileC(negate(E), Env)
	:- R = compileC(("negate", E), Env).
	
	compileCLt @@
	R = compileC(E0 < E1, Env)
	:- R = compileC((("<", E0), E1), Env).
	
	compileCNeq @@
	R = compileC(E0 \= E1, Env)
	:- R = compileC((("\=", E0), E1), Env).

	compileCIf @@
	R = compileC(if(E0, E1, E2), Env)
	:- R = compileC(((("if", E0), E1), E2), Env).	
	
	% If the expression is an integer, the code should be the "pushInt"
	compileCInt @@
	R = compileC(Int, Env)
	:- int(Int), ground(Env)
	   | R = pushInt(Int).
	
	
	% let
	compileCLet @@
	H = [compileC(let(Defs, E), Env) | T ]
	:- ground(Defs), ground(Env)
	   | H = [ compileLetDefs(Defs, Env), 
 		   compileC(E, 
 			    % new environmemt
 			    append(zip(reverse(mapDefVar(Defs)), 0), 
 				   addNtoRightElemOfAscList(Env, length(Defs)))
 		   ),
 		   slide(length(Defs)) | T ].
		
	compileLetDefsCons @@
	H = [compileLetDefs([Var = Def | Defs], Env) | T]
	:- ground(Env), string(Var)
	   | H = [ compileC(Def, Env),
 		   compileLetDefs(Defs, addNtoRightElemOfAscList(Env, 1)) | T].
	
	compileLetDefsNil @@
	H = [compileLetDefs([], Env) | T]
	:- ground(Env)
	   | H = T.
	
	
	% letrec
	compileCLetRec @@
	H = [compileC(letrec(Defs, E), Env) | T ]
	:- ground(Defs), ground(Env)
	   | H = [ alloc(length(Defs)),
 		   compileLetRecDefs(Defs, 
 				     % new environmemt
 				     append(zip(reverse(mapDefVar(Defs)), 0), 
 					    addNtoRightElemOfAscList(Env, length(Defs))),
 				     length(Defs)
 		   ), 
 		   compileC(E, 
 			    % new environmemt
 			    append(zip(reverse(mapDefVar(Defs)), 0), 
 				   addNtoRightElemOfAscList(Env, length(Defs)))
 		   ),
 		   slide(length(Defs)) | T ].
	
	
	compileLetRecDefsCons @@
	H = [compileLetRecDefs([Var = Def | Defs], Env, N) | T]
	:- ground(Env), string(Var), N_1 = N - 1
	   | H = [ compileC(Def, Env), update(N_1),
 		   compileLetRecDefs(Defs, Env, N_1) | T].
	
	compileLetRecDefsNil @@
	H = [compileLetRecDefs([], Env, 0) | T]
	:- ground(Env)
	   | H = T.


	% Mark 6

	compileCNullaryConstr @@
	R = compileC(eConstr(Tag, 0), Env)
	:- ground(Env)
	   | R = pack(Tag, 0).


	% Application is compiled into the postfix-notated code
	% like "(E0, E1) -> [E1, E0, mkap]".
	% However, the de bruijn index (based index) on the E0 should be increased by 1
	% with the consideration of the increasement of the length of the stack.
	compileCApp @@
	H = [ compileCApp((E0, E1), Env) | T ]
	:- ground(Env)
	   | H = [ compileC(E1, Env), 
 		   compileC(E0, addNtoRightElemOfAscList(Env, 1)),
 		   mkap | T ].
	
	compileCAppCheckSaturated @@
	R = compileC((E0, E1), Env)
	:- ground(E0), ground(E1) 
	   | R = checkSaturated( saturatedCons([ makeSpine((E0, E1)) ]), 
				 [ makeSpine((E0, E1)) ], 
				 (E0, E1),
				 Env 
	).

	% Check checkSaturated or not
	R = checkSaturated(X =:= Y, Spine, (E0, E1), Env)
	:- X =:= Y, ground(E0), ground(E1)
	   | R = compileCS(reverse(Spine), Env).

	R = checkSaturated(X =:= Y, Spine, (E0, E1), Env)
	:- X =\= Y, ground(Spine)
	   | R = compileCApp((E0, E1), Env).
	
	R = checkSaturated(notSaturated, Spine, (E0, E1), Env)
	:- ground(Spine)
	   | R = compileCApp((E0, E1), Env).

	
	% CompileCS
	% Compiles the constructor followed by applications
	compileCSConstr @@
	R = compileCS([ eConstr(Tag, Arity) ], Env)
	:- ground(Env)
	   | R = pack(Tag, Arity).

	compileCSSpines @@
	H = [ compileCS([E1, E2 | Es], Env) | T]
	:- ground(Env)
	   | H = [ compileC(E1, Env), 
		   compileCS([E2|Es], addNtoRightElemOfAscList(Env, 1)) | T].


	% Auxiliary functions for compiling "pack" 

	% Make application spines
	isConstrSaturated @@
	R = saturatedCons([eConstr(Tag, Arity) | T])
	:- int(Tag) 
	   | R = (Arity =:= length(T)).
	

	% The default cases (other than constructor) for the "isConstrSaturated"
	isConstrSaturatedDefaultNil @@
	R = saturatedCons([])
	:- R = notSaturated. % This may not happen and not needed.

	isConstrSaturatedDefaultInt @@
	R = saturatedCons([Int | T])
	:- int(Int), ground(T)
	   | R = notSaturated.
	
	isConstrSaturatedDefaultVar @@
	R = saturatedCons([Var | T])
	:- string(Var), ground(T)
	   | R = notSaturated.
	
	isConstrSaturatedDefaultPlus @@
	R = saturatedCons([E1 + E2 | T])
	:- ground(E1), ground(E2), ground(T)
	   | R = notSaturated.
	
	isConstrSaturatedDefaultMinus @@
	R = saturatedCons([E1 - E2 | T])
	:- ground(E1), ground(E2), ground(T)
	   | R = notSaturated.
	
	isConstrSaturatedDefaultTimes @@
	R = saturatedCons([E1 * E2 | T])
	:- ground(E1), ground(E2), ground(T)
	   | R = notSaturated.
	
	isConstrSaturatedDefaultDiv @@
	R = saturatedCons([E1 / E2 | T])
	:- ground(E1), ground(E2), ground(T)
	   | R = notSaturated.
	
	isConstrSaturatedDefaultNeg @@
	R = saturatedCons([negate(E) | T])
	:- ground(E), ground(T)
	   | R = notSaturated.
	
	isConstrSaturatedDefaultLessThan @@
	R = saturatedCons([E1 < E2 | T])
	:- ground(E1), ground(E2), ground(T)
	   | R = notSaturated.
	
	isConstrSaturatedDefaultNeq @@
	R = saturatedCons([E1 \= E2 | T])
	:- ground(E1), ground(E2), ground(T)
	   | R = notSaturated.
	
	isConstrSaturatedDefaultIf @@
	R = saturatedCons([if(E1, E2, E3) | T])
	:- ground(E1), ground(E2), ground(E3), ground(T)
	   | R = notSaturated.
	
	isConstrSaturatedDefaultLet @@
	R = saturatedCons([let(Defns, E) | T])
	:- ground(Defns), ground(E), ground(T)
	   | R = notSaturated.
	
	isConstrSaturatedDefaultLetRec @@
	R = saturatedCons([letrec(Defns, E) | T])
	:- ground(Defns), ground(E), ground(T)
	   | R = notSaturated.
	
	isConstrSaturatedDefaultCase @@
	R = saturatedCons([case(E, Alts) | T])
	:- ground(E), ground(Alts), ground(T)
	   | R = notSaturated.
	
	
	% Make application spines
	makeSpineApp @@
	H = [ makeSpine((E1, E2)) | T] 
	:- H = [ makeSpine(E1), E2 | T].
	
	% The default cases (other than application) for the "makeSpine"
	makeSpineDefaultInt @@
	R = makeSpine(Int) 
	:- int(Int) | R = Int.    
	
	makeSpineDefaultVar @@
	R = makeSpine(Var) 
	:- string(Var) | R = Var.    
	
	makeSpineDefaultCase @@
	R = makeSpine(case(E, Alts)) 
	:- R = case(E, Alts).

	makeSpineDefaultIf @@
	R = makeSpine(if(E1, E2, E3)) 
	:- R = if(E1, E2, E3).
	
	makeSpineDefaultPlus @@
	R = makeSpine(E1 + E2) 
	:- R = E1 + E2.
	
	makeSpineDefaultTimes @@
	R = makeSpine(E1 * E2) 
	:- R = E1 * E2.
	
	makeSpineDefaultMinus @@
	R = makeSpine(E1 - E2) 
	:- R = E1 - E2.
	
	makeSpineDefaultDiv @@
	R = makeSpine(E1 / E2) 
	:- R = E1 / E2.
	
	makeSpineDefaultNeg @@
	R = makeSpine(negate(E)) 
	:- R = negate(E).
	
	makeSpineDefaultLessThan @@
	R = makeSpine(E1 < E2) 
	:- R = (E1 < E2).
	
	makeSpineDefaultNeq @@
	R = makeSpine(E1 \= E2) 
	:- R = (E1 \= E2).
	
	makeSpineDefaultLet @@
	R = makeSpine(let(Defns, E)) 
	:- R = let(Defns, E).

	makeSpineDefaultLetRec @@
	R = makeSpine(letrec(Defns, E)) 
	:- R = letrec(Defns, E).
	
	makeSpineDefaultConstr @@
	R = makeSpine(eConstr(Tag, Arity)) 
	:- R = eConstr(Tag, Arity).

    }.
}.
\end{lstlisting}

\section{Implementation of G-Machine in HyperLMNtal}
\label{sec:gm-source}

\begin{lstlisting}[
  language=lmntal,
  basicstyle={\ttfamily\footnotesize} % \footnotesize
]
/*
   G-Machine on LMNtal
   Mark 6
*/

{ gm. 
    {'$callback'(zerostep).
	% Some helper functions for the compiler.

	appendCons @@
	R = append([H|T], L)
	:- R = [H | append(T, L)].

	appendNil @@
	R = append([], L)
	:- R = L.

	revAppendCons @@
	R = revAppend([H|T], Acc)
	:- R = revAppend(T, [H|Acc]).

	revAppendNil @@
	R = revAppend([], Acc)
	:- R = Acc.

	reverse @@ 
	R = reverse(List) 
	:- R = revAppend(List, []).

	length_init @@
	R = length(L) 
	:- R = length(L, 0).

	length_cons @@
	R = length([H|L], N)
	:- Nplus1 = N + 1, ground(H)
	   | R = length(L, Nplus1).

	length_nil @@
	R = length([], N)
	:- R = N.
    }.


    % G-Machine

    % Looks up the global node and pushes the address to the top of the stack.
    pushGlobal @@
    code = [pushGlobal(F)|Is], stack = S, m(MF, !A)
    :- F == MF
       | code = Is, stack = [!A|S], m(MF, !A).

    % Makes the number node and pushes the address.
    pushInt @@
    code = [pushInt(N)|Is], stack = S
    :- int(N)
       | code = Is, stack = [!A|S], !A = nNum(N).

    % Make application node of the top 2 addresses on the stack.
    mkap @@
    code = [mkap|Is], stack = [!A1, !A2|S]
    :- code = Is, stack = [!A|S], !A = nApp(!A1, !A2).

    % Push the address of the local variable on the Nth element on the stack.
    % Put the stopper "pushing" on the head of the code when pushing.
    code = [push(N)|Is], stack = S
    :- code = [pushing|Is], stack = [lookupStack(N)|S].

    {'$callback'(zerostep).
	% Traverse the stack to accomplish the push of the Nth element on the stack.
	pushLookup @@
	H = [lookupStack(N), A|T]
	:- N > 0, N_1 = N - 1
	   | H = [A, lookupStack(N_1)|T].

	% Push the Nth element on the stack and clear the stopper "pushing".
	pushResolve @@
	code = [pushing|Is], stack = S, H = [lookupStack(0), !A|T]
	:- code = Is, stack = [!A|S], H = [!A|T].
    }.

    % Update the Nth element of the stack.
    update @@
    code = [update(N)|Is], stack = [!A|S]
    :- code = [updating|Is], stack = updateStack(N, !A, S).

    {'$callback'(zerostep).
	updateN @@
	H = updateStack(N, A, [Ai|S])
	:- N > 0, N_1 = N - 1
	   | H = [Ai|updateStack(N_1, A, S)].

	update0nApp @@ 
	H = updateStack(0, !A, [!Ai|S]), !Ai = nApp(!AL, !AR), code = [updating|Is]
	:- H = [!Ai|S], !Ai >< !A, code = Is.

	update0nGlobal @@ 
	H = updateStack(0, !A, [!Ai|S]), !Ai = nGlobal(VN, Body), code = [updating|Is]
	:- ground(Body), int(VN)
	   | H = [!Ai|S], !Ai >< !A, code = Is.

	% Mark 6 (implementation with no confidence)
	update0nConstr @@ 
	H = updateStack(0, !A, [!Ai|S]), !Ai = nConstr(Tag, Args), code = [updating|Is]
	:- int(Tag), ground(Args)
	   | H = [!Ai|S], !Ai >< !A, code = Is.


	% Update the indirection node pointing null.
	% This is for the "let rec" (Mark 3).
	update0nNull @@ 
	H = updateStack(0, !A, [!Null|S]), !Null = null, code = [updating|Is]
	:- H = [!A|S], !Null >< !A, code = Is. 
   }.


    % clear the top N element of the stack.
    {'$callback'(zerostep).
	popN @@
	code = [pop(N)|Is], stack = [!A0|S]
	:- N > 0, N_1 = N - 1
	   | code = [pop(N_1)|Is], stack = S.
    }.

    pop0 @@
    code = [pop(0)|Is]
    :- code = Is.


    % Unwind application node.
    unwindApp @@
    code = [unwind], stack = [!A|S], !A = nApp(!A1, !A2)
    :- code = [unwind], stack = [!A1, !A|S], !A = nApp(!A1, !A2).

    % Unwind the number node.
    unwindInt @@
    code = [unwind], stack = [!A|S], dump = [asc(I2, S2)|D], !A = nNum(N)
    :- int(N), ground(S)
       | code = I2, stack = [!A|S2], dump = D, !A = nNum(N).

    % Unwind the global node and pushes the codes to the code.
    % (3.19)
    unwindGlobalWithEmptyDump @@
    code = [unwind], stack = [!A0|S], dump = [], !A0 = nGlobal(VN, Cs)
    :- ground(Cs), int(VN)
       | code = rearranging(Cs), 
	 stack = rearrange(VN, !A0, S),
	 dump = [],
	 !A0 = nGlobal(VN, Cs).

    % Unwind the global node and pushes the codes to the code.
    % (3.29)
    unwindGlobalWithNonEmptyDump @@
    code = [unwind], stack = [!A0|S], dump = [asc(I2, S2)|D], !A0 = nGlobal(VN, Cs)
    :- ground(S), int(VN)
       | code = unwindingGlobal(length(S), VN),
	 stack = [!A0|S],
	 dump = [asc(I2, S2)|D],
	 !A0 = nGlobal(VN, Cs).

   {'$callback'(zerostep).
	% (3.29)
	unwindGlobalWithNonEmptyDumpKLessThanN @@
	code = unwindingGlobal(K, N), dump = [asc(I2, S2)|D]
	:- K < N
	   | code = unwindingGlobalsVariables(I2, S2), 
	     dump = D.

	unwindingGlobalsVariabes @@
	code = unwindingGlobalsVariables(I2, S2), stack = [!A0, !A1|S]
	:- code = unwindingGlobalsVariables(I2, S2), stack = [!A1|S].

	unwindingGlobalsVariabesFinish @@
	code = unwindingGlobalsVariables(I2, S2), stack = [!Ak]
	:- code = I2, stack = [!Ak|S2].


	% (3.29)
	unwindGlobalWithNonEmptyDumpKGreaterThanN @@
	code = unwindingGlobal(K, N), stack = [!A0|S], !A0 = nGlobal(VN, Cs)
	:- K >= N, ground(Cs), int(VN)
	   | code = rearranging(Cs), 
	     stack = rearrange(VN, !A0, S),
	     !A0 = nGlobal(VN, Cs).
    }.

	% Unwind the constructor node
	% Mark 6
	unwindNConstrWithNonEmptyDump @@
	code = [unwind], stack = [!A|S], dump = [asc(I2, S2)|D], !A = nConstr(Tag, Args)
	:- ground(S)
	   | code = I2, stack = [!A|S2], dump = D, !A = nConstr(Tag, Args).


    % Rearrange the pointers on the stack.
    % Replace the pointer to the application node with the pointer to its right child. 
    {'$callback'(zerostep).
	rearrange @@
	H = rearrange(VN, !Ai_1, [!A|T]), !A = nApp(AL, !AR) 
	:- VN > 0, VN_1 = VN - 1 
	   | H = [!AR|rearrange(VN_1, !A, T)], !A = nApp(AL, !AR) .

    finishRearrange @@
    code = rearranging(CS), H = rearrange(0, !An, T), 
    :- code = CS, H = [!An|T].
    }.

    % Evaluate the arguments of the primitive operators
    eval @@
    code = [eval|Is], stack = [A|S], dump = D
    :- code = [unwind], stack = [A], dump = [asc(Is, S)|D].

    % Arithmetic instructions
    add @@
    code = [add|Is], stack = [!A0, !A1|S], !A0= nNum(N0), !A1 = nNum(N1)
    :- N = N0 + N1
       | code = Is, stack = [!A|S], !A0= nNum(N0), !A1 = nNum(N1), !A = nNum(N).
    
    addSameNode @@
    code = [add|Is], stack = [!A0, !A0|S], !A0= nNum(N0)
    :- N = N0 + N0
       | code = Is, stack = [!A|S], !A0 = nNum(N0), !A = nNum(N).

    sub @@
    code = [sub|Is], stack = [!A0, !A1|S], !A0= nNum(N0), !A1 = nNum(N1)
    :- N = N0 - N1
       | code = Is, stack = [!A|S], !A0= nNum(N0), !A1 = nNum(N1), !A = nNum(N).

    subSameNode @@
    code = [sub|Is], stack = [!A0, !A0|S], !A0= nNum(N0)
    :- N = N0 - N0
       | code = Is, stack = [!A|S], !A0 = nNum(N0), !A = nNum(N).

    mul @@
    code = [mul|Is], stack = [!A0, !A1|S], !A0= nNum(N0), !A1 = nNum(N1)
    :- N = N0 * N1
       | code = Is, stack = [!A|S], !A0= nNum(N0), !A1 = nNum(N1), !A = nNum(N).

    mulSameNode @@
    code = [mul|Is], stack = [!A0, !A0|S], !A0= nNum(N0)
    :- N = N0 * N0
       | code = Is, stack = [!A|S], !A0 = nNum(N0), !A = nNum(N).

    div @@
    code = [div|Is], stack = [!A0, !A1|S], !A0= nNum(N0), !A1 = nNum(N1)
    :- N1 =\= 0,  N = N0 / N1
       | code = Is, stack = [!A|S], !A0= nNum(N0), !A1 = nNum(N1), !A = nNum(N).

    divSameNode @@
    code = [div|Is], stack = [!A0, !A0|S], !A0= nNum(N0)
    :- N0 =\= 0, N = N0 / N0
       | code = Is, stack = [!A|S], !A0 = nNum(N0), !A = nNum(N).

    neg @@
    code = [neg|Is], stack = [!A0|S], !A0 = nNum(N0)
    :- N = - N0
       | code = Is, stack = [!A|S], !A= nNum(N), !A0 = nNum(N0).


    % Comparison instructions
    lessThanTrue @@
    code = [lt|Is], stack = [!A0, !A1|S], !A0= nNum(N0), !A1 = nNum(N1)
    :- N0 < N1
       | code = Is, stack = [!A|S], !A0= nNum(N0), !A1 = nNum(N1), !A = nNum(1).

    lessThanFalse @@
    code = [lt|Is], stack = [!A0, !A1|S], !A0= nNum(N0), !A1 = nNum(N1)
    :- N0 >= N1
       | code = Is, stack = [!A|S], !A0= nNum(N0), !A1 = nNum(N1), !A = nNum(0).

    lessThanFalseSameNode @@
    code = [lt|Is], stack = [!A0, !A0|S], !A0= nNum(N0)
    :- code = Is, stack = [!A|S], !A0= nNum(N0), !A = nNum(0).


    neqTrue @@
    code = [neq|Is], stack = [!A0, !A1|S], !A0= nNum(N0), !A1 = nNum(N1)
    :- N0 =\= N1
       | code = Is, stack = [!A|S], !A0= nNum(N0), !A1 = nNum(N1), !A = nNum(1).

    neqFalse @@
    code = [neq|Is], stack = [!A0, !A1|S], !A0= nNum(N0), !A1 = nNum(N1)
    :- N0 =:= N1
       | code = Is, stack = [!A|S], !A0= nNum(N0), !A1 = nNum(N1), !A = nNum(0).

    neqFalseSameNode @@
    code = [neq|Is], stack = [!A0, !A0|S], !A0= nNum(N0)
    :- code = Is, stack = [!A|S], !A0= nNum(N0), !A = nNum(0).


    % Condition instruction
    condTrue @@
    code = [cond(I1, I2)|Is], stack = [!A|S], !A= nNum(1)
    :- ground(I2)
       | code = append(I1, Is), stack = S, !A = nNum(1).

    condFalse @@
    code = [cond(I1, I2)|Is], stack = [!A|S], !A= nNum(0)
    :- ground(I1)
       | code = append(I2, Is), stack = S, !A = nNum(0).


    % Slides the top of the stack N and discard the slided elements.
    slide @@
    code = [slide(N)|Is]
    :- int(N)
       | code = [sliding(N)|Is].

    {'$callback'(zerostep).
	slideN @@
	code = [sliding(N)|Is], stack = [!A0, !A1|S]
	:- N > 0, N_1 = N - 1
	   | code = [sliding(N_1)|Is], stack = [!A0|S].

	slide0 @@
	code = [sliding(0)|Is]
	:- code = Is.
    }.

    % Alloc 
    % Allocates N "null pointer"s.
    % These "null pointers" will be updated and replaces with a "real" addresses
    % (so there is no need for worrying for the null pointer traversing).
    {'$callback'(zerostep).
	allocN @@
	code = [alloc(N)|Is], stack = S
	:- N > 0, N_1 = N - 1
	   | code = [alloc(N_1)|Is], stack = [!Null|S], !Null = null.
    }.

    alloc0 @@
    code = [alloc(0)|Is], stack = S
    :- code = Is, stack = S.


    % Mark 6
    % Implementation of data types

    % Pack constructor with the arguments
    pack @@
    code = [pack(Tag, Arity) | Is]
    :- code = [packing(Tag, Arity, []) | Is].

    {'$callback'(zerostep).
	packN @@
	R = packing(Tag, Arity, Args), stack = [A|S]
	:- Arity > 0, Arity_1 = Arity - 1
	   | R = packing(Tag, Arity_1, [A|Args]), stack = S.

	pack0 @@
	code = [packing(Tag, 0, Args) | Is], stack = S
	:- code = Is, stack = [!A|S], !A = nConstr(Tag, reverse(Args)).
    }.

    % Casejump 
    % A jump instructions for the case expression
    {'$callback'(zerostep).
	casejumpLookup @@
	code = [casejump([asc(Tag1, I2) | Cases]) | Is], 
	stack = [!A|S], 
	!A = nConstr(Tag2, SS)
	:- Tag1 =\= Tag2, ground(I2)
	   | code = [casejump(Cases) | Is], 
	     stack = [!A|S], 
	     !A = nConstr(Tag2, SS).
    }.

    casejumpResolve @@
    code = [casejump([asc(Tag1, I2) | Cases]) | Is], 
    stack = [!A|S], 
    !A = nConstr(Tag2, SS)
    :- Tag1 =:= Tag2, ground(Cases)
       | code = append(I2, Is), 
	 stack = [!A|S], 
	 !A = nConstr(Tag2, SS).

    % Split
    code = [split(N) | Is], stack = [!A|S], !A = nConstr(Tag, Args)
    :- ground(Args), int(N) % There may be no need for the N
       | code = Is,
	 stack = append(Args, S), 
	 !A = nConstr(Tag, Args).


    % Garbage collector (based on reference counting)
    {'$callback'(zerostep).
	gcNApp @@
	!A = nApp(!E1, !E2) 
	:- num(!A) =:= 1
	   | .
	
	gcNNum @@
	!A = nNum(N) 
	:- num(!A) =:= 1, int(N)
	   | .
	
	gcNConstr @@
	!A = nConstr(Tag, Args) 
	:- num(!A) =:= 1, int(Tag), ground(Args)
	   | .
    }.
}.
\end{lstlisting}

\chapter{Experiments}


\begin{figure}[h]
  \begin{center}
  \begin{tabular}{rrr}
    \hline
    input & result & number of the steps of G-machine \\
    \hline\hline
    1 & 3   & 98    \\
    2 & 5   & 168   \\
    3  & 9   & 308   \\
    4  & 15  & 518   \\
    5  & 25  & 868   \\
    6  & 41  & 1428  \\
    7  & 67  & 2338  \\
    8  & 109 & 3808  \\
    9  & 177 & 6188  \\
    10 & 287 & 10038 \\
    \hline
  \end{tabular}
  \end{center}  
  \caption{The number of the steps took in \texttt{nfib}}
\end{figure}

\begin{center}
\begin{figure}
\begin{tabular}{rrr}
\hline
The index of the prime & the prime & the number of the steps of G-machine \\
\hline\hline
0  & 2   & 79    \\
1  & 3   & 218   \\
2  & 5   & 479   \\
3  & 7   & 788   \\
4  & 11  & 1341  \\
5  & 13  & 1746  \\
6  & 17  & 2395  \\
7  & 19  & 2896  \\
8  & 23  & 3641  \\
9  & 29  & 4678  \\
10 & 31  & 5323  \\
11 & 37  & 6456  \\
12 & 41  & 7393  \\
13 & 43  & 8182  \\
14 & 47  & 9215  \\
15 & 53  & 10588 \\
16 & 59  & 11961 \\
17 & 61  & 12942 \\
18 & 67  & 14411 \\
19 & 71  & 15684 \\
20 & 73  & 16809 \\
21 & 79  & 18470 \\
22 & 83  & 19887 \\
23 & 89  & 21596 \\
24 & 97  & 23645 \\
25 & 101 & 25206 \\
26 & 103 & 26619 \\
27 & 107 & 28276 \\
28 & 109 & 29785 \\
29 & 113 & 31538 \\
30 & 127 & 34655 \\
31 & 131 & 36504 \\
\hline
\end{tabular}
\caption{The number of the steps took in the Sieve of Eratosthenes}
\end{figure}
\end{center}

\begin{center}
  \begin{figure}
  \begin{tabular}{rrr}
    \hline
    Number of the irrelevant rules
    &
    \begin{minipage}{0.3\textwidth}
    \vspace{0.5em}
      execution time (sec)\\
      with the irrelevantRule1
    \vspace{0.2em}
    \end{minipage}
    &
    \begin{minipage}{0.3\textwidth}
      \vspace{0.5em}
      execution time (sec)\\
      with the irrelevantRule2
    \vspace{0.2em}
    \end{minipage}
    \\
    \hline\hline
    0 & 1.13 & 1.13\\
    1 & 2.13 & 1.16\\
    2 & 3.09 & 1.23\\
    3 & 3.94 & 1.26\\
    4 & 4.96 & 1.32\\
    5 & 5.91 & 1.35\\
  \hline
  \end{tabular}
  \caption{The execution time with the patial maching rules}
\end{figure}
\end{center}

\end{document}